\documentclass[11pt,nofootinbib,notitlepage,eqsecnum]{revtex4-1}
\usepackage[english]{babel}
\usepackage{bbm}
\newcommand{\para}[1]{\begin{center}\textit{#1}\end{center}}

\usepackage{subfig}
\usepackage{graphicx,tikz}
\usetikzlibrary{arrows,shapes.geometric,positioning}
\usepgfmodule{shapes}
\usetikzlibrary{decorations.markings,decorations.text,decorations,patterns,calc}
\usepackage{pgfkeys}
\usepackage{amsmath,amssymb,amsthm,bbold,MnSymbol}

\usepackage{hyperref}

\newcommand{\secref}[1]{section \ref{#1}}

\newcommand{\figref}[1]{Fig. \ref{#1}}
\newcommand{\defref}[1]{definition \ref{#1}}
\newcommand{\lemref}[1]{lemma \ref{#1}}

\newtheorem{lem}{Lemma}
\newtheorem{theorem}{Theorem}
\newtheorem{definition}{Definition}
\newtheorem{notation}{Notation}

\def\dif{\mathrm{d}}
\def\pathdif{\mathcal{D}}

\def\Z{\mathbb{Z}}
\def\R{\mathbb{R}}
\def\C{\mathbb{C}}

\newcommand{\group}[1]{\mathrm{#1}}
\newcommand{\algebra}[1]{\mathfrak{#1}}

\def\SU{\mathrm{SU}}

\def\SO{\mathrm{SO}}

\def\Spin{\mathrm{Spin}}

\newcommand{\threej}[6]{\begin{pmatrix}  #1&#3&#5\\#2&#4&#6\end{pmatrix}}
\newcommand{\sixj}[6]{\begin{Bmatrix}  #1&#3&#5\\#2&#4&#6\end{Bmatrix}}

\newcommand{\fusion}[1]{f^{#1}_{{#1}^+{#1}^-}}

\def\hilbert{\mathcal{H}}

\newcommand{\sqi}[1]{\mathrm{L}^2\!\left(#1\right)}

\newcommand{\ket}[1]{|#1\rangle}
\newcommand{\bra}[1]{\langle #1|}
\newcommand{\scal}[1]{\langle #1\rangle}
\newcommand{\escal}{\langle\cdot|\cdot\rangle}

\newcommand{\mrix}[1]{\mathbbm{#1}}
\renewcommand{\exp}[1]{\operatorname{exp}\left(#1\right)}

\newcommand{\tr}{\operatorname{Tr}}

\def\beq{\begin{equation}}
\def\eeq{\end{equation}}
\def\bq{\begin{equation*}}
\def\eq{\end{equation*}}

\def\kin{\text{kin}}
\def\phys{\text{phys}}
\def\inv{\text{inv}}
\def\Inv{\mathrm{Inv}}
\def\time{\mathcal{T}}
\newcommand{\n}{\mathfrak{n}}
\renewcommand{\l}{\mathfrak{l}}

\begin{document}
\title{Linking covariant and canonical LQG II: Spin foam projector}

\author{Antonia Zipfel}
\email{antonia.zipfel@gravity.fau.de}
\affiliation{Institut  f\"ur Quantengravitation, Department Physik
Universit\"at Erlangen\\ Staudtstrasse 7, D-91058 Erlangen, Germany}

\author{Thomas Thiemann}
\email{thomas.thiemann@gravity.fau.de}
\affiliation{Institut  f\"ur Quantengravitation, Department Physik
Universit\"at Erlangen\\ Staudtstrasse 7, D-91058 Erlangen, Germany}

\begin{abstract}
\begin{center}
{\bf Abstract}\\
\end{center}
In a seminal paper, Kaminski, Kisielowski an Lewandowski for the first time 
extended the definition of spin foam models to arbitrary boundary graphs. This is a prerequisite
in order to make contact to the canonical formulation of Loop Quantum Gravity (LQG) 
whose Hilbert space contains all these graphs. This makes it finally possible to investigate 
the question whether any of the presently considered spin foam models yields a rigging map 
for any of the presently defined Hamiltonian constraint operators. 

In the moment, a description of the KKL extension in terms of Group Field Theory (GFT) is out of technical reach because the interaction part of a GFT
Lagrangian dictates the possible valence of a dual graph and so far is 
geared to duals of simplicial triangulations. To get rid of this restriction one would have to allow all possible interaction terms based on certain invariant polynomials of arbitrarily many 
gauge group elements what is currently out of technical control.
Therefore one has to define
the sum over spin foams with given boundary spin networks in an independent fashion
using natural axioms, most importantly a gluing property for 2-complexes. These axioms are motivated by the requirement that spin foam amplitudes should define a rigging map (physical inner product) induced by the Hamiltonian constraint. 
This is achieved by constructing a spin foam operator $\hat{Z}[\kappa]$ based on abstract 2-complexes $\kappa$ (rather than embedded ones) that acts on the gauge invariant kinematical Hilbert space $\hilbert_{0}$ of Loop Quantum Gravity by identifying the spin nets induced on the boundary graph of $\kappa$ with states in $\hilbert_{0}$. 

In the analysis of the resulting object we are able to identify an elementary spin foam transfer matrix $\hat{Z}$ that allows to generate any finite foam as a finite power of the transfer matrix. It transpires that the sum over spin foams $\kappa$, as written, does not define a projector on the physical Hilbert space.
This statement is independent of the concrete spin foam model and Hamiltonian constraint.
However, the transfer matrix potentially contains the necessary ingredient in order to 
construct a proper rigging map in terms of a modified transfer matrix.
\end{abstract}
\maketitle

\newpage

 \tableofcontents

\newpage

\section{Motivation}
To quantize a Field Theory one can either choose a canonical approach, quantize the Hamiltonian and solve the Schr\"odinger equation, or a covariant one, which rests on the path integral description going back to  Feynman's famous PhD thesis \cite{Feynman:1948}. In Loop Quantum Gravity (LQG), a background independent quantization of General Relativity, the canonical formulation \cite{lqgcan3,lqgcan2}  originates from a reformulation of the ADM action \cite{Arnowitt:1959ah} in terms of gauge connections by Ashtekar and Barbero \cite{Ashtekar:1986yd} while the covariant or spin foam model \cite{baez, lqgcov,Rovelli}, was initiated by Reisenberger's and Rovelli's `sum over histories' \cite{ReisenbergerRovelli97}. In both approaches many technical and structural difficulties arise from the constrained nature of GR deeply rooted in the diffeomorphism invariance of the theory. Particularly, the non-polynomial Hamiltonian constraint, although a quantization has been known for a while (see \cite{Thiemann96a}), is challenging and up to today the physical Hilbert space $\hilbert_{\phys}$ cannot be determined satisfactorily. On the other hand, spin foam models suffer from second class constraint which cannot be implemented strongly. The covariant model has matured a lot but the correct treatment of the constraints is still under debate (see e.g. \cite{Alexandrov:2010un}).

Even though both approaches differ significantly it was often emphasized in the past that they should converge to the same theory. Heuristically, the discrete time-evolution of a spin network on a spatial hypersurface, which defines a basis state in the gauge invariant kinematical Hilbert space of canonical LQG, leads to a colored 2-complex that is the main building block of spin foams. Therefore the partition functions defined by the latter can be either understood as propagator between two 3D geometries or as a rigging map, a generalized projector onto $\hilbert_{\phys}$. This paper will especially focus on the latter train of thoughts.

The subsequent analysis will be mainly based on \cite{Engle:2007uq} (EPRL-model) and  \cite{Kaminski:2009fm} (KKL-model). Closely related to these is the FK-approach \cite{FK}. The boundary space of the EPRL/KKL-model can be formally identified with subspaces of $\hilbert_{0}$ which will be used here in order to define a spin foam operator $\hat{Z}[\kappa]$ for the canonical theory. 
Even if the operators $\hat{Z}[\kappa]$ are equipped with appropriate properties so that the sum $\sum_{\kappa} \hat{Z}[\kappa]$ has a chance to define a projector into $\hilbert_{\phys}$, the object we obtain does not provide a rigging map.
This conclusion is independent of the details of a spin foam model or of a Hamiltonian constraint. To prove that a method to split the operator into smaller building blocks is developed. This splitting procedure is also interesting from a purely technical point of view since it gives a better handle on the sum over all complexes $\kappa$ in the EPRL/KKL-partition function. On the positive side, the splitting 
property just mentioned allows to extract a spin foam transfer matrix which, if proper regularized, defines a modified transfer matrix that potentially yields a proper rigging map.

The paper is organized as follows:
The mathematical foundations for later manipulations of graphs and 2-complexes will be laid in section \ref{sec:2-compl}. For the sake of self-containedness section \ref{sec:model} summarizes and compares a generalization of the EPRL-model with the KKL-model. Furthermore, we will review so-called projected spin networks \cite{Alexandrov:2002br,Dupuis:2010jn} which provide a link between canonical and covariant LQG.
In section \ref{ssec:road} a general framework for merging both theories will be developed guided by the concept of rigging maps or group averaging methods for simpler constrained systems. On this basis, a list of properties that the operator $\hat{Z}[\kappa]$ should satisfy will be deduced. In section \ref{ssec:operfoam} a spin foam operator will be proposed that displays all the features worked out before.
Section \ref{sec:main} contains the proof that each operator $\hat{Z}[\kappa]$ can be split into simple blocks $\hat{Z}$ based on 2-complexes which only contain a small number of internal vertices all connected to an initial spin net (see section \ref{ssec:time}). This result can be used to show that the proposed projector is not of the required form (section \ref{ssec:ttproj}). 
However, $\hat{Z}$ may still contain the necessary information in order to construct 
a spin foam model using a modification of $\hat{Z}$ with the properties of a rigging map.
We conclude by summarizing and discussing the results in section \ref{sec:conclusion}.
\section{Foams and graphs}
\label{sec:2-compl}
The first part of this section gives a short review of the kinematical Hilbert space used in LQG focussing on spin network functions and will be followed by an introduction of piecewise linear complexes.
\subsection{Spin networks}
\label{ssec:spinnets}
The kinematical Hilbert space $\hilbert_{\kin}$ of canonical LQG is the space of complex valued, square-integrable functions $\Psi[A]$ of (generalized) connections $A$ on a spatial hypersurface $\Sigma$ embedded in space-time $\mathcal{M}$. A connection on a manifold can be reconstructed from the set of holonomies
\beq
\label{eqn:holo}
h_p(A)=\mathcal{P}\exp{\int_p A}
\eeq
along all (semianalytic) paths\footnote{A path is an equivalence class of curves under reparametrization and retracing.} $p$ where $\mathcal{P}$ denotes path ordering. Likewise, holonomies provide a map from the groupoid of paths into $\SU(2)$. Instead of evaluating a holonomy along a single path one can also use finite systems of path:
\begin{definition}
A \underline{semianalytic graph} $\gamma$ embedded in $\Sigma$ is a finite set of oriented \footnote{Taking the holonomy along a path always implies an orientation of the path.}semianalytic paths (links $\l$) which intersect at most in their endpoints (nodes $\n$).

A graph is called \underline{closed} if every node is the endpoint of at least two links and it is called \underline{connected} if it cannot be written as the disjoint union of two graphs.

In the following, $E(\gamma)$ and $V(\gamma)$ will denote the set of all links and nodes in $\gamma$, respectively.
\end{definition}
The Hilbert space $\hilbert_{\kin}$ is spanned by cylindrical functions 
\beq
\Psi_{\gamma}(A):=\psi(h_{\l_1}(A),\dots,h_{\l_n}(A))
\eeq
where $\psi$ is a function on $\SU(2)^n$ and the scalar product is given by the Ashtekar-Lewandowski measure $\mu_{AL}$ which reduces to the Haar measure $\mu_H(\l)$ of $\SU(2)$ on every link $\l\in\gamma$ (compare with \cite{Ashtekar:1994mh}). More precisely, for a fixed graph $\gamma$ with $n$ links $\hilbert_{\kin,\gamma}$ is isomorphic to $\sqi{\SU(2)^n\!,\mu_H}$. Let $j:=\{j_{\l}\}$ be a labeling of the links by irreducible representations $\hilbert_{j_{\l}}$ of dimension $d_{j_{\l}}:=2j_{\l}+1$ and $m:=\{m_{\l}\}$ and $n:=\{n_{\l}\}$ be magnetic indices associated to the target $t(\l)$ and source $s(\l)$ of $\l\in E(\gamma)$. Since the matrix elements of the Wigner matrices $R^{j_{\l}}(g_{\l})$, $g_{\l}\in\SU(2)$, define an orthogonal basis of $\hilbert_{j_{\l}}$ the functions
\beq
\label{eqn:nongauge}
T_{\gamma,j,m,n}(\{g_{\l}\})=\prod_{\l\in E(\gamma)}\sqrt{d_{j_{\l}}}\;R^{j_{\l}}_{m_{\l} n_{\l}}(g_{\l})~,
\eeq
build an orthonormal basis of $\sqi{\SU(2)^n\!,\mu_H}$. To restore gauge invariance one needs to assign an intertwiner to each node $\n$, that is a group homomorphism $\iota:V_1\to V_2$. At the node $\n$ the space $V_1$ is formed by the tensor product of all irreducible representations $\hilbert_{j_{\l_i}}$ assigned to the outgoing links $\l_i$ at $\n$ and $V_2$ equals the tensor product of all irreducible representations $\hilbert_{j_{\l'_i}}$ assigned to the ingoing links $\l'_i$:
\beq
\iota_{\n}:\hilbert_{j_{\l_1}}\otimes\cdots\otimes\hilbert_{j_{\l_k}}\to\hilbert_{j_{\l'_1}}\otimes\cdots\otimes\hilbert_{j_{\l'_r}}~.
\eeq
The space of all intertwiners, $\iota_{\n}$ constitutes a Hilbert space $\hilbert_{\n,\inv}$ when equipped with a scalar product $(\cdot,\cdot)$
\begin{equation*}
(\tilde{\iota_{\n}},\iota_{\n})
	=(\tilde{\iota}_{\n}^{\dagger}) ^{n_{\l'_1}\cdots n_{\l'_r}}_{\;\;\;m_{\l_1}\cdots m_{\l_k}} \;
	(\iota_{\n})^{m_{\l_1}\cdots m_{\l_k}}_{\;\;\;n_{\l'_1}\cdots n_{\l'_r}}
	:=\tr(\tilde{\iota}^{\dagger}_{\n}\iota_{\n})~.
\end{equation*}	
defined by the natural contraction of magnetic indices $m_{\l_i},n_{\l'_i}$ where $\dagger$ denotes hermitian conjugation. Due to the compatibility conditions of recoupling theory (see appendix \ref{app:HA}) $\hilbert_{\n,\inv}$ is finite dimensional. Equivalently we could define $\iota_{\n}$ to be an invariant tensor
 \beq
\iota_{\n}: \bigotimes_{\l' \text{ incoming}}\hilbert^{\ast}_{j_{\l'}}\otimes\bigotimes_{\l \text{ outgoing}}\hilbert_{j_{\l}}\to\C~.
\eeq 
where $\hilbert_j^{\ast}$ is the contragredient representation. Therefore, we often will identify $\hilbert_{\n,\inv}$ with the space of invariant tensors
\beq 
\label{eqn:intbasic}
\Inv\left(\bigotimes_{\l' \text{ incomming}}\hilbert^{\ast}_{j_{\l'}}\otimes\bigotimes_{\l \text{ outgoing}}\hilbert_{j_{\l}}\right)\;.
\eeq
equipped with the trace as inner product. 

We are now ready to give an explicit definition of the gauge invariant kinematical Hilbert space $\hilbert_{0}$, which will be mostly referred to as kinematical space since non-invariant elements will not be considered. The space is spanned by so called \emph{spin network functions} 
 \begin{align}
\label{eqn:gauginv}
\begin{split}
T_{\gamma,j,\iota}(\{g_{\l}\})
		&:=\prod_{\l\in E(\gamma)}\sqrt{d_{j_{\l}}}\;[R^{j_{\l}}(g_{\l})]^{n_{s(\l)}}\!_{ m_{t(\l)}}\prod_{\n\in V(\gamma)} (\iota_{\n})^{\{m_{t(\l)=\n}\}}\!_{\{n_{s(\l')=\n}\}}\\
		&:=\tr\left[\prod_{\l\in E(\gamma)}\sqrt{d_{j_{\l}}}\;R^{j_{\l}}(g_{\l})\prod_{\n\in V(\gamma)}  \iota_{\n}\right]~.
\end{split}
\end{align}
This function is truly gauge invariant if all magnetic indices are contracted or equivalently if the graph $\gamma$ is \emph{closed}.

An intertwiner depends in general on the ordering\footnote{Different orderings can be related by a change of basis in the intertwiner space.} of the tensor product \eqref{eqn:intbasic} which is why an orientation of the nodes has to be introduced indicating the order of the links.
\begin{definition}
A \underline{spin network} (short: spin net)  $(\gamma,j,\iota)$ consists of an oriented semianalytic graph, a labeling of links by irreps $j:=\{j_{\l}\}$ and an assignment of intertwiners $\iota:=\{\iota_{\n}\}$ to the nodes (see \figref{fig:labelingspinnet}).
\end{definition}
\begin{figure}[t]
 \begin{center}
	\includegraphics{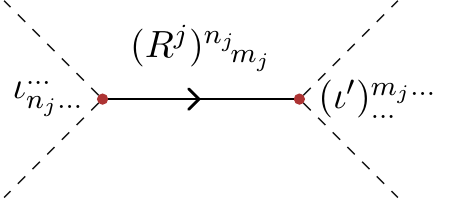}
 \end{center}
\caption{A link of a Spin-network with intertwiner $\iota\in\Inv\left(\cdots\otimes\hilbert_j\otimes\cdots\right)$ associated to the source and $\iota'\in\Inv\left(\cdots\otimes\hilbert^{\ast}_j\otimes\cdots\right)$ associated to the target.}
\label{fig:labelingspinnet}
\end{figure}
In order that $(\gamma,j,\iota)$ labels a linearly independent set of states we require $j_{\l}\neq0$ for all $\l\in E(\gamma)$ and exclude 2-valent nodes whose adjacent links have co-linear tangents\footnote{Not excluded are two-valent intertwiners whose tangents are not co-linear}.
The complex conjugate of a spin network $T_{\gamma,j,\iota}$ can be obtained by reversing the orientation of all links of $\gamma$ since $\overline{R^{j}_{mn}(g)}=R^{j}_{nm}(g^{-1})$.

The \emph{trace} of a spin net is a map  $(\gamma,j,\iota)\to\C$,
\beq
\label{eqn:sntr}
\tr (\gamma,j,\iota):=\tr\left[\prod_{\n} \iota_{\n}\right]
\eeq
defined by contracting the intertwiners. 
\subsection{Complexes}
Since piecewise linear cell complexes are fundamental for the construction of the covariant model, they will be briefly reviewed in the sequel to clarify the notation and set-up the ground for the later considerations. Good introductions to piecewise linear topology can be found for example in \cite{pl1} and \cite{pl2}.
\begin{definition}[\cite{pl1}]\mbox{}
\label{def:complex}
\vspace*{-5pt}
\begin{itemize}
\item A compact \underline{$n$-cell} in $\R^m$, $m\geq n$, is the convex hull of a finite set of affine independent points, called \underline{vertices}, which span an n-dimensional affine subspace.
\item Let $A$, $B$ be compact cells and $P$ be the hyperplane of dimension $m$ spanned by $B$. If $P\cap A=B$ and $P\cap(A\,\backslash\, B)=\emptyset$, then $B$ is an \underline{m-face} of $A$. It is called \emph{proper} if the dimension of $B$ is strictly lower than the dimension of $A$. The set of all proper faces of $A$ is called the frontier $\dot{A}$ of $A$. 
\item An \underline{n-complex} $\mathcal{C}$ is a finite union of compact $m$-cells, $m\leq n$, with at least one compact $n$-cell such that the following two conditions hold:
\begin{enumerate}
\item If $A\in\mathcal{C}$ then all faces of $A$ are in $\mathcal{C}$.
\item If  $A,B\in\mathcal{C}$ then either $A\cap B=\emptyset$ or $A\cap B\in\mathcal{C}$ is a common face of $A$ and $B$.
\end{enumerate}
\item The union of all cells of $\mathcal{C}$ is called the \underline{underlying polyhedron} $\overline{\mathcal{C}}$.
\end{itemize}
\end{definition}
A complex is a collection of all building blocks together with their gluing relations along common faces while the underlying polyhedron is the whole object glued together. If not necessary we will not make this explicit distinction to simplify the notation but it should be kept in mind that these are in principle different objects. For instance $\overline{\mathcal{C}}$ is a topological space while $\mathcal{C}$ itself is just a set. 

A compact n-cell is homeomorphic to an n-ball and the frontier homeomorphic to an $(n-1)$-sphere (for a proof see e.g. \cite{pl1,pl2}). This can be understood by an easy example: Let $f$ be a 2-cell with a vertex $v$ in its interior and an edge $e$ joining $v$ and another vertex of $f$ (see \figref{fig:noface}). If $(e,v)$ would be in the frontier of $f$ then there would exist a straight line $P$ with $P\cap f=e$. But such line would  divide $f$ into two separate faces (figure on the right). Therefore the figure on the left of \figref{fig:noface} is not a convex cell. On the other hand, it is also not a 2-complex since $f\cap e=e$ is not a face of the 2-cell $f$. This is summarized by 
\begin{lem}
\label{lem:v2f}
Every vertex of a 2-cell $f\in\mathcal{C}$ is contained in exactly two 1-cells in the frontier of $f$.
\end{lem}
The reader might be concerned that convexity is to strong if 2-complexes shall describe the time evolution of a spin-network. Indeed, for a semianalytic link $\l$ the `time-evolved' face $\l\times[0,1]$ will certainly not define a 2-cell. Even if the link is approximated by p.l. 1-cells. Nevertheless, it is, of course, possible to approximate $\l\times[0,1]$ by a collection of convex faces which itself defines a 2-complex. Since such an approximation is somewhat arbitrary, the final model should be independent of this. Let us finally remark that the above lemma is still valid if we drop convexity as long as a face has no self-intersections, i.e. is homeomorphic to a 2-ball. The latter will be always assumed! Also all following assertions and theorems can be formulated and proven without using explicitly convexity. It is just convenient to keep it for the moment while it has to be relaxed later on\footnote{A more appropriate choice would be to define the model on ball rather than p.l.- complexes (see section \ref{ssec:absvembed}).}.    
\begin{figure}[t]
 \begin{center}
	\subfloat{\includegraphics{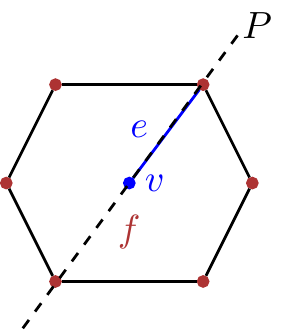}}
	\hspace{2.5 cm}
	\subfloat{\raisebox{.18\height}{\includegraphics{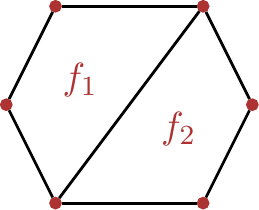}}}
 \end{center}
\caption{A face with a single internal vertex and edge is not a 2-complex. By contrast, the right picture is a 2-complex consisting of two faces glued together along their common 1-cell.}
\label{fig:noface}
\end{figure}
Let us now continue with the description of n-complexes. To efficiently characterize their local properties we introduce the following notations: 
\begin{notation}\nonumber\mbox{}
\begin{itemize}
\item A cell $c$ is called \underline{adjacent} to a different cell $c'$ if $c\cap c'\neq\emptyset$
\item The set of n-cells of $\kappa$ is denoted by $\kappa^{(n)}$.
\item The \underline{interior} $c_{int}$ of $c$ contains all points $p\in c$ which are not contained in any proper face of $c$
\item The \underline{vicinity} $\mathcal{V}(c)$ of a cell $c$ is the set of all cells $b$ for which $c\in\dot{b}$.\footnote{Note, $\mathcal{V}(c)$ is not a complex itself, since the faces of a cell $b\in\mathcal{V}(c)$ are only contained in $\mathcal{V}(c)$ if they are adjacent to $c$. Whereas the frontier of an $n$-cell is an $(n-1)$-complex.}
\item The total number of cells in some set $S$ is denoted by $|S|$
\end{itemize}
\end{notation}
\begin{definition}\mbox{}
\begin{enumerate}
\item A complex $\mathcal{C}$ is called \underline{connected} if its underlying polyhedron is connected. Thus for any two sub-complexes $\mathcal{C}_1$, $\mathcal{C}_2$ such that $\mathcal{C}=\mathcal{C}_1\bigcup\mathcal{C}_2$ there exist at least one cell $c$ satisfying $\dot{c}\cap\mathcal{C}_1\neq\emptyset\neq\dot{c}\cap\mathcal{C}_2$
\item If $\mathcal{C}_1$ and $\mathcal{C}_2$ are two complexes then $\mathcal{C}_1$ is called a \underline{subdivision} of $\mathcal{C}_2$ iff $\overline{\mathcal{C}_1}=\overline{\mathcal{C}_2}$ and every cell of $\mathcal{C}_1$ is a subset of some cell of $\mathcal{C}_2$. A subdivision is called \underline{proper} if $|\mathcal{C}_1|>|\mathcal{C}_2|$.
\end{enumerate}
\end{definition}
For LQG only a special kind of 1- and 2-complexes, graphs and foams, are of interest. An abstract graph is a 1-complex without isolated vertices while a foam is a 2-complex whose boundary graph is closed (see below). For convenience 1-cells are called edges and 2-cells faces. Furthermore, vertices in a graph will be mostly called `nodes' and labeled by $n$ while edges in a graph will be called mostly `links' and labeled by $l$ to distinguish between graphs and 2-complexes. 

A priori we also want to work with complexes without specifying an embedding. Thus, all attributes like orientation and coloring of a complex must be defined in a way independent of the embedding.
\begin{definition}
\label{def:Or1}
\mbox{}
\begin{itemize}
\item The orientation of an edge $e$ determines source $s(e)$ and target $t(e)$ vertex of $e$.
\item Suppose $\dot{f}$ consist of $n$ edges then define a one-to-one map $\mathcal{Z}_f:\{1,\dots,n\}\to\{e|e\in\dot{f}\}$ so that $\mathcal{Z}_f(i)\mapsto e_i$ and $e_i\cap e_{i+1}=v_i$ is a vertex of $\dot{f}$ for all $i<n$ and $e_1\cap e_n=v_0$.
\item The face orientation is the equivalence class of $\mathcal{Z}_f$ under cyclic permutations.
\end{itemize}
\end{definition}
Because $\dot{f}$ constitutes a closed loop (\lemref{lem:v2f}) there exist exactly two inequivalent orientations (cyclic/anticyclic) of a 2-cell $f$. Furthermore, $\mathcal{Z}_f$  induces an edge orientation choosing $s(e_i)=e_{i-1}\cap e_i$ and $t(e_i)=e_{i}\cap e_{i+1}$. This orientation is not unique if the edge is contained in the frontier of more than one face, i.e. the induced orientation of $f$ can be opposite to that of $f'$ on the common edge $e$. In this case the orientation of $f$ is \emph{antidromic} to that of $f'$, otherwise it is \emph{dromic} (see \figref{fig:faceorient}). Due to convexity $f$ and $f'$ intersect at most in one edge so that this definition is consistent. Even in the more general case, when faces are allowed to intersect in more than one edge but the frontiers $\dot{f}$, $\dot{f'}$ are still homeomorphic to $S^1$, the induced orientation on all common edges are either all opposed or all equal.

\begin{figure}[t]
 \begin{center}
	\subfloat{\includegraphics{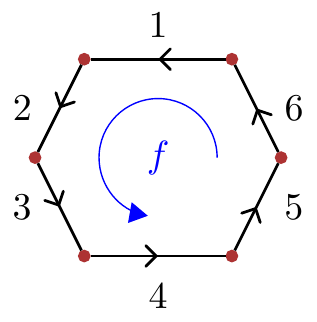}}
	\hspace{2 cm}
	\subfloat{\includegraphics{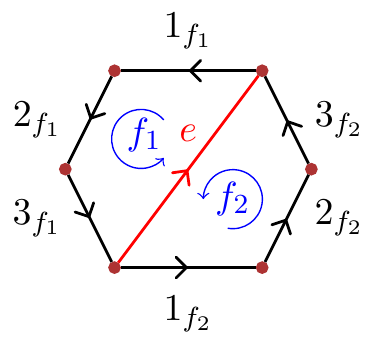}}
 \end{center}
\caption{A 2-cell can be oriented by successively counting the edges $e_i$ in $\dot{f}$. This induces an orientation (black arrows) on $e_i$ with $s(e_1)=e_1\cap e_6$. When the face is subdivided by an edge $e$ (red) then the faces $f_1,f_2$ are oriented such that the induced orientation on $e\in\dot{f}$ is preserved. Thus $f_1$ and $f_2$ are oriented antidromic and here $f_1$ is ingoing to while $f_2$ is outgoing of the red edge.} 
\label{fig:faceorient}
\end{figure}
Independently from the face orientation one can still assign an edge orientation. If the induced orientation of $f$ agrees with this independent orientation then $f$ is \emph{ingoing} otherwise it is called \emph{outgoing} with respect to the given edge. 

Besides the above, the labeling by intertwiners (see below) requires an ordering:
\begin{definition}\mbox{}\\
Let $c$ be an n-cell of the complex $\mathcal{C}$ and $\mathcal{V}^{(n+1)}(c)$ the set of all $(n+1)$ cells in the vicinity of $c$ then the bijection
\beq
\mathcal{Z}_c:\{1,\dots,m=|\mathcal{V}^{(n+1)}(c)|\}\to\mathcal{V}^{(n+1)}(c) 
\eeq
is called an \underline{ordering} of $c$. Two orderings are equivalent if they only differ by cyclic permutations. 
\end{definition}
In contrast to face orientations there exist more than just two inequivalent orderings, for instance a four valent internal edge has six inequivalent orderings.
\subsection{Foams}
\label{ssec:foams}
As mentioned above, not all 2-complexes can be used in LQG. For example, if one adds a single vertex, which is not contained in any edge or face, to a given 2-complex then this is still a well-defined 2-complex but does not give rise to a well-defined spin foam amplitude. To link canonical and covariant LQG we additionally need a method how to associate graphs and 2-complexes. 
 \begin{definition}\mbox{}
 \label{def:foam}
 \begin{itemize}
 \item The interior $\kappa_{int}$ of a 2-complex $\kappa$ is the set of all faces, all edges, which are contained in more than one face, and all vertices contained in more than one internal edge.  
 \item The boundary graph $\partial\kappa$ of a 2-complex $\kappa$ is the set of all edges (links) contained in only one face and vertices (nodes) contained in only one internal edge $e\in\kappa^{(1)}_{int}$. 
\item A graph $\gamma$ is said to \underline{border} $\kappa$ iff there exists a one-to-one (affine) map $c:\gamma\times[0,1]\to \kappa$ mapping each face $\l\times[0,1]$ and each edge $\n\times[0,1]$ of $\gamma\times [0,1]$ to a unique face and a unique internal edge in $\kappa$ respectively.
\item A 2-complex $\kappa$ whose boundary graph $\partial\kappa$ is the disjoint union of connected graphs $\gamma$ bordering $\kappa$ is called a \underline{foam}.
\end{itemize}
\end{definition}
We alert the reader that by definition a graph has no faces.

In the literature the boundary graph of a foam is often defined by either just the combinatorial definition (see e.g. \cite{Rovelli}) or just by bordering graphs (see appendix of \cite{lqgcov}). Neither of this is sufficient since for example $\partial\kappa$ is in general not a well-defined graph. Particularly, if the intersection point $\n$ of two or more boundary links is contained in several internal edges then $\n\notin\partial\kappa$ and consequently $\partial\kappa$ is not even a 1-complex.  On the other hand, a graph bordering $\kappa$ does not have to be closed.
\begin{lem}
\label{lem:boundary}
Let $\kappa$ be a foam then the boundary graph is the disjoint union of closed connected graphs. A face $f\in\kappa$ intersects a connected graph $\gamma\subset\partial\kappa$ at most in one link $l_f$.
\end{lem}
\begin{proof} 
Suppose $\gamma\in\partial\kappa$ is not closed then there is at least one node $\n$ adjacent to one and only one link $l$ in the boundary graph. Since $\gamma$ is bordering $\kappa$, $\n$ is also an endpoint of an internal edge $e_{\n}$. But $e_{\n}$ is contained in only one face, namely the face generated by $[0,1]\times \l$ and consequently $e_{\n}\in\partial\kappa$. $\lightning$.

Since whenever a connected graph $\gamma\in\partial\kappa$ borders $\kappa$ there exists a one-to-one affine map $\gamma\times[0,1] \to\kappa$, this  implies that a face $f$ cannot intersect $\gamma$ in more than one link. 
\end{proof}
Note, \lemref{lem:boundary} does not exclude faces intersecting the boundary graph in several disconnected graphs $\gamma,\gamma'\in\partial\kappa$, $\gamma\cap\gamma'=\emptyset$.
\begin{lem}
\label{lem:intvert}
Let $v$ be an internal vertex of the foam $\kappa$ then all edges $e\in\mathcal{V}(v)$ are internal.
\end{lem}
\begin{proof}
Suppose $e\in\mathcal{V}(v)$ is an element of $\partial\kappa$ but since $v\notin\partial\kappa$ then $\partial\kappa$ is not a graph.~$\lightning$
\end{proof}
Lemma \ref{lem:boundary} and \lemref{lem:intvert} also show that every face has at least two internal edges.
\begin{definition}
\label{def:vertexgr}\mbox{}\\
Subdivide all edges $e$ adjacent to an internal vertex $v\in\kappa$ by a vertex $m(e)$ in the interior of $e$ and all faces $f\in\mathcal{V}(v)$ by an edge $e(f)$ with endpoints $m(e)$ and $m(e')$ whenever $e,e'\in\dot{f}$ and $e,e'\in\mathcal{V}(v)$. This yields a 1-complex $\gamma_v=\{m(e),e(f)|e,f\in\mathcal{V}(v)\}$ called vertex boundary graph.
\end{definition}
Since every $e\in\mathcal{V}(v)$ is contained in at least two faces (\lemref{lem:intvert}) and because of \lemref{lem:v2f},  $\gamma_v$ is the disjoint union of \emph{closed} connected graphs! 
\begin{figure}[h]
 \begin{center}
	\includegraphics{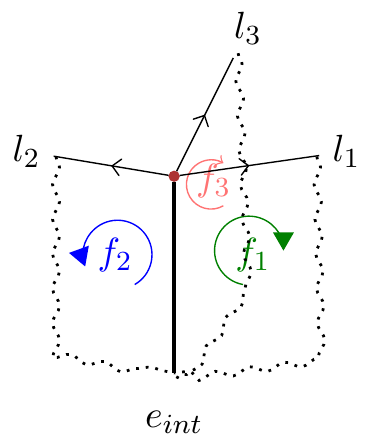}
 \end{center}
\caption{The face orientation induces an orientation on the boundary links $l_i$ while the ordering on the internal edge $\mathcal{Z}_{e_{int}}:\{1,2,3\}\to f_1,f_2,f_3$ induces the ordering on the boundary node (red) such that $l_i$ is the unique link contained in $f_i$.} 
\label{fig:edgeordering}
\end{figure}
\begin{definition}
\label{def:orfoam}
An \underline{oriented foam} is a foam $\kappa$ whose edges and faces are oriented such that all faces $f$ touching the boundary graph $\partial\kappa$  
are ingoing to $\l_f=f\cap\partial\kappa$.
Furthermore, all internal edges $e$ carry an ordering $\mathcal{Z}_{e}$ which induces an ordering $\mathcal{Z}_{\n}$ on the boundary nodes $\n$ by $\mathcal{Z}_{\n}(\l_f)=\mathcal{Z}_{e_{\n}}(f_{\l})$ where $e_{\n}$ is the unique internal edge with $\n\in\dot{e}_{\n} $ and $f_{\l}$ is the unique face containing the wedge spanned by $e_{\n}$ and the boundary link $\l_f$ (see \figref{fig:edgeordering}) \footnote{An ordering of internal vertices is not necessary.}.

Since $\partial\kappa$ borders $\kappa$, internal edges $e$ intersecting the boundary graph in a connected graph $\gamma$ are either all in- or all outgoing of $\gamma$ corresponding to the embedding $\gamma\times[0,1]$ respectively $\gamma\times[-1,0]$. If all internal edges are outgoing of $\gamma$ it is called \underline{initial} and otherwise \underline{final}.  
\end{definition}
Below, subdivisions of oriented foams play a major role for example in order to construct vertex graphs or to analyze equivalence classes within the model. A subdivision of a foam should again yield a well defined foam, e.g. it is not allowed to split a boundary link without splitting the ingoing face as well. Moreover, the orientation of $\kappa$ should be preserved: Suppose we split an edge $e\in\kappa$ by a vertex $v_0$, then the new edges $e_1\cup e_2=e$ obey $s(e_1)=s(e)$, $t(e_1)=v_0=s(e_2)$ and $t(e_2)=t(e)$, if $e$ is internal then $e_1$,$e_2$ inherit the order $\mathcal{Z}_e$ of $e$. If $e\in\partial\kappa$ then $v_0$ is adjacent to only two boundary links and the order is unique.

Let $v,v'\in\dot{f}$ be two vertices such that linking $v$ and $v'$ by an edge $e_0$ in $f$ yields two new faces $f_1\cup f_2=f$. The new faces inherit the orientation of $f$ so that the induced orientation on all old edges is preserved while on $e_0$ the orientations of $f_1$ and $f_2$ are antidromic. Therefore, the direction of $e_0$ can be chosen freely (see \figref{fig:faceorient}).

For example, the induced as well as the free edge orientation on the half-edges $e_{m(e)},e'_{m(e')}$ connecting $v\in\kappa_{int}$ and vertices $m(e),m(e')$ of a vertex graph $\gamma_v$ (see \defref{def:vertexgr}) is preserved whereas $e(f)$ in $\gamma_v$ is oriented such that the wedge spanned by $e_{m(e)},e'_{m(e')}$ is outgoing (see \figref{fig:vertexboundary}). 

\begin{figure}[t]
 \begin{center}
	\subfloat{\includegraphics{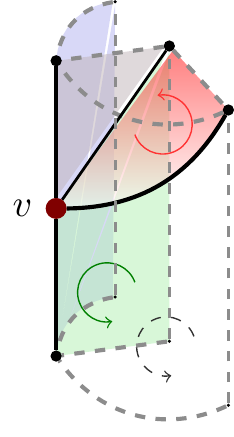}}
	\hspace{1.5 cm}
	\subfloat{\raisebox{1.2\height}{\Huge{$\gamma_v=$}\hspace{.25 cm}{\raisebox{-.4\height}{\includegraphics{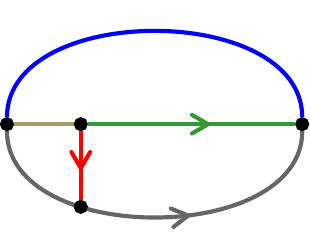}}}}}
 \end{center}
\caption{The vertex boundary graph can be constructed by cutting out the vertex $v$ along the faces adjacent to $v$. Here, the dotted lines indicate the splitting edges $e(f)$ and the bold black points the vertices $m(e)$. The edges $e(f)$ are oriented, such that the corresponding wedge is outgoing. } 
\label{fig:vertexboundary}
\end{figure}
Another important example is gluing of (non-oriented) foams along common closed components of their boundary graphs: Suppose $\gamma_1\in\partial\kappa_1$ is isomorphic to $\gamma_2\in\partial\kappa_2$ then a new foam $\kappa'$ can be constructed by identifying $\gamma_1=\gamma_2=\gamma$ defining a subdivision of $\kappa_1\sharp\kappa_2$ where $\gamma$ is removed. The same can be done for oriented complexes if their orientations match so that $\kappa'$ is an oriented subdivision of $\kappa_1\sharp\kappa_2$. Consequently, the orientations of faces glued together must be antidromic and if the internal edge $e\in\kappa_1$ is ingoing to $\n\in\gamma$ then the corresponding edge $e'\in\kappa_2$ must be outgoing of $\n$ (see \figref{fig:gluing}).\footnote{Since boundary links inherit the orientation of the faces intersecting $\partial\kappa$ this implies that the orientation of $\gamma_1$ in $\kappa_1$ is opposite to that of $\gamma_2$ in $\kappa_2$ and strictly speaking they are not isomorphic. But since in a subdivision the orientation of splitting edges is not determined the gluing is still well-defined when assuming that $\gamma$ is not oriented.} 

\begin{figure}[t]
 \begin{center}
	\subfloat{\includegraphics{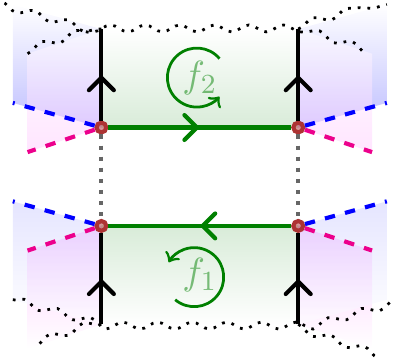}}
	\hspace{2 cm}
	\subfloat{\raisebox{.5 cm}{\includegraphics{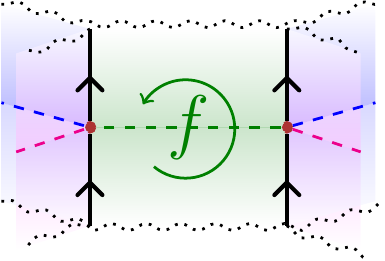}}}
 \end{center}
\caption{Two foams can be glued together along a common closed component of their boundary graphs if the face orientation (green arrows) and internal edge orientations (black arrows) match.} 
\label{fig:gluing}
\end{figure}

\subsection{Spin foams}
\label{ssec:spinf}
Similar to the coloring of graphs in \secref{ssec:spinnets}, foams will be labeled by representation data of a gauge group $G$. In LQG we are especially interested in the cases $G=\SO(3,1)$ respectively $G=\SO(4)$. Since $\SO(4)$ is a compact semisimple Lie group the representation theory is comparably easy and therefore we will focus on the latter.

A spin foam $(\kappa,\{\hilbert_f\},\{Q_e\})$ consists of an oriented foam $\kappa$ and an assignment of a Hilbert space $f\to\hilbert_f$ (irreducible representation space of $\SO(4)$) to every face $f\in\kappa$. This induces a Hilbert space\footnote{The total order of the tensor product in \eqref{eqn:edgespace} is determined by the edge order.} on every edge
\beq
\label{eqn:edgespace}
e\;\mapsto\hilbert_e:=\!\!\bigotimes_{f' \text{ ingoing to } e} \hilbert_{f'}\;\;\otimes\!\!\bigotimes_{f \text { outgoing of } e} \hilbert^{\ast}_{f}\;.
\eeq
and an invariant subspace $\hilbert_{e,Inv}$ spanned by intertwiners $\iota:\hilbert_{e}\to\C$. To each internal edge we associate an operator $Q_e:\hilbert_{e,Inv}\to\hilbert_{e,Inv}$ in such a way that the domain of $Q_e$ is associated to the source of $e$ and the image of $Q_e$ is associated to the target of $e$.

Let $n_i$ ($n_o$) be the total number of faces ingoing to (outgoing from) the edge $e$ and $(x_1,\dots,x_{d_f})$ be a basis of of $\hilbert_f$ with $d_f:=\mathrm{dim}\hilbert_f$ then  $\iota_e\in\hilbert_{e,\inv}$ is a tensor of rank $(n_i,n_o)$ 
\beq
\label{eqn:indexnot0}
\iota_e=\left(\iota_e\right)^{\qquad A_{f'_1},\dots,A_{f'_{n_i}}}_{A_{f_1},\dots,A_{f_{n_o}}}\; \;x^{A_{f_1}}\otimes\cdots\otimes x^{A_{f_{n_o}}}\otimes x_{A_{f'_1}}\otimes\cdots\otimes x_{A_{f'_{n_i}}}~.
\eeq
The expansion of $Q_e$ in a basis $\{\iota_e\}$ of $\hilbert_{e,\inv}$ reads 
\beq
\label{eqn:indexnot}
Q_e:=(Q_e)^{\;\;\;\iota_{t(e)}}_{\iota_{s(e)}}\;\;\iota_{s(e)}^{\dagger}\otimes\iota_{t(e)}~.
\eeq 
and, following the above, the dual $\iota_{s(e)}^{\dagger}$ is attached to the source and $\iota_{t(e)}$ to the target of $e$.

\begin{figure}[t]
 \begin{center}
 	\hspace{-2 cm}
	\subfloat{\includegraphics{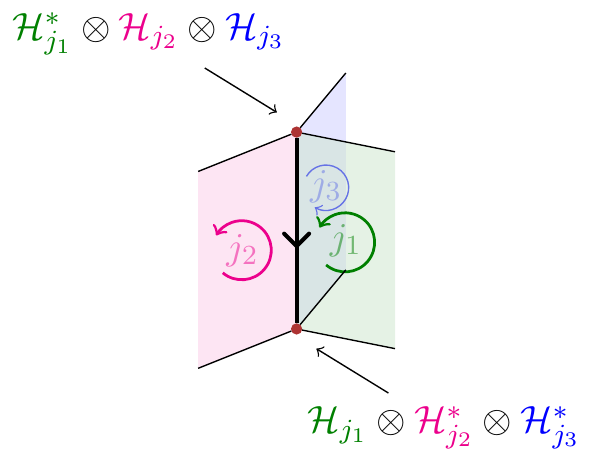}}
	\hspace{-2 cm}
	\subfloat{\includegraphics{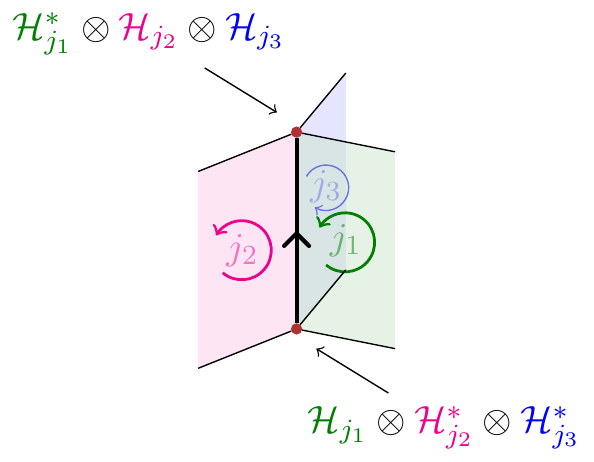}}
	\subfloat{\raisebox{.6cm}{\includegraphics{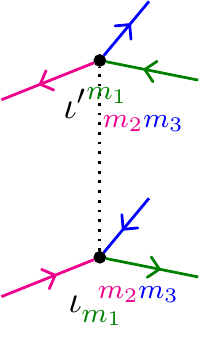}}}
 \end{center}
\caption{The intertwiner space associated to target/source of an internal edge is independent of the edge orientation and depends only on the face orientations since $\Inv(\hilbert_{j_1}^{\ast}\otimes\hilbert_{j_2}\otimes\hilbert_{j_3} )^{\ast}=\Inv(\hilbert_{j_1}\otimes\hilbert_{j_2}^{\ast}\otimes\hilbert_{j_3}^{\ast} )$ (compare the two figures on the left). The convention for assigning the Hilbert spaces and orientations of boundary and vertex graphs is chosen such that it agrees with the convention for spin nets. Compare the figures on the right with \figref{fig:vertexboundary} and \figref{fig:labelingspinnet}}
\label{fig:internaledge}
\end{figure}
The marking $(\hilbert_f,Q_e)$ of the bulk $\kappa_{int}$ induces a spin net structure on $\partial\kappa$: A boundary link $\l_f$ contained in the unique face $f$ is labeled by $\hilbert_f$ and a node $\n_e\in\partial\kappa $ is labeled by $\iota_e$, if the internal edge $e$ adjacent to $\n_e$ is ingoing, and by the dual intertwiner if $e$ is outgoing. By \lemref{lem:boundary} each boundary link $\l_f\in\dot{f}$ in $\kappa$ is adjacent to exactly two internal edges $e,e'$ which are either both ingoing to or both outgoing of $\partial\kappa$ and therefore, if $f$ is ingoing to $e$ it is outgoing of $e'$. In both cases $\hilbert_f$ is associated to $t(\l_f)$ while the dual is associated to the source (see \figref{fig:internaledge}). In fact, whether the dual or the original Hilbert space is associated to a node only depends on the face orientation and the whole model can be formulated without specifying edge orientations (see \cite{Kisielowski:2011vu}). However, in the subsequent discussion it is more convenient to keep all orientations as defined above.

Similarly, the coloring of $\kappa$ induces a spin net on vertex boundary graphs $\gamma_v$,
see definition \ref{def:vertexgr}.
This yields a natural contraction of the intertwiners by
\beq
\label{eqn:verttr}
\mathcal{A}_v(\{\iota_{e_v}\})
=\tr\left[\prod_{e\in\mathcal{V}(v)} \iota_{e_v}\right]
\eeq
where $e_v:=e_{v,e}$ is the half-edge of $e$ adjacent to $v$ and we assumed that all edges $e\in\mathcal{V}(v)$ are incoming to $v$. Note, that all intertwiners which are not assigned to boundary nodes can be  contracted in this way defining the spin foam trace
\beq
\label{eqn:sftr}
\tr(\kappa,\hilbert_f,Q_e):=\left[\prod_{e\in\kappa^{(1)}_{int}} \sum_{\iota_e}(Q_e)^{\;\;\iota_{e_{t(e)}}}_{\,\iota'_{e_{s(e)}}}\!\prod_{v\in\kappa^{(0)}_{int}} \mathcal{A}_v(\{\iota_{e_v}\})\right]\;\;\bigotimes_{\n\in(\partial\kappa)^{(0)}}\iota_{e_{\n}}~.
\eeq
To simplify the notation we did not display whether $\iota$ is a dual intertwiner or not and we will continue to do so if not explicitly necessary. When, in addition, group elements $g_{\l}\in G$ are attached to all boundary links $\l$ then one obtains the \emph{spin foam partition function} 
\beq
\label{eqn:abstractpart}
\begin{split}
Z[\kappa](\{g_{\l}\})
:=\!\!\sum_{\{\iota_{e_v}\},\{\rho_f\}}& \left[\prod_{e\in\kappa_{int}} (Q_e)^{\;\;\iota_{e_{t(e)}}}_{\,\iota'_{e_{s(e)}}}\!\prod_{v\in\kappa_{int}} \mathcal{A}_v(\{\iota_{e_v}\})\right]\\
&\times \underbrace{\tr\left[ \prod_{\l_f\in(\partial\kappa)^{(1)}} R^{j_{f_{\l}}}(g_{\l_f})\prod_{\n_e\in(\partial\kappa)^{(0)}} \iota_{e_{\n}}\right]}_{T_{\partial\kappa,\iota_{\n},j_{\l}}(\{g_{\l_f}\})\prod\limits_{\l_f\in\partial\kappa}\frac{1}{\sqrt{d_{\rho_{\l_f}}}}}~.
\end{split}
\eeq

Notice that no claim about convergence of \eqref{eqn:abstractpart} is made at this point
for generic  $\kappa$ which therefore may only define a `distributional' linear 
functional on the boundary space $\hilbert_{\partial\kappa}$ spanned by spin nets based on $\partial\kappa$. To fix one's intuition,  consider the following easy but important example:
\begin{definition}
\label{def:trivevol}
The \underline{trivial evolution} $\kappa_0$ is an oriented foam which has no internal vertices and whose boundary graph $\partial\kappa$ is the disjoint union of two graphs $\gamma_1$ and $\gamma_2$ such that there exist a (non-oriented) isomorphism $\gamma_1\simeq\gamma\simeq\gamma_2$.
\end{definition}
Since by definition the boundary links of $\partial\kappa_0$ inherit the orientation of the face in
which they are contained and since for every face $f$ there are two links $\l_f\in\gamma_1$, $\l'_f\in\gamma_2$ and $\l_f,\l'_f\in\dot{f}$ it follows that  the orientation of $\gamma_1$ is opposite to the one of $\gamma_2$. Moreover, each internal edge $e$ is adjacent to two nodes in the boundary graph, w.l.o.g. fix $s(e)\in\gamma_1$ and $t(e)\in\gamma_2$, so that the spin net on $\gamma_1$ is dual to the one induced on $\gamma_2$. Concluding,
\beq
\begin{split}
Z[\kappa_0](\{g_{\l}\})
= &\sum_{\{\iota_v\},\{\rho_f\}}\left[\prod_f\frac{1}{d_{\rho_f}}\prod_{e\in\kappa_{0,int}} (Q_e)^{\;\;\iota_{t(e)}}_{\;\iota_{s(e)}}\right]\\
&T_{\gamma_2,\iota_{t(e)},j_{\l'_f}}(\{g_{\l'_f}\})\;\otimes\;(T_{\gamma_1,\iota_{s(e)},j_{\l_f}}(\{g_{\l_f}\}))^{\dagger}~.
\end{split}
\eeq
The partition function \eqref{eqn:abstractpart} is invariant if one adds or removes faces labeled by the trivial  representation. Later on we will also include additional face amplitudes such that $Z[\kappa]$ is also invariant under colored subdivisions defined in the following
\begin{definition}
\label{def:coloredsubdivision}
A colored subdivision of a spin foam $(\kappa,\hilbert_f,Q_e)$ is an oriented subdivision of $\kappa$ such that for the new colored foam $(\kappa',\hilbert_f',Q_e')$ holds
\begin{enumerate}
\item $\hilbert_f=\hilbert_{f'_1}=\cdots=\hilbert_{f'_n}$ if $f\in\kappa$; $f'_,\dots,f'_n\in\kappa'$ and $f'_1\cup\cdots\cup f'_n=f$
\item if $e'\notin\kappa$ then $Q_{e'}=\mathbb{1}$ and $\iota_{e'}$ is a two-valent intertwiner
\item  if $e'_1,\dots,e'_n\in\kappa'_{int}$ such that $e'_1\cup\cdots\cup e'_n=e\in\kappa_{int}$ then $Q_{e'_1}\circ\cdots\circ Q_{e'_n}=Q_e$.
\end{enumerate}
\end{definition}
Two spin foams $(\kappa_1,\{\hilbert_f\},\{Q_e\})$ and 
$(\kappa_2,\{\hilbert_{f'}\},\{Q_{e'}\})$ can be glued together along a common graph $\gamma\in\partial\kappa_{1/2}$ if the orientation matches and the induced spin network functions on $\gamma$ are mutually conjugated. Then $\hilbert_{f_{\l}\sharp f'_{\l}}=\hilbert_{f_{\l}}\equiv\hilbert_{f'_{\l}}$, where $\l\in\gamma$ is contained in $f_{\l}\in\kappa_1$ and  $f'_{\l}\in\kappa_2$ respectively,
and $Q_{e\sharp e'}=Q_e\circ Q_{e'}$ where $e\in \kappa,e'  \in \kappa'$.

\subsection{Triangulations and foams}
\label{sec:trian}
Before we conclude the mathematical part and give a physical motivation for the above model we will briefly discuss triangulations of 4-manifolds and relations to foams as defined in \defref{def:foam}. 
One of the main ingredients of covariant LQG is the truncations of degrees of freedom by introducing a triangulation of space-time. A triangulation of a smooth compact n-manifold $\mathcal{M}$ is a triple $(\mathcal{M},\Delta,f)$ where $\Delta$ is a (simplicial) complex and $f:\overline{\Delta}\to\mathcal{M}$ a piecewise differential homeomorphism (see appendix \ref{app:triangulation} for details and an extension to non-compact manifolds). In 1940 Whitehead \cite{whitehead:1940} proved\footnote{Originally Whitehead proved the assertion in the $C^1$ category but already extended it to $C^k$-triangulations. To ensure uniqueness up to p.l. homeomorphisms and ensure that $\overline{\Delta}$ is a p.l. manifold the embedding map $f$ must be sufficiently smooth, i.e. $C^1$ is not enough (for a counter example see \cite{cannon1979}).} that any smooth manifold $\mathcal{M}$ has an essentially unique triangulation up to p.l. homeomorphisms. Moreover, the underlying polyhedron $\overline{\Delta}$ is a p.l. manifold which means that any point in the interior of $\overline{\Delta}$ has a neighborhood which is p.l. homeomorphic to an $n$-simplex. Thus, any $(n-1)$-cell in the interior of $\overline{\Delta}$ is a proper face of two $n$-cells and the set $\partial\Delta$ of all $(n-1)$-cells contained in only one $n$-cell induces a proper triangulation of the boundary of $\mathcal{M}$  whereupon $\partial^2\mathcal{M}=\emptyset$ implies that any lower dimensional ($\leq n-2$) cell must be contained in at least two higher dimensional cells.

Let $\Delta:=\{A^{(i)}_j|i=0,\dots n;j=1\dots q_i\}$ be a (simplicial) n-complex triangulating $\overline{\Delta}$ where $i$ labels the dimension of the cell and $q_i$ is the number of $i-$cells. Let $a^{(i)}_j$ denote the barycenter of $A^{(i)}_j$. The barycenters of n-cells define the dual vertices. The one-cell dual $\ast[A^{(n-1)}_k]$ to $A^{(n-1)}_k=A^{(n)}_i\cap A^{(n)}_j$ is the union of the edge joining $a^{(n)}_i$ and $a^{(n-1)}_k$ and the edge joining $a^{(n-1)}_k$ and $a^{(n)}_j$.  Inductively, the dual cell of $A^{m}_i$ is defined to be the $(n-m)$-dimensional subset of all points $x$ for which there exist $\lambda, \mu>0$, $\lambda+\mu=1$, such that $x=\lambda\, a^{(m)}_i+\mu\, b$ where $b$ is a point in some $\ast[A^{(m+1)}_j]$ dual to a cell $A^{(m+1)}_j$ in the vicinity of $A^{(m)}_i$ (see \figref{fig:dual}). The set of all dual cells is the \emph{dual complex} $\ast\Delta$ of $\Delta$.

\begin{figure}[t]
 \begin{center}
	\subfloat{\includegraphics{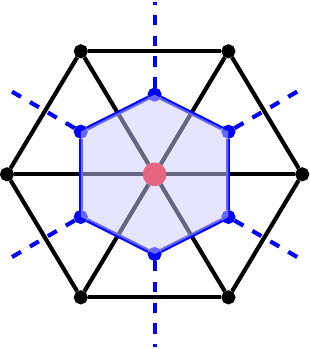}}
 \end{center}
\caption{In two dimensions the complex dual to the black triangulation can be constructed by joining the barycenter (blue vertices) of the triangles where the blue face is dual to the red vertex.}
\label{fig:dual}
\end{figure}
In general $a^{(n)}_i,\,a^{(n-1)}_k$ and $a^{(n)}_j$ are not collinear and thus dual cells are not convex but compact polyhedra.
\begin{lem}[\cite{pl1}]
\label{lem:dual}
If $\overline{\Delta}$ is a p.l. n-manifold and $A\in\Delta$ an m-cell then $\ast A$ is a p.l. $(n-m)$-ball (or equivalently: p.l. homeomorphic to an $(n-m)$-simplex). If $A\in\partial\Delta$ then the cell $\sharp A$ dual to $A$ in the subcomplex $\partial\Delta$ is an $(n-m-1)$-ball in the frontier of $\ast A$.
\end{lem}   
Obviously, $\ast\Delta$ is generically not a cell-complex in the strict sense of \defref{def:complex}. Yet, it is a ball complex\footnote{For a proof see \cite{pl1}}, that is a collection  $\{B_j|j=1,\dots,r\}$ of m-balls, $m\leq n$, which obey
\begin{enumerate}
\item $\overline{\Delta}=\bigcup\limits_{j=1}^r B_j$
\item $\mathring{B}_j\cap\mathring{B}_i=\emptyset$, if $i\neq j$ where $\mathring{B}$ is the interior of $B$
\item $\dot{B}_i$ is a finite union of balls of lower dimension in $\ast\Delta$ and every dual m-ball, $m<n$, lies in the frontier of at least one $m+1$-ball. 
\end{enumerate}
From the third property and \lemref{lem:dual} follows immediately that the subset $\ast\partial\Delta\subset\ast\Delta$, containing all $(n-1)$-balls $B^{(n-1)}_k$, which are adjacent to only one $n$-ball, and all balls in their frontier $\dot{B}^{(n-1)}_k$, is dual to the subcomplex $\partial\Delta$. 
\begin{definition}
Let $\Delta$ be a triangulation of a compact $4$-manifold then the \underline{dual 2-complex} $\kappa$ is the set obtained by removing all balls of dimension greater than two from $\ast\Delta$ and additionally all 2-balls from $\ast\partial\Delta$.
\end{definition} 
Since property two and three listed above still hold every 1-ball $e$ in $\kappa$ is contained in at least one 2-ball $f$. A 1-ball $e$ is adjacent to exactly one 2-ball $f$ if and only if $e\subset\partial\overline{\Delta}$ by \lemref{lem:dual}. As above we will call 1-balls contained in more than one $2$-ball internal, otherwise it is called external. Again by lemma \ref{lem:dual}, every vertex of $\kappa$ dual to a 4-cell must be the intersection of several internal edges. By the above construction every node in the boundary is the barycenter of a 3-cell and therefore the endpoint of exactly one internal 1-ball. Besides that, the dual 1-complex of $\partial\Delta$ is closed (every node of $\partial\kappa$ must be contained in at least two 1-balls), otherwise $\partial^2\Delta$ would not be empty, and bordering $\kappa$. This proves the first part of
\begin{theorem}
\label{theorem1}
If $\kappa_{\Delta}$ is the 2-complex dual to a triangulation $\Delta$ of a compact 4-manifold then $\kappa_{\Delta}$ is combinatorially equivalent to a foam $\kappa$, i.e. there exists a bijection\footnote{This map is defined on the complexes not on the underlying polyhedra! Furthermore, $\kappa$ is a p.l. complex in the strict sense of definition \ref{def:complex}.} $g:\kappa_{\Delta}\to\kappa$ mapping each n-cell of $\kappa_{\Delta}$ to an n-cell of $\kappa$ preserving the gluing relations (if $A$ is a common face of $B$ and $C$ then $g(A)$ is a common face of $g(B)$ and $g(C)$). Moreover, $\overline{\kappa_{\Delta}}$ is p.l. homeomorphic to $\overline{\kappa}$.
\end{theorem}
\begin{proof}
To prove that $\overline{\kappa_{\Delta}}$ and $\overline{\kappa}$ are p.l. homeomorphic we construct the following subdivision $\kappa'_{\Delta}$ and $\kappa'$: Since dual cells are by construction the underlying polyhedra of cell-complexes p.l. homeomorphic to m-balls, we can fix a point $x$ in the interior of a dual face $f\in\kappa_{\Delta}$ in such a way that the straight lines connecting $x$ and any barycenter $a^i_{(n-1)}\subset\dot{f}$ or any vertex of $f$ lies in $f$. When splitting every face in that way we obtain a simplicial complex $\kappa'_{\Delta}$ which is a subdivision of $\kappa_{\Delta}$. On the other hand, cells in $\kappa$ are already convex so that one can choose any point in the interior of each face $\tilde{f}\in\kappa$ and each edge $e\in\kappa$. By joining the points as above one can find a simplicial subdivision of $\kappa$ which is combinatorially equivalent to $\kappa'_{\Delta}$. Define $h:\overline{\kappa_{\Delta}}\to\overline{\kappa}$ by $h(x_i)=y_i$, if $x_i$ is a vertex of $\kappa'_{\Delta}$ and $y_i$ the corresponding vertex of $\kappa'$,  and extend it linearly. This gives the desired p.l. homeomorphism mapping $n$-cells of $\kappa'_{\Delta}$ to $n$-cells of $\kappa'$.    
\end{proof}
\section{Covariant Quantum Gravity}
\label{sec:model}
\subsection{BF-theory and EPRL-model}
The covariant quantization of GR is based on the observation that gravity is closely related to topological BF-theories. These theories are defined on the principal $G$-bundle over a smooth $D$-dimensional manifold $\mathcal{M}$ with connection $A$. The basic fields are the curvature $F[A]=\mathrm{d} A+A\wedge A$ and a (Lie) algebra $\mathfrak{g}$-valued $(D-2)$-form $B$. Classically, the BF-action  
\beq
\label{eqn:bf}
\mathcal{S}_{BF}=\int_{\mathcal{M}}\tr (B\wedge F[A])
\eeq
for four dimensions with gauge group $G=\SO(4)$ in euclidean models respectively $G=\SO(3,1)$ in lorentzian ones is equivalent to the Holst action \cite{HolstAction} iff the $B$-field can be expressed in terms of tetrads $E$ and the Hodge dual $\star$
\beq
\label{eqn:simp}
B=\star(E\wedge E)+\frac{1}{\beta} E\wedge E\;.
\eeq
The wedge product is taken with respect to the external indices, the trace in \eqref{eqn:bf} contracts the internal indices and $\beta$ is the Barbero-Immirzi parameter. The variation of \eqref{eqn:bf} with respect to the $B$-field constrains the curvature to vanish and formally the path integral is given by
\begin{align}
\begin{split}
\label{eqn:BFpart}
Z_{BF}(\mathcal{M}):=&\int\mathcal{D}\!A\int\mathcal{D}\!B\, \exp{i\int_{\mathcal{M}}\tr( B\wedge F)}\\
=&\int\mathcal{D}\!A\,\delta(F)
\end{split}
\end{align}
To obtain a covariant model of LQG we will first discretize,then quantize $\mathcal{Z}_{BF}$ and finally implement the \emph{simplicity constraints} \eqref{eqn:simp}.
\subsubsection{Discretized BF-theory}
\label{sec:discBF}
If $\Delta$ is a simplicial triangulation of a closed manifold $\mathcal{M}$ then the vector space $C^n(\Delta)$ of formal linear combinations of $n$-cells in $\Delta$ equipped with the scalar product $\scal{\sigma_i,\sigma_j}=\delta_{ij}$ is isometric to the space of $n$-forms with scalar product $\scal{\omega,\omega'}=\int\tr(\omega\wedge\star\omega)$. Furthermore, there is a one-to-one correspondence between the operations $(\wedge,\ast,d)$ and operations in $C^n(\Delta)$ (see \cite{sen-2000-61}). For example the hodge dual acts on cells by mapping to dual cells.

Within this scheme the $B$ fields of BF-theory are smeared on $(D-2)$-cells and $F$ on the dual faces such that
\beq 
\label{eqn:BFdisc}
S_{BF}=\sum_{c^{(D-2)}\in\;\Delta^{(D-2)}} \tr\left(\left[\int_{\ast{\Delta}c^{(D-2)}} F\right]\left[\int_{c^{(D-2)}}B\right]\right)\;.
\eeq
Remarkably, this step is independent of the chosen triangulation due to the topological nature of BF-theory. Only after the implementation of the simplicity constraint rendering the theory local the discretization yields a truncation of local degrees of freedom.

Recall that connections of a gauge theory are naturally regularized by holonomies $h_e[A]$ along paths $e\subset\mathcal{M}$ and therefore the `measure $\pathdif A$' in \eqref{eqn:BFpart} can be replaced by $\prod_e\dif\mu_H(g_e)$. Similarly, the curvature is regularized along a loop $\alpha$ enclosing a compact 2-d-surface $f$ since in second order approximation $h_{\alpha}[A]\approx \mathbf{1}+F(f)$ with $F(f)=\int_f F\in G$. Thus the curvature integral in \eqref{eqn:BFdisc} can be replaced by
\beq
\label{eqn:faceel}
\int_{f\in\ast\Delta^{(2)}} F\approx\prod_{e\in \dot{f}}g^{\;\;\epsilon_{e_f}}_{e}\equiv g_f~.
\eeq
for $D=4$. Here, $g_e$ are group elements attached to the edges $e$ bounding a face $f$ in the dual 2-complex\footnote{For the following it is not important that $\kappa$ is a ball-complex and the reader can safely assume that $\kappa$ is a foam.} $\kappa$ equipped with an orientation. The order of the group elements $g_e$ in \eqref{eqn:faceel} is determined up to cyclic permutations by the orientation of the face and $\epsilon_{ef}$ equals $1$ if $f$ is ingoing and $-1$ if $f$ is outgoing of $e$. Combining equation \eqref{eqn:faceel} and \eqref{eqn:BFdisc}, \eqref{eqn:BFpart} can be approximated by
\beq
\label{eqn:BFpartdisc}
Z_{BF}(\kappa)=\int\prod_{e\in\kappa^{(1)}}\dif g_e\prod_{f\in\kappa^{(2)}}\delta\left(\prod_{e\in \dot{f}}(g_{e})^{\epsilon_{ef}}\right)~.
\eeq
The above procedure can be easily generalized to arbitrary 4-manifolds: If $\mathcal{M}$ is non-compact one has to pass over to locally finite complexes (see appendix \ref{app:triangulation}). In order to keep everything finite we will not bother about this but always assume that $\mathcal{M}$ is a compact region of space-time. In the case that $\mathcal{M}$ has a non-empty boundary the action \eqref{eqn:bf} must be supplemented by a boundary term in order to leave the equations of motions unaltered (see e.g. \cite{Oriti:2002hv}). Without going into too much detail, $Z_{BF}$ can be constructed as in \eqref{eqn:BFpartdisc} just that the integral is only taken over 
bulk-variables. 

Following \cite{Ding:2010fw}, we split each edge $e$ into half-edges $l_{s(e)}$ and $l_{t(e)}$, where $l_{s(e)}$ is adjacent to the source and $l_{t(e)}$ to the target, and reorientate the half-edges in such a way that they are all oriented towards the splitting point. The half-edges are now labeled by group elements $g_{l_{s(e)}}$ and $g_{l_{t(e)}}$ obeying
\beq
\label{eqn:newel1}
g_e= g_{l_{s(e)}}g^{-1}_{l_{t(e)}}\;.
\eeq
After introducing these new variables, the group elements can be rearranged defining
\beq
\label{eqn:newel}
g_{f_v}:=g_{l_v}^{-1}g_{{l'}_v}
\eeq 
where $l_v$ is the half-edge in the frontier of $f$ adjacent to $v$. Note, if $\epsilon_{ef}=1$ then $l_v=l_{t(e)}$ otherwise $l_v$ is the half-edge of an edge with source $v$ (see \figref{fig:labeling}). In these variables the discretized BF-partition function is given by
\begin{figure}[t]
 \begin{center}
	\subfloat{\includegraphics{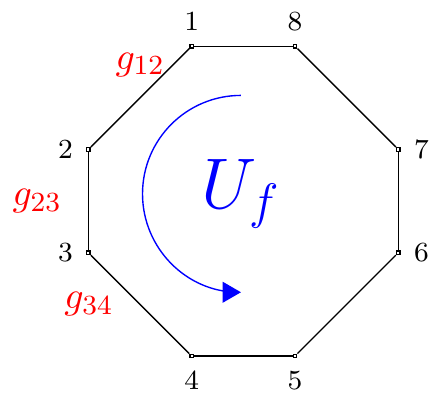}}
	\hspace{1.5 cm}
	\subfloat{\includegraphics[width=4.5 cm]{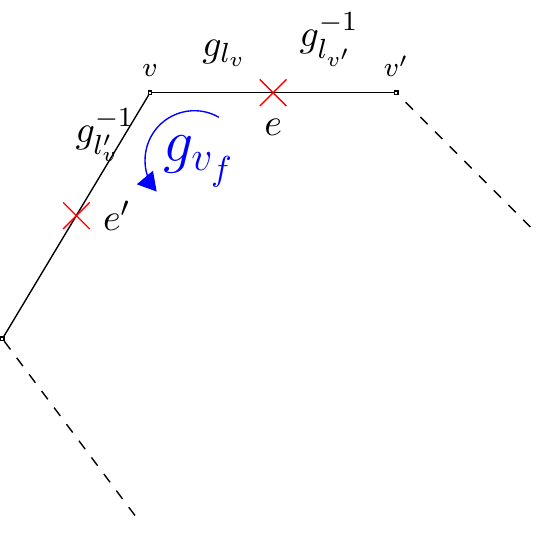}}
 \end{center}
\caption{Labeling of the face $f$ and its adjacent edges by group elements $g_{l_v}$.}
\label{fig:labeling}
\end{figure}
\beq
\label{eqn:zbound}
Z[\kappa]=\int\limits_{\SO(4)}\left(\prod_{v\in\kappa}\dif g_{f_v}\right)\; \prod_{f\in\kappa}\delta\left(\prod_{v\in\dot{f}}g_{f_v}\right)\prod_{e\in\partial\kappa}\delta(g_{f_{s(e)}}g_{f_{t(e)}}^{-1}g_e^{-1})\prod_{v\in\kappa_{int}} \mathcal{A}_v(\{g_{f_v}\})~.
\eeq
Here, $\{g_{f_v}\}$ is the set of all group elements $g_{f^1_v},\dots,g_{f^n_v}$ assigned to the $n$ faces adjacent to $v$ and
\beq
\label{eqn:vertexam1}
A_v(\{g_{f_v}\})
:=\int\left[ \prod_{l_v\in\mathcal{V}(v)}\dif{g_{l_v}}\right] \prod_{f\in\mathcal{V}(v)}\delta\!\left(g_{f_v}^{-1}g_{l_v}^{-1}g_{l'_v}\right)
\eeq
reverses the substitution \eqref{eqn:newel1} and \eqref{eqn:newel} in the bulk while $\delta(g_{f_{s(e)}}g_{f_{t(e)}}^{-1}g_e^{-1})$ reverses it on the boundary. By Weyl's orthogonality formula the group convolution $\delta(g)$ can be expressed by a sum over the characters of its irreducible representations. This can be used in order to expand the vertex amplitude \eqref{eqn:vertexam1} in terms of spin net functions. The Euclidean\footnote{For the Lorentzian model see \cite{Ding:2010fw,Engle:2007uq}} gauge group $\SO(4)\simeq \SU(2)_L\times\SU(2)_R\big/\Z_2$ is locally defined by a left (L) and right (R) action of $\SU(2)$. Because of that, irreps of $\SO(4)$ are given by the tensor representations $\rho=(j^L,j^R)$ of $\Spin(4)\simeq\SU(2)_L\times\SU(2)_R$ for which $j^L+j^R\in\mathbb{N}$. While by no means justified from the $\SO(4)$ point of view, we will work with $\Spin(4)$ from the beginning in order to avoid the above limitation on spins $j^{L/R}$. The convolution $\delta(g)$ is then defined by
\beq
\label{eqn:deltaspin4}
\delta(g)=\sum\limits_{j^L,j^R}d_{j^L}\; d_{j^R} \;\chi^{j^L}(g^L)\;\chi^{j^R}(g^R)
\eeq
with $g=(g^L,g^R)$ and $g^{L/R}\in\SU(2)$. 

Taking into account that every edge adjacent to an internal vertex is itself internal (\lemref{lem:intvert}), every group element $g_{l_v}$ in equation \eqref{eqn:vertexam1} appears at least in two different face distributions\footnote{When restricting foams to complexes dual to a triangulation then every internal edge is adjacent to at least four faces since the smallest three-cell in $\Delta$ is a tetrahedron.}. Thus, we have to integrate over products of characters. Consider for example a vertex $v$ splitting a trivalent edge into two half-edges $l_v,l'_v$. In this case one has to compute integrals of the form
\beq
\label{eqn:vertexbsp}
\mathcal{I}=\int\limits_{\SU(2)}\dif{h_{l_v}}\dif{h_{l'_v}}\prod_{i=1}^3\chi^{j_{f_i}}(h_{f_i}^{-1}h_{l_v}^{-1}h_{l'_v})
\eeq
when evaluating \eqref{eqn:vertexam1} at $v$. This integral can be easily solved (see \eqref{eqn:schurA} and \eqref{eqn:recoupl}) and yields
\beq
\mathcal{I}=\tr(\iota_{l'_v}\iota^{\dagger}_{l_v})\;\tr[\iota_{l_v}\prod_{i=1}^3 R^{j_i}(h_{f_i})\;\iota^{\dagger}_{l'_v}]
\eeq
with $\iota_{l_v},\iota_{l'_v}\in\Inv[\bigotimes\limits_{i=1}^3\hilbert_{j_{f_i}}]$. The second trace constitutes a (non-normalized) spin net function on the vertex graph $\gamma_v$. Using that $\Spin(4)$ functions $T^{BF}$ can be expanded in terms of $\SU(2)$ spin nets
\beq
T^{BF}_{\gamma_v,\rho,\iota}(\{g_{f_v}\})=T_{\gamma_v,j^L,\iota^L}(\{g^L_{f_v}\})\otimes T_{\gamma_v,j^R,\iota^R}(\{g^R_{f_v}\})\\
\eeq
a vertex amplitude at $v\in\kappa_{int}$ with vertex boundary graph $\gamma_v$ is generally given by 
\begin{align}
\label{eqn:vertexam2}
A_v(\{g_{f_v}\})
=\sum_{\{\rho_f\},\{\iota_l\}}\;\prod_{f\in\mathcal{V}(v)}\sqrt{\dim \rho_f}\;\;\; \tr\left(\bigotimes_{l\in\mathcal{V}(v)} \iota^{\dagger}_l\right) \;\;\; T^{BF}_{\gamma_v,\rho,\iota}(\{g_{f_v}\})\;.
\end{align}
The notation $\,^{\dagger}$ symbolizes that the intertwiners in $\tr$ are dual\footnote{For intertwiners based on $3j$-symbols/Clebsch-Gordan coefficients this difference is of academic nature since they are self-dual.} to the corresponding intertwiners in the spin net function. The sum over all labelings $\rho_f=(j_f^L,j^R_f)$ and the dimensional factor are remains of \eqref{eqn:deltaspin4} while the summation over orthonormal intertwiners $\sum\limits_{\iota}$ results from integrating products of more than three characters (see Appendix \ref{app:HA} equation \eqref{eqn:othogjint4}).

Each element $g_{f_v}$ associated to an internal vertex appears exactly twice in \eqref{eqn:zbound}, once in a vertex amplitude and once in the first distribution. Thus the integration over the bulk variables $g_{f_v}$ relates the vertex amplitudes by fixing the representation associated to the faces and causes
\beq
\label{eqn:partbffin}
Z[\kappa]=\sum_{\{\rho_f\},\{\iota_l\}}\prod_f \;d_{\rho_f}\prod_{v\in\kappa_{int}}A_v(\{\iota_{l_v}\})\;\; T^{BF}_{\partial\kappa,\rho,\iota}(\{g_{e_f}\})
\eeq
with
\beq
\label{eqn:vertexTr}
A_v(\{\iota_{l_v}\})=\tr\left(\bigotimes_{l_v\in\mathcal{V}(v)} \iota_{l_v}
\right)
=\tr\left(\bigotimes_{l_v\in\mathcal{V}(v)} \iota^L_{l_v}\otimes \iota^R_{l_v}
\right)~.
\eeq
This function coincides with \eqref{eqn:abstractpart} where $Q_e$ is the identity except for an additional face amplitude. 

So far, we only quantized BF-theory and still have to impose the simplicity constraint.
\subsubsection{The EPRL-model}
\label{ssec:EPRL}
Let us begin by a short overview of the EPRL-model \cite{Engle:2007uq}. To implement the simplicity constraint \eqref{eqn:simp} in the model we need to discretize it but the non-trivial dependence on the tetrad fields is complicating the matter. Therefore, we replace \eqref{eqn:simp} by $B=\Sigma +\frac{1}{\beta}\ast\Sigma$ where $\Sigma$ is a $\algebra{g}$ valued two-form satisfying\footnote{This idea goes back to \cite{bc1997}.} 
\beq
\label{eqn:simpII}
\Sigma^{IJ}\wedge \Sigma^{KL}=\frac{1}{4!}\epsilon^{IJKL}\epsilon_{MNPQ}\,\Sigma^{MN}\wedge \Sigma^{PQ}~.
\eeq 
The solutions of condition \eqref{eqn:simpII} fall into five sectors:
\begin{enumerate}
\item[(I$\pm$)] $\Sigma = \pm E \wedge E$ 
\item[(II$\pm$)] $\Sigma = \pm \ast E \wedge E$
\item[(deg)]  $\tr(\ast E \wedge E) = 0$
\end{enumerate}
The original constraint \eqref{eqn:simp} is, of course, only recovered if $\Sigma$ is in sector (II+) and thus one would need to implement an additional constraint. Nevertheless, the necessity of an additional constraint is widely ignored and we do so as well\footnote{For a suggestion of a constraint, forcing $\Sigma$ to be in (II+), see \cite{Engle:2011ps}.}.

As stated previously, 2-forms are naturally discretized on two dimensional surfaces. Consider for simplicity a 4-simplex\footnote{A 4-simplex is the complex hull of five points not all of which lie in a 3-d hyperplane.} $\sigma$ embedded in a manifold $\mathcal{M}$, label the vertices by $a=1,\dots,5$ and let $\tau_a$ be the tetrahedron not containing vertex $(a)$ and $\Delta_{ab}$ be the triangle $\tau_a\cap\tau_b$. Then,
\beq
\Sigma_{ab}^{IJ}:=\int_{\Delta_{ab}} \Sigma^{IJ}_{\mu\nu}~.
\eeq
and \eqref{eqn:simpII} is replaced by (see \cite{Engle:2007uq}):
\begin{enumerate}
\item\textit{Diagonal simplicity:}  
$
\ast\Sigma_{ab}\cdot\Sigma_{ab}=0
$
\item\textit{Off-diagonal simplicity:}
$
\ast\Sigma_{ab}\cdot\Sigma_{ac}=0 \quad\forall c\neq b, c\neq a
$
\item\textit{Dynamical simplicity}
\end{enumerate}
Furthermore, the bivectors $\Sigma_{ab}$ are closed, $\sum_{b:b\neq a}\Sigma_{ab}=0$, due to gauge-invariance. If $\sigma$ is non-degenerate, meaning that the tetrahedra span 3-dimensional subspaces and can be glued such that the resulting 4-simplex $\sigma$ spans a 4-dimensional subspace, then $\{B_{ab}\}$ satisfy additional non-degeneracy and orientation conditions. Each non-degenerate 4-simplex determines a unique set of such bivectors and each set of bivectors satisfying the above constraints determines a 4-simplex (see \cite{bc1997}). 

The dynamical simplicity constraint does not have to be implemented since diagonal, off-diagonal simplicity and closure already imply dynamical simplicity (however, these three sets of constraints
are stronger than 1., 2. 3. stated above). Moreover, the off-diagonal simplicity constraint can be replaced by the following condition:
\beq
\forall \tau_a\in\sigma\; \;\exists\, N_a\in\R^4\; \text{ s.t. } (N_a)_I (\ast\Sigma_{ab})^{IJ}=0\;\;\forall b\neq a~.
\eeq
Current spin foam models are based on a heuristic derivation of a discretized formal quantum gravity path integral form the plebanski action.
When imposing above constraints by integration over the corresponding Lagrange 
multipliers, one replaces curvatures by holonomies around loops $\dot{f}$ 
bounding faces $f$ in the dual complex. The $B$ fields, which are naturally smeared on the triangle $t_f$ dual to $f$, are replaced by invariant vector fields on the copy of $G$ corresponding to $f$. While in the formal continuum path integral the B fields were commuting integration variables and thus commuting constraints, they become now non commuting operators resulting in non commuting constraints. 
The diagonal simplicity constraint still commutes with all the other constraint and for this reason can be imposed strongly. It restricts the representations to those which obey
\beq
\left(\frac{2j_L}{1-\beta}\right)^2=\left(\frac{2j_R}{1+\beta}\right)^2~.
\eeq
The off-diagonal constraints are more complicated. In the new models \cite{Engle:2007uq,FK} they are treated by a master constraint $\hat{M}$ which projects states onto the highest weight, $j=j_L+j_R$, for $|\beta|<1$ and on the lowest weight, $j=j_L-j_R$, for $|\beta|>1$ of the decomposition $\hilbert_{(j_L,j_R)}|_{\SU(2)}\simeq \hilbert_{|j_L-j_R|}\oplus\cdots\oplus\hilbert_{j_L+j_R}$. This is a weak implementation of $\hat{M}$ in the sense that there exist some Hilbert space $\hilbert$ such that $\scal{\psi,\hat{M}\phi}=0$ for all elements $\psi,\phi\in\hilbert$. As shown in \cite{Ding:2009jq,Ding:2010fw} such a space of weak solutions is spanned by elements 
\beq 
\label{eqn:boundarysn}
T^{EPRL}_{\gamma_v,j^{\pm},\eta}(\{g_{f_v}\})
=\prod_{f\in\mathcal{V}(v)}\sqrt{d_{j_f^+}\;d_{j_f^-}}\;
\tr\left[ \prod_{f_v\in\mathcal{V}(v)}R^{j_{f_v}^+}(g_{f_v}^+)\; \;R^{j_{f_v}^-}(g_{f_v}^-) \prod_{e_v\in\mathcal{V}(v)}\!\tau^{\text{EPRL}}(\eta_{e_v})\right]~.
\eeq
where $\gamma_v$ is the vertex boundary graph associated to a 4-simplex (see \figref{fig:4-simplex}), $j^{\pm}$ and $j$ are $\SU(2)$-irreducibles satisfying 
\begin{figure}[t]
 \begin{center}
	\includegraphics[width=3 cm]{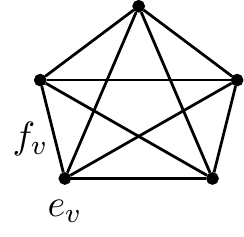}
 \end{center}
\caption{The vertex boundary graph of a vertex $v$ dual to a 4-simplex $\sigma$ is obtained by associating a node to every tetrahedron of $\sigma$ and a link to every triangle. In the dual 2-complex every node of $\gamma_v$ corresponds to an internal edge $e_v\in\mathcal{V}(v)$ and every link of $\gamma_v$ to a face $f_v\in\mathcal{V}(v)$.}
\label{fig:4-simplex}
\end{figure}
\beq
\label{eqn:j}
j^{\pm}\equiv\frac{|\beta\pm1|}{2}j,
\eeq
and $\tau^{\text{EPRL}}$ is a map 
\begin{gather}
\begin{gathered}
\label{eqn:eint}
\tau^{\text{EPRL}}: \Inv_{\SU(2)}\left(\bigotimes\limits_{i=1}^4 \hilbert_{j_i}\right)
\to\Inv_{\SU(2)}\left(\bigotimes\limits_{i=1}^4 (\hilbert_{j^+_i}\otimes \hilbert_{j^-_i})\right)\\[3pt]
\left[\tau^{\text{EPRL}}(\eta_{e_v})\right]^{\{j^+_f,m^+_f\},\{j^-_f,m^-_f\}}=\tr\left[ \eta_{e_v}^{\{j_f,A_f\}} \prod_{f\in\mathcal{V}(e)} [C_f]^{j^+_f,m^+_f;j^-_f,m^-_f}_{j_f,A_f}\right] 
\end{gathered}
\end{gather}
coupling $j^+$ and $j^-$ to $j$ by the Clebsch-Gordan coefficient  $C^{j^+,m^+;\;j^-,m^-}_{j,A}:=\scal{j,A |j^+,m^+;j^-,m^-}$. In the subsequent discussion we will call this space of weak solutions $\hilbert^{EPRL}_{\gamma_v}$. The map $\tau^{\text{EPRL}}$ is, of course, only well-defined if $\frac{|\beta\pm1|}{2}k$ is a half-integer which puts additional constraints on $\beta$ and $j$. Yet, this problem only occurs in the Euclidean theory and can be avoided by requiring $\beta$ to be an odd integer\footnote{In this case, $\tau^{\text{EPRL}}$ is injective (see third reference of \cite{Kaminski:2009fm}).}.

Following the above considerations, the off-diagonal constraints are implemented weakly in the model when projecting the BF-Amplitude onto $\hilbert^{EPRL}$:
\beq
\label{eqn:vertamp}
\mathcal{A}^{EPRL}_v(\{g_{f_v}\})=\sum_{j_f,\iota_e}\;\;\scal{T^{EPRL}_{\gamma_v,j_f,\iota_e}|\mathcal{A}_v}\;T^{EPRL}_{\gamma_v,j_f,\iota_e}(\{g_{f_v}\})~.
\eeq
This is non-zero iff $(j^L,j^R)\equiv(j^-,j^+)$ and obviously also implements diagonal simplicity. Plugging this back into the full partition function results in
\beq
\label{eqn:zbound1}
\begin{split}
Z[\kappa]= &\sum_{\{j_f^{\pm}\},\{\eta_{e_v}\}}\;\prod_f\;d_{j_f^+}\;d_{j_f^-}\prod_{v\in\mathcal{V}_{int}}\;
\left\{\prod_{e_v}\sum_{\iota_{e_v}^+,\iota_{e_v}^-} f^{\eta_{e_v}}_{\iota_{e_v}^+,\iota_{e_v}^-}\mathcal{A}_v(\{\iota^{\dagger}_{e_v}\})\right\}\\
&\times\sum_{\{j_l\},\{\eta_{e_{\n}}\}} \left(\prod_{\l_f\in\partial\kappa^{(1)}}\frac{1}{\sqrt{d_{j_{\l_f}^+}\;d_{j_{\l_f}^-}}}\right)
T^{EPRL}_{\partial\kappa,j^{\pm}_{\l_f},\eta_{e_{\n}}}(\{g_{\l_f}\})
\end{split}
\eeq
where $ f^{\eta_e}_{\iota_e^+,\iota_e^-}$ are the well known fusion coefficients \cite{Engle:2007uq}
\beq
\label{eqn:fusion}
f^{\eta}_{\iota^+,\iota^-}:=\tr\left[ \overline{\tau^{\text{EPRL}}(\eta)}\;\;\iota^+\iota^-\right] ~.
\eeq
The above model can be extended to non-degenerate arbitrary triangulations (see \cite{Ding:2010fw}) by making use of Minkowski's theorem \cite{Minkowski:poly} stating that a polyhedron is uniquely determined, up to inversion and translations, by its face areas and normals. 
\subsubsection{The KKL-model}
\label{sec:KKL}
The above approach meets several technical challenges. 
Apart from those connected to the non commutative nature of the simplicity constraints 
mentioned above, several difficulties arise when trying to combine covariant and canonical LQG. Heuristically, the `time evolution' of a spin-network would produce a spin foam but a generic foam is not dual to a triangulation (see \secref{sec:triangulation} for a detailed discussion).
Furthermore, \eqref{eqn:eint} is an $\SU(2)$ intertwiner and accordingly $T^{EPRL}$ is not $\Spin(4)$ but $\SU(2)$ invariant. Both problems are avoided in the KKL-approach \cite{Kaminski:2009fm}.

Consider an arbitrary foam (\secref{ssec:spinf}) whose faces are colored by irreps of $\Spin(4)$ and whose edges are labeled by operators $Q_e$ in the induced intertwiner space. If we choose $Q_e$ to be the identity we formally recover BF-theory \eqref{eqn:partbffin}, to `implement' the simplicity constraint one has to restrict the coloring to EPRL data:
\begin{align}
&f\to\rho_f\equiv(j_f^+,j^-_f)&&\forall f\in\kappa^{(2)}\\
&e\to \zeta^{KKL}(\eta_{s(e)})\otimes\zeta^{\dagger}_{KKL}(\eta^{\dagger}_{t(e)})&&\forall e\in\kappa^{(1)}_{int}
\end{align}
where 
\begin{gather}
\begin{gathered}
\label{eqn:EPRL}
\zeta^{KKL}:\Inv_{\SU(2)}\left(\bigotimes_{f}\hilbert^{j_f}\right)\to
\Inv_{\Spin(4)}\left(\bigotimes_{f} \hilbert^{\rho_f}\right)\\
\eta\mapsto \sum_{\iota^{\pm}} f_{\iota^+\,\iota^-}^{\eta}\, \iota^+\otimes\,\iota^-
\end{gathered}
\end{gather}
maps $\SU(2)$ intertwiners $\eta$ to $\Spin(4)$ ones. Assuming that all edges are incoming, the vertex amplitude \eqref{eqn:verttr} is given by
\begin{gather}
\begin{gathered}
\tr\left(\bigotimes_{e_v\in\mathcal{V}(v)}\zeta^{KKL}(\eta_{e_v})\right)
=\left[\prod_{e_v\in\mathcal{V}(v)} \sum_{\iota_{e_v}^+,\iota_{e_v}^-}f^{\eta_{e_v}}_{\iota_{e_v}^+,\iota_{e_v}^-}\right]
\mathcal{A}_v(\{\iota^{\pm}_{e_v}\})\;.
\end{gathered}
\end{gather}
When each edge is labeled by an operator of the type $Q_e= \ket{\zeta^{KKL}(\eta)}\bra{\zeta^{KKL}(\eta)}$ then the KKL-partition function 
\beq
\begin{split}
\label{eqn:KKL1}
Z^{KKL}[\kappa]= &\sum_{\{j_f^{\pm}\},\{\eta_{e_v}\}}\prod_{v\in\kappa^{(0)}_{int}}\;
\left\{\prod_{e_v}\sum_{\iota_{e_v}^+,\iota_{e_v}^-} f^{\eta_{e_v}}_{\iota_{e_v}^+,\iota_{e_v}^-}\mathcal{A}_v(\{\iota_{e_v}\})\right\}
T^{KKL}_{\partial\kappa,j^{\pm}_{\l_f},\eta_{e_{\n}}}(\{g_{\l_f}\})
\end{split}
\eeq
is almost the same as \eqref{eqn:zbound1} but differs in the states induced on the boundary graph
\beq
\label{eqn:KKLstates}
T^{KKL}_{\partial\kappa,j_{\l},\eta_{\n}}(\{g_{\l}\}):=\tr\left(\prod_{\l\in\partial\kappa^{(1)}}R^{j^+_{\l}}\!(g^+_{\l})\; R^{j^-_{\l}}\!(g^-_{\l})
\prod_{\n\in\partial\kappa^{(0)}}\left[\sum_{\iota^+_{\n},\iota^-_{\n}} 
f^{\eta_{\n}}_{\iota^+_{\n}\iota^-_{\n}}\;\iota^+_{\n}\otimes\iota^-_{\n}\right]\right)\;.
\eeq
In contrast to the $\SU(2)$ intertwiner $\tau^{\text{EPRL}}(\eta)$ the intertwiner $\zeta^{KKL}(\eta)$ is $\Spin(4)$ invariant and therefore the space $\hilbert_{KKL}$ spanned by the states \eqref{eqn:KKLstates} is a proper subspace of $\hilbert_{BF}$. For a visualization of the different spin nets see \figref{fig:KKL}.
\begin{figure}[t]
 \begin{center}
	\subfloat[A link in $\hilbert^{KKL}$]{\includegraphics{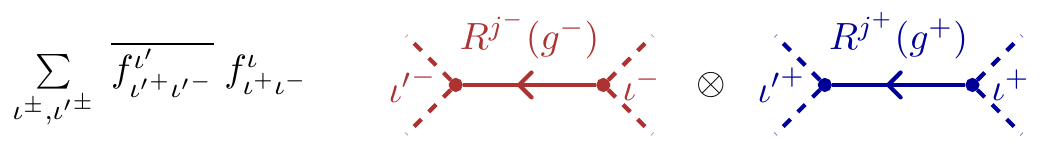}}
	\\
	\subfloat[A link in $\hilbert^{EPRL}$]{\includegraphics{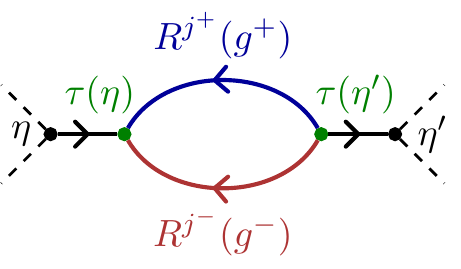}}
 \end{center}
\caption{Different graphical visualization of a link in $\hilbert^{KKL}$ and in $\hilbert^{EPRL}$.}
\label{fig:KKL}
\end{figure}

Although, \eqref{eqn:KKLstates} are linearly independent they are not orthogonal (see \cite{Kaminski:2009fm}) with respect to the BF-scalar product since
\begin{gather}
\begin{gathered}
\label{eqn:BFscalKKL}
\scal{T^{KKL}_{\partial\kappa,j_{\l},\eta_{\n}}|T^{KKL}_{\partial\kappa,j'_{\l},\eta'_{\n}}}_{BF}:=\!\!\!\!\int\limits_{\Spin(4)}\!\!\! \prod_{\l}\dif g_{\l}\;\overline{T^{KKL}_{\partial\kappa,j_{\l},\eta_{\n}}(\{g_{\l}\})}\;\;T^{KKL}_{\partial\kappa,j'_{\l},\eta'_{\n}}(\{g_{\l}\})\\
=\prod_{\l}\frac{\delta_{j_{\l},j'_{\l}}}{ d_{j^+_{\l}}\;d_{j^-_{\l}}} \prod_{\n}\underbrace{\left[\sum_{\iota_{\n}^{\pm}} \overline{f^{\eta'_{\n}}_{\iota_{\n}^+\iota^-_{\n}}}\;f^{\eta_{\n}}_{\iota_{\n}^+\iota^-_{\n}}\right]}_{H^{\eta_{\n}}_{\eta'_{\n}}}
\end{gathered}
\end{gather}
and $H^{\eta}_{\eta'}$ is in general not even diagonal. However, the KKL-map \eqref{eqn:EPRL} is injective \cite{Kaminski:2009fm} and if $\{a\}$ is an orthonormal basis in $\SU(2)$ intertwiner space then $\{\zeta^{KKL}(a)\}$ constitutes a basis in the KKL-intertwiner space. Instead of diagonalizing this basis we can, introduce an operator $Q$
\begin{align}
&Q:\Inv^{KKL}\left(\bigotimes_{e} \hilbert_{\rho_e}\right) \to \Inv^{KKL}\left(\bigotimes_{e} \hilbert_{\rho_e}\right)\\
\begin{split}
\label{eqn:qexp}
& \qquad Q[\zeta^{KKL}(a)]=\sum_b Q^{\,a}_{\;\;b}\;\zeta^{KKL}(b)
\end{split}
\end{align}
such that 
\beq
\label{eqn:diag}
\delta_{c}^{\;a}=\scal{\zeta^{KKL}(c)|Q[\zeta^{KKL}(a)]}_{BF}
\eeq
as suggested in \cite{Bahr:2010bs}. By expanding the KKL-map $\zeta^{KKL}(a)=f^a_{a^+ a^-}\;a^+\otimes a^-$ in an orthonormal basis $\{a^+\otimes a^-\}$ of $\Spin(4)$ intertwiners, $Q$ can be defined equivalently by
\begin{align}
\label{eqn:idkkl}
\begin{split}
Q&=\tilde{Q}^{\,a}_{\;\;b}\; (\zeta^{KKL}(a))^{\dagger}\otimes\zeta^{KKL}(b)\\[3pt]
	&=\overline{\fusion{a}}\;\tilde{Q}^{\,a}_{\;\;b}\;\fusion{b}\; (a^+\otimes a^-)^{\dagger}\otimes(b^+\otimes b^-)
\end{split}
\end{align}
where $\overline{\fusion{a}}\;\tilde{Q}^{\,a}_{\;\;b}\;\fusion{b}=\delta^{\,a^+}_{\;\;b^+}\,\delta^{\,a^-}_{\;\;b^-}\, [H^{-1}]^{\,a}_{\;\;b}$. Thus, the natural scalar product on $\hilbert_{KKL}$ is the product
\begin{gather}
\begin{gathered}
\label{eqn:scalarKKL}
\scal{T^{KKL}_{\gamma,a_{\n},j_{\l}}|T^{KKL}_{\gamma',b_{\n'},j'_{\l'}}}_{KKL}
:=\scal{T^{KKL}_{\gamma,a_{\n},j_{\l}}|\prod_{\n\in\gamma} Q_{\n}|T^{KKL}_{\gamma',b_{\n'},j'_{\l'}}}_{BF}\\
=\delta_{\gamma,\gamma'}\prod_{\l} \frac{\delta_{j_{\l},j'_{\l}}}{d_{j_{\l}^+}d_{j_{\l}^-}}\prod_{\n}\delta_{a_{\n},b_{\n}}
\end{gathered}
\end{gather}
with respect to which the states $T^{KKL}$ are orthogonal. Note, the identity operator w.r.t \eqref{eqn:scalarKKL} is formally
\beq
\left(\prod_{\l\in\gamma} d_{j_{\l}^+}d_{j_{\l}^-}\right)\;\mathbb{1}_{\gamma, KKL}\;=\ket{T^{KKL}_{\gamma}}_{BF}\bra{T^{KKL}_{\gamma}}\prod_{\n\in\gamma} Q_{\n}~.
\eeq
In contrast, the EPRL-states \eqref{eqn:vertamp} are already orthogonal and therefore it is possible to implement the simplicity constraint before performing the integration on the bulk variables. In the second approach the model is defined by restricting the representations of the BF-partition function \eqref{eqn:partbffin}. If we would have done this before integration then we would encounter additional edge amplitudes due to \eqref{eqn:BFscalKKL}. Therefore, it is advisable to label each internal edge by an operator $Q_e$ so that \eqref{eqn:KKL1} is replaced by
\beq
\label{eqn:BHKKL1}
\tilde{Z}^{KKL}[\kappa]=
\sum_{\substack{\{j_f\}\\\{a_{e_v}\}}}
\sum_{\{a^{\pm}_{e_v}\}}\;
\prod_{e\in\kappa_{int}}\mathcal{A}_e(a_{s(e)}^{\pm},a_{t(e)}^{\pm})
\prod_{v\in\kappa_{int}} \mathcal{A}_v(a^{\pm}_{e_v})\;
T^{BF}_{\partial\kappa,a_{e_{\n}},j^{\pm}_{\l_f}}(\{g_{\l_f}\})
\eeq
with edge amplitude
\beq
\mathcal{A}_e(a^{\pm}_{s(e)},a^{\pm}_{t(e)}):=\overline{f^{\;a_{s(e)}}_{a_{s(e)}^+a_{s(e)}^-}}\;\;(\tilde{Q}_e)^{a_{s(e)}}_{a_{t(e)}}\;\;f^{a_{t(e)}}_{a_{t(e)}^+a_{t(e)}^-}\;.
\eeq
Here, the fusion coefficients are absorbed in the edge amplitude and therefore $T^{KKL}$ had to be replaced by $T^{BF}$. Nevertheless,  \eqref{eqn:BHKKL1} still defines a distribution in $\hilbert_{KKL}$.
In the following we will mainly work with the object \eqref{eqn:BHKKL1} and just write $Z[\kappa]$ instead of $\tilde{Z}^{KKL}[\kappa]$ to keep the notation simple.
\subsection{Projected spin networks}
\label{ssec:psn}
The advantage of the KKL-model is the preservation of covariance (see \cite{Rovelli:2010ed}) but merging canonical and covariant approach is more complicated and will involve so-called projected spin nets \cite{Alexandrov:2002br,Dupuis:2010jn}. 

The difficulty is to find a map $\hilbert_{0}\to\hilbert_{KKL}$ projecting the $\SU(2)$ invariant functions in $\hilbert_{0}$ onto $\Spin(4)$ invariant functions $T^{KKL}$. To do so we first need to establish an isomorphism between $\SU(2)$ and a $\SU(2)$-subgroup of $\Spin(4)$. Unfortunately, there exist no canonical choice of such a subgroup but one has to fix a normal $\mathbf{n}$ left invariant by the $\SU(2)$-subgroup $\SU_{\mathbf{n}}(2)\in\Spin(4)$.

The manifold $\SU(2)$ is isomorphic to the sphere $S^3$, which is uniquely determined by the set of vectors $\mathbf{n}\in\R^4$, $\|\mathbf{n}\|=1$, that is, we can define a bijection
\beq
\label{eqn:isomsphere}
\omega:S^3\to\SU(2)\quad \mathbf{n}\mapsto \omega(\mathbf{n})=\frac{1}{2}\mathbf{n}_{\mu}\sigma^{\mu}
\eeq
where $\sigma^0=\mrix{1}_{2}$ and $\sigma^i$ are the Pauli matrices and construct a projection $\pi_2:\Spin(4)\to\SO(4)$, $(g_L,g_R)\mapsto E(g_L,g_R)$, such that $\omega(E(g_L,g_R)
\cdot \mathbf{n}):=g_L\;\omega(\mathbf{n})(g_R)^{-1}$. Fix $\time\equiv(1,0,0,0)$ then the $\SU(2)$-subgroup $\SU_{\time}(2)\subset\Spin(4)$ stabilizing $\time$ is the set of all elements $(h,h)\in\Spin(4)$. Since any normal is uniquely determined by the action of $\SO(4)$ on $\time$, i.e.
\beq
\label{eqn:omegan}
\omega(\mathbf{n})=\omega(E(B_{\mathbf{n}}^L,B_{\mathbf{n}}^R)\time)=B_{\mathbf{n}}^L(B_{\mathbf{n}}^R)^{-1}
\eeq
for some $B_{\mathbf{n}}=(B_{\mathbf{n}}^L,B_{\mathbf{n}}^R)\in\Spin(4)$, the subgroup $\SU_{\mathbf{n}}(2)$ is the set of all elements 
\bq
B_{\mathbf{n}}\triangleright(h,h)
=(B_{\mathbf{n}}^L\, h\,(B_{\mathbf{n}}^L)^{-1},B_{\mathbf{n}}^R\, h\,(B_{\mathbf{n}}^R)^{-1})~. 
\eq
Note, the projection $\pi_2$ is two-to-one because $E(g_l,g_R)=E(-g_L,-g_R)$ which is due to the fact that $\Spin(4)$ is the double cover of $\SO(4)$.

On the one hand, it is necessary to fix a normal in order to identify the different copies of $\SU(2)$ but, on the other hand, this breaks $\Spin(4)$-invariance. A way out of this dilemma is to consider spin network functions whose nodes $v$ are also labeled by normals $\mathbf{n}_v$ that transform in the defining $\SO(4)$-representation: $\Lambda\triangleright \mathbf{n}:= E(g_L,g_R)\mathbf{n}$ for $\Lambda=(g_L,g_R)\in\Spin(4)$. Let $\mathcal{K}$ be the space of square integrable, gauge invariant functions   
\beq
\label{eqn:proj}
\phi(\gamma,\{g_{\l}\}, \{\mathbf{n}_v\})=\phi(\gamma,\{\Lambda_{s(\l)}g_{\l}\Lambda_{t(\l)}^{-1}\},\{\Lambda_v\triangleright \mathbf{n}_v\})
\eeq
with the scalar product
\beq
\label{eqn:scalproj}
\scal{\phi|\phi'}
=\delta_{\gamma,\gamma'}\left(\prod_v\int_{S^3}\!\dif n_v\,\delta(\mathbf{n}_v-\mathbf{n}'_v)\right)
\int[\prod_{\l}\dif{g_{\l}}]\;
\overline{\phi(\gamma,\{g_{\l}\},\{\mathbf{n}_v\})}\phi'(\gamma',\{g_{\l'}\},\{\mathbf{n'}_{v'}\})~.
\eeq
Remarkably, the so-called \emph{projected spin network functions} \ref{eqn:proj} do not depend on the choice of the normal $\mathbf{n}_v$. In \cite{Dupuis:2010jn}, the authors have shown, using Schur orthogonality, gauge invariance \eqref{eqn:proj} and the properties of the intertwiners \eqref{eqn:eint}, that $\mathcal{K}$ is spanned by the orthonormal functions
\begin{gather} 
\label{eqn:basisproj}
\begin{gathered}
\phi_{\gamma,j^{R,L},\eta}(\{g_{\l}\},\{\mathbf{n}_v\}):=\\
\prod_{\l\in\gamma}\sqrt{d_{j_{\l}^L}\;d_{j_{\l}^R}}\;
\tr\left[\prod_{\l\in\gamma}
R^{j_{\l}^L}((B^L_{\mathbf{n}_{s(\l)}})^{-1}g^L_{\l}\,B^L_{\mathbf{n}_{t(\l)}})\;
R^{j_{\l}^R}((B^R_{\mathbf{n}_{s(\l)}})^{-1}g^R_{\l}\,B^R_{\mathbf{n}_{t(\l)}})\prod_v\tau^{\text{EPRL}}(\eta_v)\right]~.
\end{gathered}
\end{gather}
In contrast to the EPRL-states where $\eta$ couples to the highest (lowest) weight of $\hilbert_{j^+}\otimes\hilbert_{j^-}$ the coupling here is not restricted. It is even allowed that $\eta_{s(\l)}$ and $\eta_{t(\l)}$ couple to different spins $j_{s(\l)},j_{t(\l)}\in\{|j_L-j_R|,\dots,j_L+j_R\}$. 

When fixing a time gauge, $\mathbf{n}_v\equiv\time\;\forall v\in\gamma^{(0)}$, and restricting $g_{\l}$ to the subgroup $\SU_{\time}(2)$ then \eqref{eqn:basisproj} reduce to usual $\SU(2)$-spin network functions provided that $j_{s(\l)}=j_{t(\l)}$ and vanishes otherwise. This can be easily verified by using the equivariant property of intertwiners, 
\bq
[R^{j_1}(h)]^{m_1}\!_{n_1}\;[R^{j_2}(h)]^{m_2}\!_{n_2}\;C^{j_1,n_1;j_2,n_2}_{j_3,n_3}
= C^{j_1,m_1;j_2,m_2}_{j_3,m_3}\; [R^{j_3}(h)]^{m_3}\!_{n_3}  ~,
\eq
and the normalization of Clebsch-Gordan- coefficients, $C^{j_1,m_1;j_2,m_2}_{j_3,m_3}C_{j_1,m_1;j_2,m_2}^{j'_3,m'_3}=\delta_{j_3}^{j'_3}\delta_{m_3}^{m'_3}$. More precisely, if $j_{s(\l)}=j_{t(\l)}$ then
\beq
\label{eqn:Kproj}
\begin{split}
\phi_{\gamma,j,\eta}(\{h_{\l}\})&=\prod_{\l}\sqrt{d_{j_{\l}^L}\;d_{j_{\l}^R}}\;
\tr\left[\prod_{\l}
R^{j_{\l}^L}(h_{\l})\;R^{j_{\l}^R}(h_{\l})\prod_v\tau^{\text{EPRL}}(\eta_v)\right]\\
&=\prod_{\l}\sqrt{d_{j_{\l}^L}\;d_{j_{\l}^R}}\;
\tr\left[\prod_{\l} R^{j_{\l}}(h_{\l})\prod_v\eta_v\right]~.
\end{split}
\eeq
Vice versa, kinematical states $T\in\hilbert_{0}$ to $\mathcal{K}$ can be lifted via the expansion of convolutions of $\SU(2)$ and $\Spin(4)$ in terms of characters $\chi$ and $\Theta$ respectively. Explicitly,
\begin{gather}
\label{eqn:lift}
\begin{gathered}
\left[L\;T_{\gamma,j,\eta}\right](\{g_{\l}\},\{\mathbf{n}_v\}):=\\
\prod_{\l} \left[N_{\l}\sum_{\tilde{j}_{\l}}\, \int\limits_{\SU(2)}\!\!\!\dif{h_{\l}}\;\dif{k_{\l}}\;\overline{\chi^{\tilde{j}_{\l}}(k_{\l}h_{\l})}\;\Theta^{\tilde{j}^L_{\l},\tilde{j}_{\l}^R}(B_{\mathbf{n}_{s(\l)}}^{-1}\,g_{\l}\, B_{\mathbf{n}_{t(\l)}}h_{\l})\right]T_{\gamma,j,\eta}(\{k_{\l}\})
\end{gathered}
\end{gather}
with some normalization constant $N_{\l}$. Below, we are only interested in the case where $j^L/j^R$ are determined by the simplicity constraint $j^{L/R}=j^{\pm}$. 
\para{Relation between $\mathcal{K}$, $\hilbert^{EPRL}$ and $\hilbert^{KKL}$}
When the normals are fixed but the group elements left arbitrary then the states \eqref{eqn:basisproj} are obviously basis states of $\hilbert^{EPRL}_{\gamma}$. Integrating  over the normals yields states in $\hilbert^{KKL}_{\gamma}$,
\beq
\begin{split}
&\int\limits_{S^3}\prod_{v\in\gamma^{(0)}}
\dif{\mathbf{n}_v}\;\phi_{\gamma,j,\eta}(\{g_{\l}\}, \{\mathbf{n}_v\})\\
&:=\int\limits_{\Spin(4)}\prod_{v\in\gamma^{(0)}}\dif{B_v}\;\phi_{\gamma,j,\eta}(\{(B_{s(\l)})^{-1}g_{\l}\;B_{t(\l)}\})\\
&=\tr\left\{\prod_{\l\in\gamma^{(1)}}\sqrt{d_{j_{\l}^+}d_{j_{\l}^-}}
R^{j_{\l}^+}(g^+_{\l})\; R^{j_{\l}^-}(g^-_{\l})
\prod_{v\in\gamma^{(0)}}\left[\sum_{\iota^+_v,\iota^-_v} 
f^{\eta_v}_{\iota^+_v\iota^-_v}\;\iota^+_v\otimes\iota^-_v\right]\right\}\;,
\end{split}
\eeq
which shows that
\begin{gather}
\label{eqn:nnProj}
\begin{gathered}
P_{\gamma}:\hilbert_{\inv,\gamma}\to\hilbert^{KKL}_{\gamma}\\
[P_{\gamma}\psi_{\gamma}](\{g_{\l}\})=\int\prod_v\dif{\mathbf{n}_v}\;[L\psi_{\gamma}](\{g_{\l}\},\{\mathbf{n}_v\})\;.
\end{gathered}
\end{gather}
defines an isomorphism between $\hilbert_{0,\gamma}$ and $\hilbert_{\gamma}^{KKL}$. For instance a normalized state in $\hilbert_{0}$ (see \eqref{eqn:gauginv}) is lifted to
\begin{align}
\begin{split}
&\left[L\;T_{\gamma,j,\eta}\right](\{g_{\l}\},\{\mathbf{n}_v\})=
\prod_{\l}\frac{N_{\l}\sqrt{d_{j_{\l}}}}{d_{j_{\l}}^2}\\
&\times\tr\left(\prod_{\l} R^{j^+_{\l}}((B^+_{s(\l)})^{-1}g^+_{\l}\,B^-_{t(\l)})
R^{j_{\l}^-}((B^-_{s(\l)})^{-1}g^-_{\l}\,B^-_{t(\l)})\prod_v\tau^{\text{EPRL}}(\eta_v)\right)~.
\end{split}
\end{align}
and afterwards projected 
 \begin{align}
 \label{eqn:nProj}
 \begin{split}
&[P_{\gamma}T_{\gamma,j,\eta}](\{g_{\l}\})=\left(\prod_{\l}\frac{N_{\l}\sqrt{d_{j_{\l}}}}{d_{j_{\l}}^2}\right)
 \int\limits_{\Spin(4)}\!\!\!\left(\prod_v\dif B_v\right)\\
&\times\tr\left(\prod_{\l} R^{j^+_{\l}}((B^+_{s(\l)})^{-1}g^+_{\l}\,B^-_{t(\l)})
R^{j_{\l}^-}((B^-_{s(\l)})^{-1}g^-_{\l}\,B^-_{t(\l)})\prod_v\tau^{\text{EPRL}}(\eta_v)\right)\\
&=\left(\prod_{\l}\frac{N_{\l}}{(d_{j_{\l}}d_{j^+_{\l}}d_{j^-_{\l}})^{\frac{3}{2}}}\right)\tilde{T}^{KKL}(\{g_{\l}\})
\end{split}
\end{align}
to an orthonormal state $\tilde{T}^{KKL}:=(\prod_{\l} \sqrt{d_{j_{\l}^+}\;d_{j_{\l}^-}}) T^{KKL}$ (w.r.t. \eqref{eqn:scalarKKL}). 
\section{Spin foam Projector}
\label{sec:projector}
\subsection{The general idea}
\label{ssec:road}
Similar to many other constraint systems, zero does not lie in the point spectrum of the GR-constraints so that one has to search weak rather than strong solutions. Given a family of constraints $(\hat{C}_I)_{I\in\mathcal{I}}$ a weak solution $L\in\mathcal{D}^{\ast}_{\phys}\subset\hilbert_{0}$ is an element in the algebraic dual of a dense domain $\mathcal{D}_{0}\subset\hilbert_{0}$ for which 
\beq
\left[(\hat{C}_I)^{\ast}L\right](f):=L(\hat{C}_I^{\dagger}f)=0
\eeq
holds for all $I\in\mathcal{I}$ and $f\in\mathcal{D}_{0}$. In this equation $(\hat{C}_I)^{\ast}$ refers to the dual operator acting on $\mathcal{D}^{\ast}_{0}$ and $\hat{C}^{\dagger}$ to the hermitian adjoint acting on $\hilbert_{0}$. For physical measurements and interpretation $\mathcal{D}^{\ast}_{\phys}$ must be equipped with a scalar product. Unfortunately, it is not possible to naively use the kinematical product $\escal_{0}$ since $L$ is generically not in the topological dual. Instead, assume that $\mathcal{D}^{\ast}_{\phys}$ is the algebraic dual of a dense subspace $\mathcal{D}_{\phys}$ of a Hilbert space $\hilbert_{\phys}$ whose scalar product $\scal{\cdot,\cdot}_{phys}$ can be constructed by an anti-linear (rigging) map \footnote{For more details on the construction of a rigging map see e.g. \cite{lqgcan2} and references therein.}   
\beq
\label{eqn:rig}
\eta:\mathcal{D}_{0}\to\mathcal{D}^{\ast}_{0}
\eeq
such that
\beq
\scal{f|f'}_{\phys}:=\scal{\eta[f]|\eta[f']}_{0}:=\eta[f](f')\qquad f,f'\in\mathcal{D}_{0}~.
\eeq
If this rigging map exists then $\hilbert_{\phys}$ is the completion of  $\mathcal{D}_{\phys}:=\eta(\mathcal{D}_{0})/\text{ker}(\eta)$. For well-behaved systems $\{C_I\}$ (closed, locally compact Lie-group) a rigging map can be constructed by exponentiating the constraints 
\beq
\label{eqn:feyn1}
[\eta(f)](f')=\int_T\dif{\mu(T)}\scal{\exp{i t^I\hat{C}^I} f,f'}_{0}
\eeq
with multipliers $(t^I)_{I\in\mathcal{I}}\in T$ and a suitable invariant measure $\mu(T)$. Thus, a rigging map solves  two problems in one stroke: it projects on the subspace of solutions and defines a scalar product. 

For closed finite constraint systems a rigging map always exist. But the constraints in GR do not generate a Lie-algebra but a Lie-algebroid and it is not clear that the above procedure can be applied. Nevertheless, it is often emphasized that spin foams could provide such a rigging map even though one starts with a different action and constraint algebra and therefore with a different symplectic structure (see e.g. \cite{Alexandrov:2010un}). Ignoring these problems we want to take a rather naive point of view and regard spin foams as a computational algorithm to construct a projector onto, or at least into, the physical Hilbert space.

The spin foam partition function $Z[\kappa]$ is often interpreted heuristically as the evaluation of a two dimensional `Feynman diagram' $\kappa$ appearing in the transition amplitude   
\beq
\label{eqn:feyn2}
\int\dif{N}\;\scal{T_{s_1}|\exp{i N\hat{H}}|T_{s_2}}``="\sum_{\kappa:\gamma_1\to \gamma_2}\scal{T_{s_1}|Z^{\dagger}[\kappa]|T_{s_2}}~.
\eeq 
with spin nets $s_{1/2}=(\gamma_{1/2},j_{1/2},\iota_{1/2})$. The reason for taking the adjoint spin foam amplitude $Z^{\dagger}[\kappa]$, that is, the amplitude $Z[\kappa^{\ast}]$ associated to the complex $\kappa^{\ast}$ obtained from $\kappa$ by reversing all internal face and edge orientations, is that it is more convenient in the later to interpret $T_{s^1}$ as the ingoing spin net. For the moment this is just a mere convention.  

The sum on the right hand side of equation \eqref{eqn:feyn2} presumes the existence of a tool to identify semianalytic graphs in $\hilbert_{0}$ with p.l. graphs in the boundary of $\kappa^{\ast}$. This issue will be discussed at length in the next subsection, for the discussion below it suffice to assume that such foams exist. More precisely, we assume that the boundary spin net $T_{\partial\kappa^{\ast},j',\iota'}$ of $\kappa^{\ast}$ can be identified with the $\SU(2)$-nets $T_{s'_1}\otimes T^{\dagger}_{s'_2}$ so that 
\beq
\scal{T_{s_1}|Z^{\dagger}[\kappa]|T_{s_2}}
=\sum_{\{j_f\},\{\iota_v\}}\prod_{v\in\kappa_{int}} \mathcal{A}_v\prod_{e\in\kappa} \mathcal{A}_e\prod_{f\in\kappa}\mathcal{A}_f\;
 \scal{T_{s_1},T_{s'_1}}\;
\scal{T_{s'_2},T_{s_2}}~.
\eeq
where vertex, edge and face amplitude, $\mathcal{A}_v$, $\mathcal{A}_e$ and $\mathcal{A}_f$, depend on the model. In particular, we can either work with the first approach (\secref{ssec:EPRL}) and project the states in $\hilbert^{EPRL}$ to $\hilbert_{0}$ by restricting the group elements on $\partial\kappa$ to $\SU(2)$-elements or work with the KKL-proposal and projective spin nets. Since $T^{KKL}$ are manifestly covariant we prefer the second. However, it should be kept in mind that one could equivalently define the model on spin foams of the EPRL kind.

In analogy to \eqref{eqn:feyn1}, one can now postulate a rigging map
\begin{gather}
\begin{gathered}
\label{eqn:project}
\eta_{\gamma_1,\gamma_2}: \hilbert_{0,\gamma_1}\to \hilbert^{\ast}_{0,\gamma_2} \\
\eta_{\gamma_1,\gamma_2}[T_{s_1}](T_{s_2})=\sum_{\kappa:\gamma_1\to \gamma_2}\scal{T_{s_1}|Z^{\dagger}[\kappa]|T_{s_2}}\;.
\end{gathered}
\end{gather}
The state $\eta[T_{s}]$ is clearly distributional and, thus, an element of the algebraic rather than the topological dual as $Z^{\dagger}[\kappa]$ includes an infinite sum over all labelings. If $\eta$ is a proper rigging map then it should satisfy
\beq
\label{eqn:constraint-proj}
\eta[T_{s_1}](\hat{H}T_{s_2})=\sum_{s_m}\;\sum_{\kappa:\gamma_1\to \gamma_m}\scal{T_{s_1}|Z^{\dagger}[\kappa]|T_{s_m}}\scal{T_{s_m}|\hat{H}|T_{s_2}}=0
\eeq
for all $T_{s_1},T_{s_2}\in\hilbert_{0}$. The sum over all intermediate spin nets $s_m$ including a sum over all possible graphs $\gamma_m$ seems to be ill-defined since the kinematical Hilbert space of LQG is not separable. Even graphs which only differ sightly in their shape and not in their combinatorics give rise to orthogonal spin nets and thus are to be considered inequivalent. Nevertheless, only finitely many summands of \eqref{eqn:constraint-proj} will be non-zero and this problem is avoided.

Of course, one could also include a weight $w(\kappa)$ in \eqref{eqn:project} as it is generated in GFT \cite{Freidel:2005qe,Oriti:2006se,Oriti:2011jm}. But since we take all possible 2-complexes into account, not only those ones dual to a simplicial triangulation, the relation to GFT cannot be made precise and we would need to introduce $w(\kappa)$ in an ad hoc fashion. At the end of section \ref{ssec:ttproj} we will discuss this issue in more detail. For the time being, we  choose the easiest option and assume $w(\kappa)=1$ for all $\kappa$. But note that the conclusion of the present article is unaffected by any weight function that satisfies the natural gluing condition $w(\kappa_1\sharp \kappa_2)=w(\kappa_1)\; w(\kappa_2)$.
\\[5pt]
To define $\eta$ in equation \eqref{eqn:project} precisely requires more work than just evaluating the amplitudes of \eqref{eqn:BHKKL1}. Apart from the fact that the states in $\hilbert_{0}$ must be lifted to $\hilbert^{KKL}$ so as to match the states induced on $\partial\kappa$ we postulate some reasonable properties the rigging map should obey. 
\begin{enumerate}
\item The map $\eta_{\gamma_1,\gamma_2}$ formally decomposes into a sum of operators $\hat{Z}^{\dagger}[\kappa]:\hilbert_{0,\gamma_1}\to\hilbert_{0,\gamma_2}$ whose matrix elements are proportional to the spin foam amplitude \eqref{eqn:BHKKL1}.
\item The operator $\hat{Z}^{\dagger}[\kappa_0]$ based on the trivial evolution (see \defref{def:trivevol}) defines an isometry$P_{\gamma}$  between $\hilbert_{0,\gamma}$ and $\hilbert^{KKL}_{\gamma}$ such that $\hat{Z}^{\dagger}[\kappa_0]=P_{\gamma}^{\dagger}P_{\gamma}$.
\item $\hat{Z}^{\dagger}[\kappa]$ respects the equivalence relations of spin networks.
\item Splitting of internal edges and faces should leave $\hat{Z}^{\dagger}[\kappa]$ invariant.
\item Let $\kappa_1$ and $\kappa_2$ be 2-complexes such that $\kappa_1\cap\kappa_2=\partial\kappa_1\cap\partial\kappa_2=\tilde{\gamma}$ then 
\beq
\label{eqn:glue}
\sum_{\tilde{T}\in\hilbert_{0,\tilde{\gamma}}}\scal{T_{i}|\hat{Z}^{\dagger}[\kappa_1]|\tilde{T}_{\tilde{\gamma}}}\scal{\tilde{T}_{\tilde{\gamma}}|\hat{Z}^{\dagger}[\kappa_2]|T_{f}}=\scal{T_{i}|\hat{Z}^{\dagger}[\kappa_2\sharp\kappa_1]|T_{f}}
\eeq
where $\kappa_2\sharp\kappa_1$ is the 2-complex obtained by gluing along the common graph $\tilde{\gamma}$ and $T_{i},T_{f}$ are spin network functions living on the boundary graph of $\kappa_2\sharp\kappa_1$.
\end{enumerate}
The first point captures the details of the above argument and the second point is motivated by the heuristic interpretation of foams being two dimensional Feynman graphs. From this point 
of view every internal vertex corresponds to the action of $\hat{H}$ and consequently $\kappa_0$ represents the zeroth order in $\exp{N\hat{H}}\approx \mathbf{1}+\cdots$. Thereafter $\scal{T|\hat{Z}^{\dagger}[\kappa_0]|T'}$ should represent the kinematical inner product which imposes the second property.
  
The third requirement is necessary in order to construct a self-consistent operator. Two spin nets are equivalent if they can be obtained by the following manipulations
\begin{list}{(\alph{enumi})}{\usecounter{enumi}}
\item adding new links labeled by the trivial representation
\item creating a new node labeled by the trivial intertwiner by splitting a link.
\end{list}
Since every face touching $\partial\kappa$ contributes a link in the boundary graph $\hat{Z}^{\dagger}[\kappa]$ should be invariant if we add or remove a face labeled by the trivial representation. If we split link in $\partial\kappa$ then also the adjacent face must be subdivided by a new internal edge. Therefore, the spin foam amplitude should be invariant under such splittings and also under the trivial subdivision of internal edges because it does not play a role whether the new edge $e$ splits another internal edge or joins an internal vertex. Furthermore, the model should be independent of the way a semianalytic graph is approximated (see below).  

The last condition reflects the gluing property of spin foam amplitudes $Z^{\dagger}[\kappa_1]Z^{\dagger}[\kappa_2]=Z^{\dagger}[\kappa_2\sharp\kappa_1]$ used in most models in order to fix the boundary amplitude. Furthermore, if \eqref{eqn:project} defines an improper projector\footnote{Generically 
$\eta$ will have no square, that is, constant $K$ will actually be infinite.} then $\eta$ should satisfy $\eta[\eta[T]]=K\eta[T]$ for a constant $K>0$.
\subsection{Abstract versus embedded setting}
\label{ssec:absvembed}
In the last section we discussed the general idea how to combine the covariant and the canonical approach. Even though the states induced on the boundary of a spin foam are formally equivalent to spin net states on the same graph, this does not prove equivalence of both theories. Due to the structural difference of both models, it is, for example, not clear that observables agree. Also the construction of the maps \eqref{eqn:feyn1} and \eqref{eqn:feyn2} is only formal since the correct measure of this path integral is unknown. In this section we will argue that a strict derivation of \eqref{eqn:feyn2} from BF-theory is not possible if one insists that $\kappa$ is dual to a triangulation of space-time. Essentially, this is caused by the different topological and geometrical meaning of graphs in the canonical and covariant model and will be discussed in the first subsection. In the second part we will analyze the impact of a rigging map as postulated in \eqref{eqn:project} on the canonical theory focussing on the role of diffeomorphisms.                     
\subsubsection{Triangulations, foams and graphs}
\label{sec:triangulation}

The first obvious obstacle when trying to combine covariant and canonical theory is that the canonical model is based on semianalytic paths instead of p.l. 1-cells. Nevertheless, one can always approximate a semianalytic path by p.l. ones. That is another important reason why we ask for invariance under trivial face splittings so that the `transition function' is independent of the approximation. Of course, it is not really possible to approximate spin nets defined on semianalytic graphs by spin nets on p.l. graphs since the Ashtekar-Lewandowski-measure is maximally clustering in the sense that any two spin nets are orthogonal as soon as they are defined on slightly different graphs. 
Thus one should either modify canonical LQG to accommodate p.l. structures or one eventually interprets the boundary graphs of spin foam models in the semianalytic category.

Moreover, the links of a spin net in $\hilbert_0$ can be knotted so that the `time-evolution' $\gamma\times[0,\epsilon]$ could lead to complicated self-intersections of faces. On the other hand, the Hamiltonian acts locally on the nodes and the physical impact of knotting is barely understood anyway so that we will restrict to unknotted links \footnote{The knotting class of the node can be still non-trivial. See the next section for more details}. 
  
Another problem that occurs when trying to match p.l. and s.a. graphs is the following: A p.l. cell is defined as the convex hull of its vertices and therefore completely determined by them. Yet, there are infinitely many possibilities how to glue a s.a. link between two nodes and thus several links can be glued between the same nodes. This is not possible for p.l. links. 

To summarize the previous argument: P.l. complexes are to restrictive for the purpose of defining a rigging map but we also do not want to give up all the nice properties worked out before. A way out of this dilemma is to use ball complexes as in section \ref{sec:trian} or a more combinatorial definition:
\begin{definition}\mbox{}
\label{def:abstract}
\begin{itemize}
\item An \underline{abstract} n-cell $c$ is an n-ball whose frontier is the finite union of lower dimensional balls (faces).
\item An \underline{abstract} n-complex $\cal C$ is a finite collection of $m$-balls, $m\leq n$, containing at least one $n$-cell. If $A\in{\cal C}$ then also all faces of $A$ are in $\cal C$. If $A,B\in{\cal C}$ then either $A\cap B=\emptyset$ or $A\cap B$ is a common face of $A$ and $B$.
\end{itemize}
\end{definition} 
All definitions and theorems of section \ref{sec:2-compl} can be immediately generalized by replacing `p.l.' through `abstract'. Indeed, we only give up convexity and linearity and since balls are path connected there exists subdivisions $\cal C'$ of $\cal C$ that are combinatorially equivalent to a p.l. complexes (compare with theorem \ref{theorem1}). 

One might wonder why we are putting so much effort in adapting foams to graphs and do not simply restrict the class of graphs used in the canonical theory to those which are dual to a triangulation of the hypersurface $\Sigma$. A technical reason for this is that the Hamiltonian constraint, as defined in \cite{Thiemann96a}, creates trivalent nodes that cannot be dual to a 3-dimensional polyhedron. Obviously, this can be avoided by using a different regularization, e.g. \cite{EmanueleReg}, but the only known parametrization, which leads to a non-anomalous Hamiltonian, is the original one \cite{Thiemann96a}.

Despite this more technical arguments, there are also severe reasons why the class of graphs should not be restricted in the canonical model that are deeply rooted in the different treating of geometry and topology in both theories.
When quantizing the canonical theory we start with the configuration space $\mathcal{A}$ that is the space 
of connections on a principal bundle $P(\Sigma,G) $ with base manifold $\Sigma$ and gauge group $G$. This space can be embedded  into the set of homomorphisms $\mathrm{Hom}(\mathcal{P},G)$ from the groupoid of paths $\mathcal{P}$ on $\Sigma$ to $G$ \cite{lqgcan2}. In fact $\mathrm{Hom}(\mathcal{P},G)$ defines the space of generalized connections $\overline{\mathcal{A}}$ which is used to construct the gauge variant kinematical Hilbert space $\hilbert_{\kin}=\sqi{\overline{\mathcal{A}},\mu_{AL}}$. This space is spanned by spin net functions on \emph{all} possible graphs build by glueing elements in $\mathcal{P}$, not only those ones which are dual to a triangulation. Moreover, the holonomy flux algebra does not preserve the underlying graph of a spin net and, therefore, also the span of spin net functions based on dual graphs is not preserved.

Given \emph{all} holonomies along \emph{all} paths in $\Sigma$ one can reconstruct the connection. The set $\mathrm{Hom}(\mathcal{P},G)$ also captures topological information since it can be related to the fundamental group of $\Sigma$ (see e.g. \cite{Analysis:CQ}). Again, this information cannot be captured by a single graph $\gamma$, i.e a finite collection of paths. 

The situation changes fundamentally when $\gamma$ is dual to a non-degenerate triangulation $\Delta$ of $\Sigma$. As proven by Whitehead \cite{whitehead:1940}, $\Delta$ is uniquely determined up to p.l. homeomorphisms. Astonishingly, it can be shown that in three dimension also every p.l. and every topological manifold have a unique differentiable structure up to diffeomorphisms. In other words in three dimensions the topological (TOP), piecewise linear (PL) and smooth (DIFF) category are equivalent. The equivalence of PL and DIFF was proven independently by Smale \cite{Smale1959}, Munkres \cite{Munkres1960} and Hirsch \cite{Hirsch1963} and the equivalence of TOP and DIFF by Moise \cite{Moise1977}. A triangulation also allows to partly reconstruct a metric by defining edge length and angels at each vertex of $\Delta$. 

In this sense, a graph $\gamma_{\Delta}$ dual to a triangulation captures much more topological and geometric information than an arbitrary graph. For example, closed graphs can be only dual to the triangulations of a closed (compact, without boundary) manifold. But to ensure gauge invariance the underlying graph of a spin net must be closed. By a theorem of Milnor \cite{Milnor1962} any compact 3-dimensional manifold $\Sigma$ can be uniquely decomposed into a finite number of \emph{prime manifolds} $\Sigma_i$. A compact 3-manifold is said to be prime if it is either $S^2\times S^1$, a non trivial bundle over $S^1$ with fibers homeomorphic to $S^2$ (similar to the Hopf bundle) or every 2-sphere bounds a 3-ball in $\Sigma_i$; two prime manifolds are glued together by removing a 3-ball and identifying the newly generated boundaries. Thus any graph dual to triangulation of a compact subregion in $\Sigma$ must be either represent a prime factor of $\Sigma$, or a product thereof or must be dual to a discretized 3-ball (tetrahedron). Yet, a graph dual to a 3-ball is certainly not closed and thus the boundary graph of the associated foam would contain edges that are not embedded in $\Sigma$. This shows that any graph dual to a triangulation of a region in a spatial hypersurface and bordering a foam must be related to a prime factor\footnote{A method to analyze the relation between combinatorial graphs and triangulations is crystallization and leads to colored graphs as they are used in colored GFT \cite{Gurau:2009tw}}. 

Apart from that, taking the idea of the rigging map seriously, the spin foam 
`projector'  should be based on $\kappa$ which is dual to a discretization of the foliation  
$\mathbb{R}\times \Sigma$. However, the resulting dual foam $\kappa$ is not obviously 
a discrete foliation into {\it the same} discretized leaves. 
All of these difficulties suggest to work with arbitrary abstract foams that do not 
originate as the dual of an embedded discretization of $M$.  
\subsubsection{Semianalytic, piecewise analytic and abstract}
\label{ssec:spapl}
In the following, we will discuss how one can realize \eqref{eqn:project} by using abstract complexes in the sense of definition \ref{def:abstract} while graphs are still embedded in $\Sigma$.

Due to technical reasons, one prefers to work with semianalytic diffeomorphisms $\mathrm{Diff}_{sa}(\Sigma)$ which are analytic except on some semianalytic submanifolds where they are of class $C^{(n)}$, $n>0$. It was also suggested in \cite{fairbairn-2004-45} to use instead piecewise analytic diffeomorphism, i.e. functions which are almost everywhere analytic except for a finite set of points where they are continuous but not necessarily differentiable.

A diffeomorphism $\phi$ acts on spin net functions by
\beq
\hat{U}(\phi)T_{\gamma,j_{\l},\iota_{\n}}(\{g_{\l}\})=T_{\phi(\gamma),j'_{\phi(\l)},\iota'_{\phi(\n)}}(\{g_{\phi(\l)}\})
\eeq
leaving the labeling of links invariant, that is, $j'_{\phi(\l)}=j_{\l}$. Of course, $\phi$ changes the group element $g_{\l}$ since now the holonomy is taken along $\phi(\l)$ and can also modify the intertwiners by altering the ordering of nodes at $\n$. 

In the subsequent discussion, two graphs are said to be p.a.- or s.a.- equivalent if there exist a p.a./s.a. diffeomorphism $\phi$ such that $\phi(\gamma)=\gamma'$. In \cite{fairbairn-2004-45}, the authors showed that two graphs are p.a.-equivalent iff one can find a one-parameter family (ambient isotopy) of homeomorphism $h_t:\Sigma\to\Sigma$, $t\in[0,1]$ with $h_0(\gamma)=\gamma$ and $h_1(\gamma)=\gamma'$. This kind of equivalence classes is called a singular knot. These knotting classes are countable. Consequently, the $\mathrm{Diff}_{pa}$-invariant Hilbert space $\hilbert_{diff,pa}$ must be separable.

In contrast to that, the space $\hilbert_{diff,sa}$ is non-separable. Since $\phi\in\mathrm{Diff}_{sa}$ is at least $C^{(1)}$ at every point $p\in\Sigma$ the differential $D\phi(p)$ of $\phi$ at $p$ is a linear transformation in the tangent space $T_p\Sigma$. A dilatation in $T_p\Sigma$ only effects the parametrization of the integral curves $c(t)$ with $\dot{c}(0)\in T_p\Sigma$ and so we may assume w.l.o.g that $D\phi(p)\in\group{SL}(3)$ which is eight dimensional. Now an $n$-tuple of lines through $p$ in 3d is determined by $m\geq10$ angles for $n\geq 5$ and therefore the equivalence class of an $n$-valent node is labeled by $m-8$-dimensional continuous parameter, so-called moduli $\theta$. However, it can be shown that $\hilbert_{diff,sa}$ is almost the direct integral over spaces with fixed moduli $\theta$(see \cite{lqgcan2}).

As there exist no infinitesimal operator on $\mathcal{H}_{0}$ representing the classical diffeomorphism constraint, the $\mathrm{Diff}_{pa/sa}$ invariance is imposed by a rigging map 
\begin{gather}
\label{eqn:rigdif}
\begin{gathered}
\eta_D(T_s):=\eta_{[s]_D}L_{[s]_D}\\
L_{[s]_D}:=\sum_{s'\in[s]_D}\scal{T_{s'},\;\cdot\;}\in\mathcal{D}^{\ast}_{0}
\end{gathered}
\end{gather}
where, modulo technicalities \cite{lqgcan3}, $[s]_D$ is the orbit of $s=(\gamma,j,\iota)$ under diffeomorphism and the positive number $\eta_{[s]_D}$ can be fixed such that the scalar product imposed by the rigging map \eqref{eqn:rigdif} is well-defined. More in detail, $\eta_{[s]_D}$ is equal to the product of a positive number $\eta_{[\gamma(s)]_D}$ that depends only on the orbit of the graph $\gamma(s)$ underlying $s$ but so far cannot be fixed and a factor $\eta'_{[\gamma(s)]_D,[s]_D}$  
that is chosen such that the averaging in \eqref{eqn:rigdif} respects the graph symmetries of $s$ and the scalar product is sesqui-linear.  

We can proceed similarly with \eqref{eqn:project}: In the following, two embedded spin nets belong to the same abstract equivalence class $[s]_A$ if they are embeddings of the same abstract spin net $s_A$. Now, replace \eqref{eqn:rigdif} by 
\begin{gather}
\begin{gathered}
\label{eqn:riggdifH}
\eta(T_s):=\sum_{[s']_A\in N_A} \eta_{[s]_A,[s']_A} L_{[s']_A}\\
\eta_{[s]_A,[s']_A} =\sum_{\kappa_A:s_A\to s'_A} Z^{\dagger}[\kappa]
\end{gathered}
\end{gather}
with $N_A$ denoting the set of equivalence classes and 
\beq
\label{eqn:L_A}
L_{[s']_A}=\eta_{[s']_A}\sum_{\hat{s}\in [s']_A}\scal{T_{\hat{s}},\;\cdot\;}
\eeq
where $\eta_{[s']_A}$ is a positive number with similar properties as $\eta_{[s]_D}$.
This definition is advantageous regarding two aspects: First it also implements diff-invariance since $[s]_D\subset [s]_A$ and second it allows us to directly work in the abstract setting. Yet, the equivalence class $[s]_A$ is huge and \eqref{eqn:riggdifH} does not only `wash out' the embedding information but also all information about moduli or knotting classes. On the other hand, it was shown in \cite{Bahr:2011kn} that at least in the semianalytic theory the same happens when working with embedded foams. 

One might also be concerned that \eqref{eqn:riggdifH} is in conflict with the quantization of the Hamiltonian constraint $\hat{H}$. That this is not the case can be seen as follows: As many operators in the canonical setting, $H$ must be regularized such that the regularized operator $\hat{H}^{\epsilon}$ converges to $\hat{H}$ when the parameter $\epsilon$ tends to zero. This limit is taken in a weak $^\ast$ operator topology on $\mathcal{D}^{\ast}_{diff}\times\mathcal{D}$, that is $|L(\hat{H}^{\epsilon}f)-L(\hat{H}f)|<\delta$ for all $\epsilon<\delta(\epsilon)$ and $L\in\mathcal{D}^{\ast}_{diff}$, $f\in\mathcal{D}$. The limit point $\hat{H}$, which in this case can be taken as 
$\hat{H}=\hat{H}^{\epsilon_0}$ for an arbitrary but fixed choice $\epsilon_0$ of the 
regulator $\epsilon$,  is an operator on 
the kinematical Hilbert space and not on its dual. Now $\eta$ includes an averaging over spatial diffeomorphisms and thus the value of $\epsilon_0$ is irrelevant when computing the dual
action of $\hat{H}$ on the image of $\eta$.  
Notice that we do not need to define $\hat{H}$ as an operator 
that maps  the image of $\eta$ to itself\footnote{It is clear that it does not even preserve the image of 
$\eta_D$ as $\hat{H}$ is only spatially diffeomorphism covariant but not invariant.}.
We are only interested in whether its dual action annihilates the image of $\eta$.  
Finally, we must pay attention to the fact that a diffeomorphism can change the order 
of the labeling of an abstract boundary graph resulting from an ordered foam. However, the ordering of foams is just needed for the labeling by intertwiners over which one is summing in \eqref{eqn:riggdifH}. Hence, no problem appears from the diffeomorphism averaging.
\subsection{Operator Foam}
\label{ssec:operfoam}
We will now construct explicitly a spin foam operator which displays all the desired properties and is based on abstract complexes and the map \eqref{eqn:riggdifH}. Since the Hamiltonian operator does not change the moduli or knotting class one can also use an abstract graphical calculus (see \cite{Gaul:2000ba,Alesci:2011ia}) on the canonical side once these classes are fixed. Therefore, we will directly work with abstract graphs and will leave the label $A$ away. 
\begin{definition}
Let $(\kappa,\hilbert_f,Q_e)$ be an abstract spin foam whose faces are labeled by EPRL-triples $(j,j^+,j^-)$ and whose edges carry an operator $Q_e:\zeta^{KKL}(\hilbert_{e,\inv})\to\zeta^{KKL}(\hilbert_{e,\inv})$ defined in \eqref{eqn:diag}. Suppose $\partial\kappa$ is the disjoint union of an initial graph $\gamma_i$ and final graph $\gamma_f $ then
\begin{gather}
\begin{gathered}
\label{eqn:A}
\hat{Z}[\kappa]:\hilbert_{0,\gamma_i}\to\hilbert_{0,\gamma_f}\\
\scal{T_{s_f}|\hat{Z}[\kappa]|T_{s_i}}_{0}
:=\scal{P_{\gamma_f}T_{s_f}|Z[\kappa]|P_{\gamma_i}T_{s_i}}_{KKL}~.
\end{gathered}
\end{gather}
Here, $Z[\kappa]$ is the amplitude \eqref{eqn:BHKKL1} with an additional face weight $\mathcal{A}_f=d_{j_f^+}d_{j_f^-}$  and $P_{\gamma}:\hilbert_{0,\gamma}\to\hilbert^{KKL}_{\gamma}$ is the isometry \eqref{eqn:nnProj} with normalization constant $N_e=(d_{j_e}\;d_{j_e^+}\;d_{j_e^-})^{3/2}$.\footnote{In the embedded setting \eqref{eqn:A} is replaced by
$$
\scal{T_{s_f}|\hat{Z}[\kappa_A]|T_{s_i}}:=\sum_{s,s'}\sum_{\substack{s\in[s]_A\\s'\in[s']_A}}\scal{PT_{s_f}|T^{KKL}_{s_i}}\scal{T^{KKL}_{s^A_f}|\hat{Z}[\kappa_A]|T^{KKL}_{s^A_i}}
\scal{T^{KKL}_{s'}|PT_{s_f}}~.
$$}
\end{definition}
With the results of \cite{Kaminski:2009fm}, it is straightforward to show that the adjoint $\hat{Z}^{\dagger}[\kappa]$ is equal to $\hat{Z}[\kappa^{\ast}]$ and that it displays all the desired properties listed in section \ref{ssec:road}:
\para{a. Subdivision of edges}
If $e\in\kappa_{int}$ is an internal edge which is subdivided into $e_1,e_2$ by a vertex $v_0=t(e_1)=s(e_2)$ then
\beq
\mathcal{A}_{v_0}(\iota_{t(e_1)}^{\pm},\iota_{s(e_2)}^{\pm})=\left(\iota_{t(e_1)}^+\otimes \iota_{t(e_1)}^-|\iota_{s(e_2)}^+\otimes \iota_{s(e_2)}^-\right)=\delta^{\,\iota_{t(e_1)}^+}_{\;\;\iota_{s(e_2)}^+} \delta^{\,\iota_{t(e_1)}^-}_{\;\;\iota_{s(e_2)}^-}
\eeq
and therefore
\beq
\begin{split}
&\sum_{\iota_{t(e_1)}^{\pm}}\sum_{\iota_{s(e_2)}^{\pm}}\mathcal{A}_{e_1}(\iota_{s(e_1)}^{\pm},\iota_{t(e_1)}^{\pm}) \mathcal{A}_{v_0}(\iota_{t(e_1)}^{\pm},\iota_{s(e_2)}^{\pm})\mathcal{A}_{e_2}(\iota_{s(t_2)}^{\pm},\iota_{t(e_2)}^{\pm})\\
&=\sum_{\iota_{v_0}^{\pm}}\sum_{\iota_{v_0},\iota'_{v_0}} \overline{f^{\;\iota_{s(e_1)}}_{\iota_{s(e_1)}^+\iota_{s(e_1)}^-}}\;\;(Q_{e_1})^{\iota_{s(e_1)}}_{\iota_{v_0}}\;\;\underbrace{f^{\iota_{v_0}}_{\iota_{v_0}^+\iota_{v_0}^-}
\overline{f^{\;\iota'_{v_0}}_{\iota_{v_0}^+\iota_{v_0}^-}}}_{H^{\,\iota_{v_0}}_{\;\;\iota'_{v_0}}}\;\;(Q_{e_2})^{\iota'_{v_0}}_{\iota_{t(e_2)}}\;\;f^{\iota_{t(e_2)}}_{\iota_{t(e_2)}^+\iota_{t(e_2)}^-}\\
&= \overline{f^{\;\iota_{s(e_1)}}_{\iota_{s(e_1)}^+\iota_{s(e_1)}^-}}\;\;(Q_{e})^{\iota_{s(e_1)}}_{\iota_{t(e_2)}}\;\;f^{\iota_{t(e_2)}}_{\iota_{t(e_2)}^+\iota_{t(e_2)}^-}=\mathcal{A}_{e}(\iota_{s(e)}^{\pm},\iota_{t(e)}^{\pm})~.
\end{split}
\eeq
where $\iota$ and $\iota^+\otimes \iota^-$ are normalized $\SU(2)$ respectively $\Spin(4)$ invariants assigned to $e_0$. This proves that $\hat{Z}[\kappa]$ is invariant under a subdivision of edges.\\ 
\para{b. Subdivision of faces}
A colored subdivision of a face $f$ can be obtained by joining two vertices in $\dot{f}$ by an edge $e_0$ lying in the interior of $f$. The sub-faces $f_1,f_2$ adjacent to $e_0$ inherit the coloring and orientation of $f$ so that the $\SU(2)$ intertwiner space attached to $e_0$ is $\Inv\left(\hilbert_{j_{f_2}}^{\ast}\otimes \hilbert_{j_{f_1}}\right)$. Since this space is one-dimensional the edge amplitude $\mathcal{A}_{e_0}$ reduces to the identity. However, the edge $e_0$ also gives rise to a splitting of the vertex boundary graphs at its source $s$ and target $t$ by splitting the link $e(f)\in\gamma_{s/t}$ associated to $f$. Hence, we have to insert the unique two-valent intertwiners
\beq
\left(\epsilon^{\,m_{1}}_{\;\;m_{2}}\right)^{j_1}_{j_2}=\frac{1}{\sqrt{d_{j_1}}}\left(\delta^{\,m_{1}}_{\;\;m_{2}}\right)\delta^{j_{1}}_{j_{2}}
\eeq
into the vertex amplitude
\beq
\label{eqn:facef}
\begin{split}
\mathcal{A}^{\pm}_{s/t}&=
\tr\left(\cdots \;\;(\iota^{\pm}_{e_1})^{\,\dots\dots}_{\;\;\dots m_{f_1}^{\pm}\dots}\;\;\epsilon^{m_{f_1}^{\pm}}_{\;\;\;m_{f_2}^{\pm}}\;\;
(\iota^{\pm}_{e_2})_{\;\;\dots\dots}^{\,\dots m^{\pm}_{f_2}\dots}\;\;\cdots\right)\\
&=\frac{1}{\sqrt{d_{j_f^{\pm}}}}\;\tr\left(\cdots\;\; (\iota^{\pm}_{e_1})^{\,\dots\dots}_{\;\;\dots m^{\pm}_{f}\dots}\;(\iota^{\pm}_{e_2})_{\;\;\dots\dots}^{\,\dots m^{\pm}_{f}\dots}\;\;\cdots\right)~.
\end{split}
\eeq
Here, $e_1\in\dot{f}$ and $e_2\in\dot{f} $ are the unique edges meeting at $s/t$  that also bound $f_1$ and $f_2$ respectively. Let $\kappa'$ be the complex obtained from $\kappa$ by such a subdivision then $\hat{Z}[\kappa']=d_{j^+_f}\;d_{j_f^+}\hat{Z}[\kappa]$ due to \eqref{eqn:facef}. To restore invariance under face splitting one needs to introduce a face amplitude $\mathcal{A}_f$ for which 
\beq
\mathcal{A}_{f_1}\mathcal{A}_{f_2}\frac{1}{d_{j^+_f}\;d_{j_f^+}}= \mathcal{A}_f\;.
\eeq
\para{c. Gluing and resolution of the identity}
Suppose $\kappa_1$ and $\kappa_2$ are foams whose boundaries $\partial\kappa_{1/2}=\gamma^i_{1/2}\cup \gamma^f_{1/2}$ decompose each into one final ($f$) and one initial ($i$) graph with $\gamma_1^f\cong\gamma^i_2\cong\gamma$. Recall, that the states $T^{KKL}$ induced on the boundary graph are not normalized and all internal edges adjacent to a final graph are incoming thus 
 \beq
 \begin{split}
& \left.\dots|\hat{Z}[\kappa_1]|T_{\gamma,j_{\l},\iota_{\n}}\right\rangle\\
& =\sum_{j'_f,\iota'_e}\dots \left(\prod_{\l\in\gamma^{(1)}}\;\mathcal{A}_{f_{\l}}\right)\prod_{\n\in\gamma^{(0)}}\overline{f^{\;\iota_{s(e_{\n})}}_{\iota_{s(e_{\n})}^+\iota_{s(e_{\n})}^-}}\;\;(Q_{e_{\n}})^{\iota'_{s(e_{\n})}}_{\iota'_{\n}}\scal{T^{KKL}_{\gamma,j'_{\l},\iota'_{\n}}|PT_{\gamma,j_{\l},\iota_{\n}}}_{KKL}\\
&=\dots\sum_{j'_f,\iota'_e}\dots\prod_{\l}\mathcal{A}_{f_{\l}}\frac{\delta_{j'_{\l},j_{\l}}}{\sqrt{d_{j^+_{\l}}d_{j^-_{\l}}}}\prod_{\n}\overline{f^{\;\iota_{s(e_{\n})}}_{\iota_{s(e_{\n})}^+\iota_{s(e_{\n})}^-}}\;\;(Q_{e_{\n}})^{\iota'_{s(e_{\n})}}_{\iota'_{\n}} \delta^{\,\iota'_{\n}}_{\;\;\iota_{\n}}
 \end{split}
\eeq
where $f_{\l}$ is the unique face containing $\l\in\gamma^{(1)}$ and $e_{\n}$ the edge adjacent to $\n\in\gamma^{(0)}$. When $\kappa_1$ and $\kappa_2$ are glued (see \eqref{eqn:glue}) along the spin net $s=(\gamma,j_{\l},\iota_{\n})$ this implies 
\beq
\sum_{s}
 \scal{T_{s^f_2}|\hat{Z}[\kappa_2]|T_{s}}
 \scal{T_{s}|\hat{Z}[\kappa_1]|T_{s^i_1}}
 =\prod_{\l\in\gamma}\frac{\mathcal{A}_{f_{\l}}}{d_{j^+_{\l}}d_{j^-_{\l}}}\;
  \scal{T_{s^f_2}|\hat{Z}[\kappa_2\sharp\kappa_1]|T_{s^i_1}}~.
\eeq
The foam $\kappa_2\sharp\kappa_1$ is the complex which arises when $\gamma_1$ and $\gamma_2$ are identified and then removed. More precisely, the faces $f^1_{\l}\in\kappa_1$ and $f^2_{\l}\in\kappa_2$ are combined to one face in $\kappa_2\sharp\kappa_1$ which produces an excess face amplitude. Concluding,  if the face weight is fixed to $\mathcal{A}_f=d_{j^+_f}d_{j^-_f}$ then the amplitude is invariant under face splittings and obeys a gluing property. By an analogue computation one can also show that 
\begin{gather}
\begin{gathered}
\scal{T_{\gamma,j_{\l},\iota_{\n}}|\hat{Z}[\kappa_0]|T_{\gamma,j'_{\l},\iota'_{\n}}}
=\scal{(PT)_{\gamma,j_{\l},\iota_{\n}}|(PT)_{\gamma,j'_{\l},\iota'_{\n}}}_{KKL}\\
=\scal{T_{\gamma,j_{\l},\iota_{\n}}|T_{\gamma,j'_{\l},\iota'_{\n}}}_{0}
\end{gathered}
\end{gather} 
for the trivial evolution $\kappa_0$.
\para{d. Equivalence classes}
Subdivisions and adding faces/edges labeled by the trivial representation define equivalence relations on foams/spin nets. Since they leave the amplitude/spin net function invariant one should only sum over equivalence classes in \eqref{eqn:project} and \eqref{eqn:riggdifH}. If not stated otherwise it will be always assumed that a foam/graph is minimal in the following sense:
\begin{definition}
An abstract foam/graph is called \underline{minimal} iff it cannot be obtained from another foam/graph by  subdivisions.
\end{definition}
Note, whether an abstract foam/graph is minimal does not depend on the coloring. Given a generic foam a minimal one can be obtained by successively removing 2-valent internal edges and 2-valent vertices (internal as well as external). However, not all 2-valent edges can be removed since it might happen that the removal of an edge generates a self-intersecting surface which is not homeomorphic to a 2-ball (see \figref{fig:cone} for an example). This also shows that the minimal representatives of the equivalence classes are not unique. But since the model is independent of this choice we can safely fix a minimal representative for each equivalence class in the following. Furthermore, trivial representations will be excluded as before.
\begin{figure}[t]
 \begin{center}
	\subfloat{\includegraphics{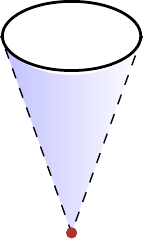}}
	\hspace{2.5 cm}
	\subfloat{\includegraphics{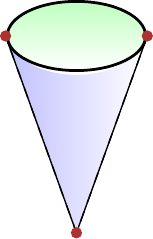}}
 \end{center}
\caption{A cone does not define an abstract complex since the face is not homeomorphic to a 2-ball. However, if it is split by two edges, as on the right hand side, it is a minimal abstract complex.}
\label{fig:cone}
\end{figure}
\section{Does the spin foam projector provide a rigging map onto $\hilbert_{\phys}$?}
\label{sec:main}
Apart from technical issues a first test on $\eta$ is to check whether the constraints are really annihilated. By construction the gauss and diffeomorphism constraint are obviously satisfied, but the Hamiltonian constraint is not. To prove this we will first develop a method to split foams into basic building blocks. The properties of the so-defined rigging map will be discussed in the sequel.
\subsection{Time ordering}
\label{ssec:time}
The rigging map $\eta$ is naturally distinguishing between in and out-going spin nets which induces an order of the internal vertices:
\begin{definition}\mbox{}\\
\label{def:time}
Suppose $v$ is an internal vertex of an abstract minimal foam $\kappa$ with non empty boundary graph such that there exists at least one edge $e$ joining $v$ and a node $\n$ in an initial graph, then $v$ is called a vertex of \underline{first generation}. Inductively a vertex of $n$th generation has at least one connection to a vertex of $(n-1)$th generation but no connections to vertices of lower generation.

 If $\partial\kappa$ only contains final graphs then we proceed backwards calling internal vertices, which are connected to $\partial\kappa$ by at least one internal edge, of generation $-1$ and so forth. 
 
If $\partial\kappa=\emptyset$ then all internal vertices are of first generation.  
\end{definition}
By definition, every internal vertex in a connected foam can be traced back along internal edges to the boundary graph and the shortest path to an initial graph, involving the least number of edges determines the generation. Suppose $\kappa$ contains a vertex $v$ which cannot be traced back to an initial part of the boundary graph, then either $\partial\kappa$ is empty or $v$ is only connected to a final graph. In the first case all vertices are of first generation while in the second case all internal vertices linked to $v$ are also detached from the initial graph. Since boundary graphs are closed and boundary nodes are only adjacent to one internal edge this is only possible if $v$ is part of a sub-foam which is completely disconnected and whose boundary graph only contains final graphs. Yet, the generation is independently defined for every completely disconnected sub-foam and therefore all internal vertices can be uniquely classified.
\begin{figure}[t]
 \begin{center}
	\subfloat{\includegraphics{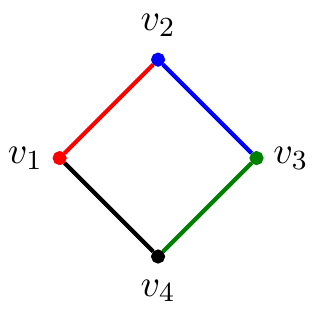}}
	\hspace{2 cm}
	\subfloat{\raisebox{1.2\height}{
\begin{tabular}{c||c|c|c|c}
 $v_4/v_3$		& $n$ 	& $n+1$ 	&$n+2$	& $\partial\kappa$\\
\hline\hline
$n-1$ 	&B 		& N.A.	&N.A.	&G\\
 \hline
$n$ 		&B		&G		&N.A.	&G\\
  \hline
$n+1$ 	&B/G/S 	&S		&S		&S\\
  \hline
$\partial\kappa$ 	&B/G/S 	&S		&S		&S\\
\end{tabular}
}}
 \end{center}
\caption{Suppose the red edge is an element of $\mathcal{E}_{n,n+1}$ and $v_1\in\mathcal{V}_n,v_2\in\mathcal{V}_{n+1}$. Then due to \lemref{lem:edges} the vertices $v_3$ resp. $v_4$ can be either of generation $n,n+1,n+2$ resp. $n-1,n,n+1$ or elements of an final spin net. Again by \lemref{lem:edges} if $v_3\in\mathcal{V}_{n+2}$ then $v_4$ must be of generation $n+1$ and the black edge is in $\mathcal{E}_{n,n+1}$. In the table on the right side we displayed all possible combinations for $n>1$ where  N.A.=not allowed, B=blue, G=green, S=black. }
\label{fig:faceth}
\end{figure}
\begin{lem}
\label{lem:edges}
Let $\mathcal{V}_n(\kappa)$ be the set of vertices of $n$th generation in $\kappa$ and suppose
$e\in\kappa_{int}$ is adjacent to $v\in\mathcal{V}_n(\kappa)$ then $v'\not v\in\dot{e}$ is either of generation $n-1$,$n$ or $n+1$ or $v'\in\partial\kappa$.
\end{lem}
\begin{proof}
The vertex $v'$ cannot be of generation $m<n-1$ since otherwise $v$ would be of generation 
lower than $n$.
If $v'\notin\mathcal{V}_{n-1}$ then $e$ is either a lowermost connection of $v'$ and consequently $v'\in\mathcal{V}_{n+1}$ or $e$ is adjacent to a vertex in $\mathcal{V}_n$ or in $\partial\kappa$. Note, if $v'$ is a boundary node then it is contained in a final graph unless $n=1$ 
in which case it can also be part of an initial graph.
\end{proof}
The set of all edges adjacent to a vertex $v\in\mathcal{V}_n$ and a vertex $v'$, which is either of generation $n+1$ or a node in a final graph, will be denoted by $\mathcal{E}_{n,n+1}$. Since $\hat{Z}[\kappa]$ is independent of internal edge orientations, we may also assume that all edges in $\mathcal{E}_{n,n+1}$ are oriented such that $s(e)\in\mathcal{V}_n$.
\begin{lem}
\label{lem:faces} 
Given a face $f$ and an edge $e_f\in\mathcal{E}_{n,n+1}$ in the frontier of $f$ then there exists at least one other edge $e'_f\in\dot{f}$ that is either an element of $\mathcal{E}_{n,n+1}$ or $s(e'_f)\in\mathcal{V}_{m}$, $m\leq n$ and $t(e'_f)\in\partial\kappa$.
\end{lem}
\begin{proof}
Since $\dot{f}$ is a closed loop the statement follows immediately (see \figref{fig:faceth}).
\end{proof}
\begin{theorem}
\label{th:split}
Every (finite in the sense of number of cells) connected, minimal, abstract spin foam $(\kappa,\{j_f\},\{Q_e\})$ can be uniquely split into minimal subfoams $(\kappa_i\{,j_{f_i}\},\{Q_{e_i}\})$ containing only vertices of $i$th generation with respect to the original foam such that for the colored foam holds $\kappa=\kappa_1\sharp\cdots\sharp\kappa_n$ 
where $n$ is the maximal generation of $\kappa$.
\end{theorem}
\begin{proof}\mbox{}\\
The theorem holds trivially for foams with empty boundary graph and w.l.o.g. we may assume that $\partial\kappa$ contains at least one connected initial(final) graph.

Consider the set $\mathcal{E}_{\partial\kappa}$ of edges intersecting the boundary graph in two points. Due to \lemref{lem:boundary} the nodes of an edge in $\mathcal{E}_{\partial\kappa}$ must lay in different, disjoint boundary graphs. Apart from that, there exist a natural orientation of internal edges adjacent to the boundary induced by the bordering property of boundary graphs. To be consistent boundary nodes $\n_i$ in an initial graph are always mapped to $\n_i\times[0,1]$ while nodes $\n_f$ of a final graph are mapped to $\n_f\times[-1,0]$ (see \defref{def:orfoam}). This implies that any edge in $\mathcal{E}_{\partial\kappa}$ must join an initial and a final graph.

 If $\kappa'$ is a foam derived from $\kappa'$ by splitting all edges in $\mathcal{E}_{\partial\kappa}$ then \lemref{lem:faces} guarantees that a face $f\in\kappa'$ is either bounded by at least two edges in $\mathcal{E}_{1,2}(\kappa')$ or by none. 
Suppose $v_f^1,\dots, v_f^m$ are the vertices of first generation in $f$ where the numbering is induced by the orientation of $f$, i.e. no other vertex of first generation is situated between $v^f_j$ and $v^f_{j+1}$. Recall that $f$ is path connected and homeomorphic to a 2-ball and therefore it is possible to connect $v^f_j$ and $v^f_{j+1}$ by an edge in $f$. Even better, we can introduce such edges $e'_{j,j+1}$ for all pairs $(v^f_j,v^f_{j+1})$ that are not already adjacent to the same edge in such a way that the edges $e'_{j,j+1}$ do not intersect. Closing the loop by joining $v^f_1$ and $v^f_n$, the face $f$ is divided into a subface $f'_0$ which has only vertices of first generation, $N/2$ faces $f'_i$ that contain exactly two edges of $\mathcal{E}_{1,2}(\kappa')$ and at most one face $\tilde{f}$ whose internal vertices are only of first generation and that intersects the initial graph in $\l_0$. Here, $N$ is the total number of edges $e^1_f,\dots,e^N_f\in\mathcal{E}_{1,2}(\kappa')$ bounding $f$. Since the frontier of a face constitutes a closed loop, the number of those edges is even. For the same reason, this splitting is independent of the face orientation and uniquely defined.

Let $\kappa''$ be the complex obtained from $\kappa'$ by subdividing all faces that contain vertices of first and second generation in the above manner, then $\kappa''$ satisfies:
\begin{itemize}
\item $\mathcal{E}_{\partial\kappa''}=\mathcal{E}_{\partial\kappa'}=\emptyset$
\item $\mathcal{E}_{1,2}(\kappa'')=\mathcal{E}_{1,2}(\kappa')$
\item $\forall f \in\kappa''$ s.t. $f\cap\mathcal{E}_{1,2}(\kappa'')\neq\emptyset\quad\exists! \;e_f,e'_f\in\mathcal{E}_{1,2}(\kappa'')\text{ and }e_f,e'_f\in\dot{f}$
\end{itemize}
The first two statement follow directly from the fact that the newly generated edges join only vertices of first generation and the third statement is a direct consequence of the splitting procedure for single faces.  

Proceed by subdividing every edge $e\in\mathcal{E}_{1,2}$ by a vertex $m(e)\in\dot{e}$ and join $m(e)$ and $m(e')$ by an edge $e(f)\in\dot{f}$ if $e$ and $e'$ are contained in the same face $f$.  Since all edges $e\in\mathcal{E}_{1,2}(\kappa'')$ are internal and therefore contained in at least two faces the set $\{m(e),e(f)\}$ give rise to a well-defined closed splitting graph $\gamma_s$ dividing $\kappa$ in two sub-foams:
\begin{itemize}
\item $\kappa_1$ containing only vertices of first generation whose boundary graph is the disjoint union of $\gamma_s$ and the initial graphs in $\partial\kappa$ 
\item $\kappa_{2\to f}$ which joins $\gamma_s$ and the final graphs of $\kappa$ 
\end{itemize}
Finally, remove all remaining subdivisions and proceed with $\kappa_{2\to f}$ in the same manner until no subfoam $\kappa_i$ contains two vertices of the same generation. From \lemref{lem:edges} follows immediately that the above splitting preserves the set of vertices of the same generation\footnote{A vertex $ v\in\kappa_{2,f}$ is of first generation iff it is second in $\kappa$} and thus $\kappa=\kappa_1\sharp\cdots\sharp\kappa_n$ (see \figref{fig:generation} for an example).

Since $\kappa$ is a minimal representative the procedure is unique. Furthermore, after the removal of all help edges and vertices the resulting blocks are again minimal. This proves the theorem. 
\end{proof}
Due to the gluing property the operator $\hat{Z}[\kappa]$ of a connected foam $\kappa$ decomposes as well into sub-operators containing only vertices of the same generation 
\begin{align}
\label{eqn:split}
\scal{T_{s_n}|\hat{Z}[\kappa]|T_{s_0}}=
\sum_{T_{s_1}\in\hilbert_{\gamma_{1}}}\!\cdots\!\!\!\!\sum_{T_{s_{n-1}}\in\hilbert_{\gamma_{n-1}}}\!\!\!\!
\scal{T_{s_n}|\hat{Z}[\kappa_{n}]|T_{s_{n-1}}} \scal{T_{s_{n-1}}|\hat{Z}[\kappa_{n-1}]|T_{s_{n-2}}} \,\cdots\,\scal{T_{s_1}|\hat{Z}[\kappa_{1}]|T_{s_0}}\;.
\end{align}
In the next section we want to apply this to the full projector. 
\begin{figure}[t]
 \begin{center}
	\subfloat{\raisebox{.5cm}{\includegraphics[height=7cm]{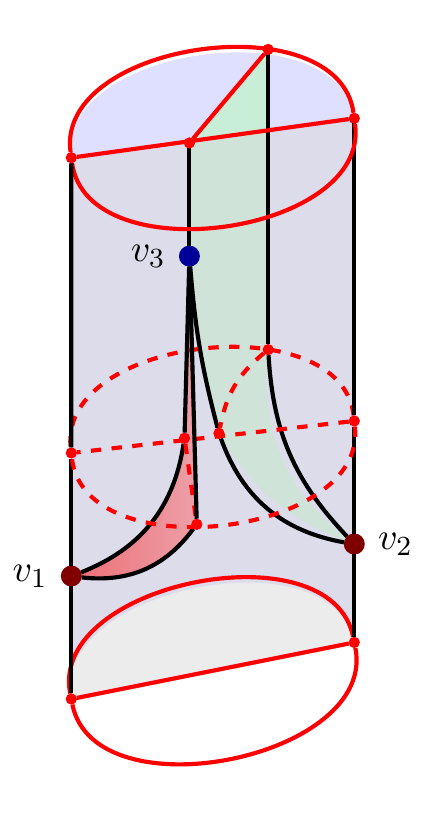}}}
	\hspace{2cm}
	\subfloat{\includegraphics[height=8cm]{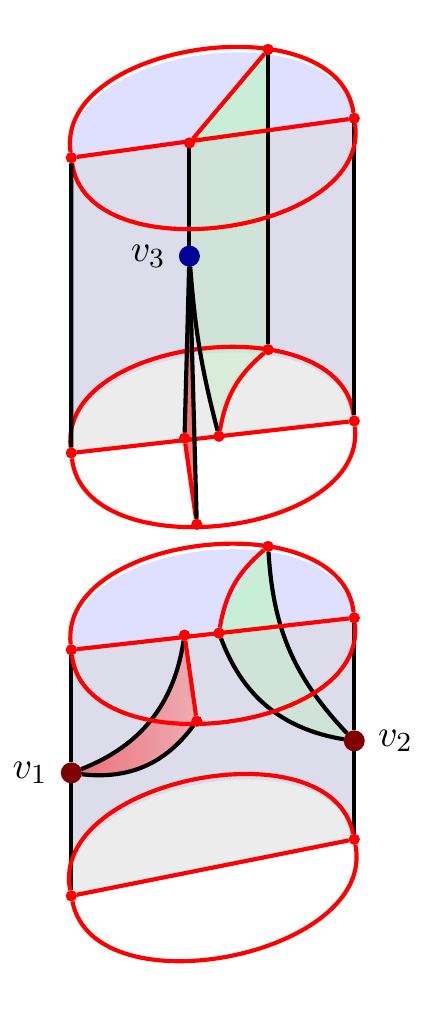}}
 \end{center}
\caption{A foam with two first order (red) and one second order (blue) vertex. The red graphs are the initial and final spin net induced on the boundary graph. The red dashed lines in the left picture indicate a cutting net such that the new blocks on the right only contain vertices of first generation.}
\label{fig:generation}
\end{figure}
\subsection{The time ordered projector}
\label{ssec:ttproj}

As above, the rigging map \eqref{eqn:riggdifH} can be restricted to fixed minimal representatives of the abstract equivalence classes by means of the map $L_{[s]_A}$ where the weight $\eta_{[s]_A}$ in \eqref{eqn:L_A} is set equal to one for simplicity. Thus $\eta$ effectively reduces to the operator averaging $\eta_{[s]_A,[s']_A}$ in \eqref{eqn:riggdifH}.
Explicitly,
\beq \label{amplitude}
\eta[T_{s_i}](T_{s_f})=\sum_{\kappa\in K_{\gamma(s_i),\gamma(s_f)}}
\;\scal{T_{s_i},\; \hat{Z}^{\dagger}[\kappa]\; T_{s_f}}
\eeq
where $K_{\gamma,\gamma'}$ is the set of {\it all} abstract minimal foams with fixed initial and final boundary graph $\gamma$ and $\gamma'$ respectively. 
Similar to the Feynman graph expansion of N-point functions in ordinary QFT,  `vacuum bubbles', that is, interior sums over connected foams with empty boundary graph just give rise to powers of $\eta^c[T_\emptyset](T_\emptyset)$ where the superscript $c$ indicates that only connected foams are involved.
Likewise,
contributions $\kappa$ of the form $\kappa=\kappa_i \cup \kappa_f$ with 
$\partial\kappa_i=\gamma(s_i),\;\partial\kappa_f=\gamma(s_f)$ and 
$\kappa_i\cap \kappa_f=\emptyset$ give again rise to powers of 
$\eta^c[T_\emptyset](T_\emptyset)$ times 
\beq
\label{a}
\eta^c[T_{s_i}](T_\emptyset)\;\;
   \eta^c[T_\emptyset](T_{s_f})~.
\eeq
The remaining contribution comes from the set $K^c_{\gamma(s_i),\gamma(s_f)}$ of connected foams with the given boundary graphs. Now suppose that the boundary graphs decompose into several disconnected components, then for example
\beq
\label{a1}
\begin{split}
\eta^{n.t.}[T_{s_1}\otimes T_{s_2}](T_{s_3}\otimes T_{s_4})
=& \left(\eta^c[T_{s_1}](T_{s_3})\right)\left(\eta^c[T_{s_2}](T_{s_4})\right)\\[2pt]
&+ \left(\eta^c[T_{s_1}](T_{s_4})\right)\left(\eta^c[T_{s_2}](T_{s_3})\right)\\[2pt]
& +\eta^c[T_{s_1}\otimes T_{s_2}](T_{s_3}\otimes T_{s_4})~.
\end{split}
\eeq
where the label $n.t.$ indicates that only foams are considered that do not split into disconnected initial and final parts as in the foregoing example.

Combining these arguments, we see that the amplitude \eqref{amplitude} is known if the connected amplitude
 \beq \label{connected.amplitide}
\eta^c[T_{s_i}](T_{s_f})=\sum_{\kappa\in K^c_{\gamma(s_i),\gamma(s_f)}}
\;\scal{T_{s_i},\; \hat{Z}^{\dagger}[\kappa]\; T_{s_f}}
\eeq 
can be computed. In fact, $\eta$ is a rigging map for the Hamiltonian constraint $\hat{H}(N)$ with lapse smearing function $N$ iff
\beq \label{annihilation}
\eta[T_{s_i}](\hat{H}(N)\;T_{s_f})=0\;\;\forall\;\; s_i,\; s_f,\; N~.
\eeq
In particular, equation \eqref{annihilation} must also hold for $s_i,s_f=\emptyset$. Moreover, the locality of the Hamiltonian action in combination with relation \eqref{a} and \eqref{a1} imply that \eqref{annihilation} is equivalent to 
\beq \label{connected.annihilation}
\eta^c[T_{s_i}](\hat{H}(N)\;T_{s_f})=0\;\;\forall\;\; s_i,\;s_f,\; N
\eeq
and thus it is sufficient to consider connected foams only in the sequel.

We can now apply the splitting procedure developed in the previous section to the connected map \eqref{connected.amplitide}:
\beq
\label{eqn:power1}
\eta^c[T_{s_i}](T_{s_f})
=\delta_{s_i,s_f}+\sum_{N=1}^\infty\;
\sum_{\hat{\kappa}_1\sharp..\sharp \hat{\kappa}_N\in K^c_{\gamma(s_i),\gamma(s_f)}}\; 
\scal{T_{s_i},\hat{Z}^{\dagger}[\hat{\kappa}_1]..\hat{Z}^{\dagger}[\hat{\kappa}_N] T_{s_f}}~.
\eeq
The second sum in \eqref{eqn:power1} extends over all `single-time-step' foams $\hat{\kappa}_{k},\; k=1,..,N$ whose internal vertices are all of first generation and whose gluing product is contained in $K^c_{\gamma(s_i),\gamma(s_f)}$, that is, consecutive foams $\hat{\kappa}_i$ and $\hat{\kappa}_{i+1}$ are glued along matching boundary graphs. The first sum runs over all possible values $N$ of maximal generation. Given the set $\hat{K}_{\gamma,\gamma'}$ of single time step foams with initial and final boundary graphs $\gamma$ and $\gamma'$ respectively, equation \eqref{eqn:power1} can be written more explicitly as 
\beq \label{evaluation}
\eta^c[T_{s_i}](T_{s_f})
=\delta_{s_i,s_f}+\sum_{N=1}^\infty\;\sum_{\gamma_1,..,\gamma_{N-1}}
\underset{\;\hat{\kappa}_k\in \hat{K}_{\gamma_{k-1},\gamma_k}} {^c\sum}\; 
\scal{T_{s_i},\hat{Z}^{\dagger}[\hat{\kappa}_1]\dots \hat{Z}^{\dagger}[\hat{\kappa}_N] T_{s_f}}
\eeq
with $\gamma_0:=\gamma(s_i),\;\gamma_N:=\gamma(s_f)$. The graphs $\gamma_1,..,\gamma_{N-1}$ belong to the afore mentioned set of minimal representatives of the abstract equivalence classes on which 
we know how to evaluate $Z(\hat{\kappa})$.

The label $c$ on the third sum in \eqref{evaluation} is to remind us that the glued product must be connected. That this is not a pure decoration can be understood form the following example: Let $\kappa$ be a connected foam made of two tubes, one connected to the initial and the other connected to the final graph, which are joined to the sides of a donut. Then, by imposing the time-splitting it might  happen that we slice the donut several times what possibly produces one-time-step foams that are not connected. Therefore, neither the graphs $\gamma_1,\dots, \gamma_{N-1}$ in \eqref{evaluation} nor the elements in $\hat{K}_{\gamma,\gamma'}$ can be restricted to the connected category.

Recall that the generation of a vertex is uniquely defined and therefore two components can be only joined by identifying vertices of the \emph{same} generation. Concluding, since $\hat{K}_{\gamma,\gamma'}$ is generated by cutting connected foams, \emph{all} disjoined components of a single time step foam $\hat{\kappa}\in\hat{K}_{\gamma,\gamma'}$ must contain at least one (non-trivial) internal vertex and are bounded by non-trivial in- and outgoing graphs.

These troubles, related to the use of  \eqref{connected.amplitide}, cannot be avoided by working with the full Rigging map \eqref{amplitude}.
The reason for this is that the maximal generation in disconnected components of a given foam do not have to agree in general what causes ordering ambiguities when passing from \eqref{connected.amplitide} to \eqref{eqn:power1}. On the other hand, when single time step foams of the above type, bordered by non-empty final and initial graphs, are glued to a connected graph then the resulting foam will always be connected. This is true even if these building blocks consist of several disconnected components. It therefore suffices to require that either $s_i$ or $s_f$ is connected. Since the Hamiltonian as defined in \cite{Thiemann96a} is acting locally on the nodes and can therefore only split but not glue, we prefer to restrict the domain of $\eta[T_{s_i}]$ to the subspace $\mathcal{H}_{0}^c$ in which the finite linear span of connected spin nets lies dense.

If the final spin net ${s_f}$ is connected then the label $c$ on the sum in \eqref{evaluation} can be removed and the whole expression can be simplified by introducing the 
{\it spin foam transfer matrix}\footnote{A similar matrix was already introduced in \cite{Bahr:2012q} in the context of holonomy spin foam models. However, in this work the authors only considered a very specific regular type of foams  in order to identify a transfer matrix.}
\beq \label{elementary}
\hat{Z}:=\sum_{\gamma,\gamma'} \;P_{\gamma'}\; 
[\sum_{\hat{\kappa}\in \hat{K}_{\gamma,\gamma'}}\; Z[\hat{\kappa}]]\;P_\gamma~.
\eeq
Here, $P_\gamma$ is the projection operator on the subspace of $\hilbert_{0}$ consisting of the closed linear span of spin network functions over $\gamma$ (with all spins non vanishing on each link). Note that  $\hat{Z}$ is still on the linear span of \emph{all} spin net functions including disconnected ones.

Since $P_\gamma P_{\gamma'}=\delta_{\gamma\gamma'} P_\gamma$ the sum over single time step foams can be replaced by 
\beq \label{identity}
P_\gamma \; \hat{Z} \; P_{\gamma'}=\sum_{\hat{\kappa}\in \hat{K}_{\gamma\gamma'}}
\;\; Z[\hat{\kappa}]
\eeq
and thus \eqref{evaluation} is equivalent to
\beq \label{evaluation1}
\eta^c[T_{s_i}](T_{s_f})
=\delta_{s_f,s_i}+\sum_{N=1}^\infty\;\sum_{\gamma_1,..,\gamma_{N-1}}
\scal{T_{s_f}, P_{\gamma(s_f)} \;\hat{Z}^{\dagger} \; P_{\gamma_{N-1}}\; \hat{Z}^{\dagger} P_{\gamma_{N-2}}\;
...\; P_{\gamma_1}\; \hat{Z}^{\dagger}\; P_{\gamma(s_i)}\; T_{s_i}}~.
\eeq
Using that $T_{s}=P_{\gamma} T_{s}$ and the fact that $\sum_\gamma P_\gamma$ 
is the identity on ${\cal H}_{0}$ 
we deduce the compact formula
\beq \label{evaluation2}
\eta^c[T_{s_i}](T_{s_f})
=\sum_{N=0}^\infty\;
\scal{T_{s_i}, (\hat{Z}^{\dagger})^N T_{s_f}}\quad\forall T_{s_f}\in\hilbert^c_{0}~.
\eeq
The operator $\hat{Z}$ no longer refers to a given boundary graph. 
Therefore, dropping the requirement that $s_f$ is connected, the right hand side of \eqref{evaluation2} can be extended to a suitable dense subset of the
whole Hilbert space $\hilbert_{0}$ even including states that are not finite linear combinations of spin net functions (see below). One should, however, keep in mind that the equality in \eqref{evaluation2} only holds if $T_{s_f}$ is an element of $\hilbert_{0}^c$. Nonetheless, if $\eta$ defines a rigging map then $\eta[T_{s_i}](HT_{s_f})$ must also vanish for all $T_{s_f}\in\hilbert^c_{0}$.

\subsection{Regularization and properties}

To test the rigging map on the subspace $\hilbert^c_{0}$ as suggested, the formal expression \eqref{evaluation2} must be regularized. The strategy here is first to regularize the expression on the right hand side of \eqref{evaluation2} by turning the formal operator $\hat{Z}$ into a densely defined quadratic form on the form domain given by the finite linear span of spin network functions in the full Hilbert space ${\cal H}_{0}$ and afterwards restrict to $\hilbert_{0}^c$.
To tame the infinite spin sums in $\hat{Z}[\kappa]$ for fixed $\kappa$, a spin cut-off $J$ has to be introduced, that is, all spins $j$ that contribute to the spin foam operator $\hat{Z}[\kappa]$ are supposed to obey $j\le J$. However, we must also impose a bound $N_f$ on the valence of the internal edges (i.e. the number of faces intersecting it) and a bound $N_e$ on the valence of internal vertices. A bound on the number of internal vertices in $\hat{\kappa}$ is not necessary since each internal vertex of first generation must be contained in an edge of the form $\n_0\times[0,1]$ where $\n_0$ is a node in an initial graph and consequently $\hat{\kappa}$ can have at most as much internal vertices as their are nodes in the initial graph. For p.l. complexes this would also restrict the number of possible internal edges and faces but in the abstract category several edges can intersect at the same endpoints and faces can be glued on the same frontier. Therefore, the cut-offs $N_f$ and $N_e$ are necessary to render $\hat{K}_{\gamma,\gamma'}$ finite.

Yet, this still does not turn $\hat{Z}$ into a densely defined operator as it has non vanishing matrix elements between any two spin network states (given $\gamma$ take any $\gamma'$ and let $\hat{\kappa}$ be the single time step foam such that all initial and final nodes
are joined via internal edges to a single internal vertex of first generation). To cure this the elementary operators $Z[\hat{\kappa}]$ should be equipped with a weight $w(\hat{\kappa})$. This weight should be such that  $w(\kappa):=\prod_{k=1}^n\;w(\hat{\kappa}_k)$ for the connected foam $\kappa:=\hat{\kappa}_1\sharp..\sharp \hat{\kappa}_n$, otherwise the gluing property would be violated and the above statements  would no longer be applicable. Denote the modified operator by $\hat{Z}'$ and pick the weight $w$ 
in such a way that 
\beq
||\hat{Z}' T_s||^2=\sum_{s'} \;|\scal{T_{s'},\hat{Z}' T_s}|^2 
=\sum_{\gamma'}\; \sum_{j',\iota'} 
|\sum_{\hat{\kappa}\in \hat{K}_{\gamma,\gamma'}} \; w(\hat{\kappa})
\scal{T_{\gamma',j',\iota'},Z(\hat{\kappa}) T_{\gamma,j,\iota}}|^2
\eeq
converges for $s=(\gamma,j,\iota)$. This is possible because firstly the set $\hat{K}_{\gamma,\gamma'} $ for given $\gamma,\gamma'$ is finite due to the bounds $N_f$ and $N_e$, 
secondly the sum over $j',\iota'$ for fixed $\gamma,\gamma'$ is finite due to the cut-off $J$,
and thirdly, due to the restriction to the embedded representatives of abstract minimal graphs, the sum over $\gamma'$ is countable.
It will therefore be sufficient to pick $w(\hat{\kappa})$ for $\hat{\kappa}\in \hat{K}_{\gamma,\gamma'}$ to be such that it suppresses the growth behavior as $\gamma,\gamma'$ become 
large after having performed the sum over $j',\iota',\hat{K}_{\gamma,\gamma'}$. 
It is likely that this growth behavior is bounded by the number
\beq
C(J,N_f,N_e)^{|E(\gamma)|+|E(\gamma')|}
\eeq
where $C(J,N_f,N_e)$ only depends on the cut-offs. The reason for this is that we expect polynomial 
growth in $J$ for every face 
due to the $nj$ symbols involved in the spin foam amplitude of which there are an order of 
$N_f N_e(|E(\gamma)|+|E(\gamma')|)$.

Having tamed $\hat{Z}$ like this as an operator densely defined on ${\cal H}_J$,
which is the subspace of ${\cal H}_{0}$ defined by the spin cut-off $J$,  it is 
still not clear that its powers\footnote{We rename $\hat{Z}'$ by $\hat{Z}$ again} are densely defined as its domain, the finite linear span of spin network functions is not preserved. Namely, the range of $\hat{Z}$ lies always in the closure of its domain.  
To improve on that, notice that $\hat{Z}$ is formally a symmetric operator. This follows from the reality of all 
amplitude factors that define the operator $\hat{Z}[\kappa]$ (at least in the Euclidian 
setting) and the fact 
that both $\kappa^{\ast}$ and $\kappa$ are elements of $\hat{K}_{\gamma',\gamma}$. Whence (\ref{elementary}) is invariant under taking adjoints and $^{\dagger}$ can be skipped in the folowing. It also follows irrespective of how $\hat{Z}$ was computed from the expression (\ref{evaluation2}) which should define a sesqui-linear form. 

Suppose $\hat{Z}$ can be extended as a self-adjoint operator on ${\cal H}_J$ with projection valued measure $E$. Let ${\cal H}_q:=E([-q,q]){\cal H}_{0}$ be the closed subspace of ${\cal H}_{0}$ on which $\hat{Z}$ acts by multiplication with $\lambda\in [-q,q]$ where $0<q<1$. More specifically its elements are of the form 
\beq
\psi_q:=\int_{-q}^q \; dE(\lambda) \; \psi,\; \psi\in {\cal H}_{0}~.
\eeq
The operator $A:=\sum_{N=0}^\infty \hat{Z}^N$ acts on these vectors 
as 
\beq
A\psi_q=\int_{-q}^q\; dE(\lambda)\; \sum_{n=0}^{\infty} \lambda^n\; \psi=\int_{-q}^q\; dE(\lambda)\; (1-\lambda)^{-1}\; \psi
\eeq
defining formally a geometric series. This implies 
\beq
||A\psi_q||^2=\scal{\psi, A^2 E([-q,q]) \psi}=\int\; d\scal{\psi,E(\lambda)\psi} [1-\lambda]^{-2}\;
\le (1-q)^{-2} ||\psi_q||^2
\eeq 
 and whence $A$ and any power of $\hat{Z}$ is even bounded on ${\cal H}_q$. Accordingly, on ${\cal H}_q$ holds $A=1+\hat{Z} A$. 
 
Let now $\psi'_q=E([-q,q])\psi' \in {\cal H}_q$ be in the domain of the averaging map and $\psi,\psi'$ in the domain of $\hat{H}(N)$ for any lapse function\footnote{$\hat{H}(N)$ also must be projected to ${\cal H}_J$}. Then, if $\eta$ is a rigging map for $\hat{H}$, we find 
\beq
\begin{split}
\label{eqn:B}
0&=\eta[\psi'_q](\hat{H}(N)\psi_q)=\scal{\psi'_q,A\hat{H}(N)\psi}=\scal{A\psi'_q,\hat{H}(N)\psi}\\
&=\scal{\psi'_q,\hat{H}(N)\psi}
+\scal{\hat{Z}\psi'_q,A\hat{H}(N)\psi}=\scal{\psi'_q,\hat{H}(N)\psi}
\end{split}
\eeq 
for all $\psi',\psi$ in the common 
domain ${\cal D}^c$ of all $\hat{H}(N)$ defined by the finite linear span of connected spin network functions over the allowed set of graphs. 
In fact, as long as $\psi$ is connected, equation \eqref{eqn:B} holds also for states $\psi'$ that are finite linear combinations of spin network functions on arbitrary graphs, including disconnected ones. Because the Hamiltonian can only split spin nets, we can therefore choose $\psi':=\hat{H}(N)\psi$. In particular
\beq
||E([-q,q])\hat{H}(N)\psi||^2 =0
\eeq
has to hold for all $0<q<1$. Thus the range of the $\hat{H}(N)$ avoids the kernel 
of $\hat{Z}$. To bring this into a familiar form, notice that it follows from the Cauchy Schwarz identity
\beq 
\scal{\psi',E(-q,q)\hat{H}(N)\psi}=0 
\eeq
for all $\psi,\psi'\in {\cal D}^c$. Dividing by $2q$ and taking $q\to 0$ we conclude 
(in the sense of the functional calculus)
\beq
\scal{\psi',\delta(\hat{Z}) \hat{H}(N)\psi}=0
\eeq
We arrive at the first conclusion: If the spin foam amplitude as above defines 
a projector on the joint kernel of the $\hat{H}(N)$, the range of any of the 
$\hat{H}(N)$ must be orthogonal to the kernel of the spin foam Hamiltonian $\hat{Z}$. In other words the $\hat{H}(N)^\dagger$ annihilate the kernel of $\hat{Z}$. If true 
this would tell us how to construct a spin foam model given $\hat{H}(N)$ or vice versa 
how to build a Hamiltonian given a spin foam model. For instance, above 
criterion would be satisfied if the `spin foam Hamiltonian' $\hat{Z}$ takes the form of a master constraint
\beq  \label{Master}
\hat{M}=\sum_{I,J} Z^{IJ} \hat{H}(N_I)\;\hat{H}(N_J)^\dagger  
\eeq
for a suitable choice of matrices $Z^{IJ}$ and smearing functions $N_I$, see 
\cite{TTMaster} for details where this kind of expression was considered 
as the source of an alternative spin foam model. In particular, such choice ($\hat{Z}=\hat{M}$) is bounded by zero from below and thus one can in principle make use 
of analytic continuation techniques in order to define the path integral 
rigorously (Feynman-Kac formula). 

However,
this identification of the kernels of $\hat{Z}$ and $\hat{M}$ cannot be correct:  
Namely, the second conclusion is the following: According to the above argumentation
\eqref{evaluation2} should be a generalized projector on the kernel of $\hat{Z}$. Instead of $T_{s_i},\; T_{s_f}$ we pick the states $\psi_q=E([p,q])\psi$ and $\psi'_q=E([p,q])\psi'$ with any connected $\psi$, any $\psi'$ and $0<p<q<1$. Thus $\scal{\psi'_q,[\sum_{N=0}^\infty \; \hat{Z}^N]\;\hat{Z}\;\psi_q}$ should vanish if $\hat{Z}=\hat{M}$. Yet, an explicit evaluation gives
\beq \label{contradiction1}
0=\int_p^q\; d\scal{\psi',E(\lambda)\psi}\;\frac{\lambda}{1-\lambda}
\eeq
where the spectral measure is defined by the polarization identity. In particular
for $\psi=\psi'$
\beq \label{contradiction2}
0=\int_p^q\; d\scal{\psi,E(\lambda)\psi}\;\frac{\lambda}{1-\lambda}
\eeq 
yields a contradiction unless all spectral measures $E_\psi=\scal{\psi, E(.)\psi}$ have 
no support in $(p,q)$. Indeed, if $\hat{Z}$ and the operators $\hat{H}(N)^\dagger$ for all choices of lapse functions really
have the same kernel then the expression (\ref{evaluation2}) is somehow 
incorrect and should better be replaced by the heuristic expression
\beq \label{evaluation3}
\eta^c[\psi'](\psi):=\scal{\psi',\delta(\hat{Z})\psi}=\lim_{T\to \infty}\int_{-T}^T
\; \frac{dt}{2\pi}\; \scal{\psi',e^{it \hat{Z}}\psi}~.
\eeq
When this is formally expanded it yields again a power series in $\hat{Z}$ as before but 
with different coefficients (at finite $T$). To make this even more obvious, suppose that 
by introducing the cut-offs $J,N_f,N_e$ the operator $\hat{Z}$ becomes bounded. By rescaling $\hat{Z}$ by a suitable global factor, i.e. by  just choosing a different weight, we may assume without loss of generality that $||\hat{Z}||<1$. But then 
\beq
[\sum_n\; \hat{Z}^{n}]\; \hat{H}(N)=(1-\hat{Z})^{-1}\hat{H}(N)=0
\eeq
is obviously a contradiction as can be seen by multiplying with the invertible operator 
$1-\hat{Z}$ from the left. In summary, the identification of $\hat{Z}$ with the master constraint 
of the $\hat{H}(N)$ is not sustainable.

We conclude this subsection with some speculations about the Lorentzian case.
The Lorentzian spin foam amplitudes can be glued similar to Euclidean ones. Moreover, the lifting defined in \secref{ssec:psn} was actually first developed for Lorentzian spin foams. Yet, certain Lorentzian vertex amplitudes are not integrable \cite{Kaminski:2010qb} and therefore the class of foams must be restricted further. Nevertheless, the splitting defined in \secref{ssec:time} does not affect vertex amplitudes, so that this problem can be possibly ignored and thus the conclusions 
reached at above are not affected by switching signatures.

\section{Merging covariant and canonical LQG: The current status}
\label{sec:conclusion}

The extension of spin foam amplitudes based on the duals of simplicial complexes 
to arbitrary complexes invented in the seminal paper
\cite{Kaminski:2009fm} offers for the first time the exciting possibility 
to test: A. whether a given spin foam model defines a rigging map for a given Hamiltonian 
constraint, B. whether a spin foam model has a rigging kernel at all and if yes to which 
Hamiltonian constraint it corresponds, or C. how to define a spin foam model such 
that it defines a rigging map for a given Hamiltonian. This is is not possible for 
spin foam models based on simplicial complexes because those necessarily have purely 4-valent boundary graphs. Yet, firstly, the rigging map must act on the full LQG Hilbert space 
and, secondly, none of the known anomaly free versions of the Hamiltonian constraint preserves the subspace spanned by spin network states based on a  purely 4-valent graph.

In order to postulate such a spin foam rigging map it is necessary to sum over all possible spin foams whose boundary graphs match the initial and
final graph of the would-be rigging map matrix element. For the regularization of the sum, group field theory (GFT) techniques cannot be utilized here because the interaction part of a GFT
Lagrangian dictates the possible valence of a dual graph and so far is 
geared to duals of simplicial triangulations. To get rid of this restriction one would have to allow all possible interaction terms based on certain invariant polynomials of arbitrarily many 
gauge group elements what is currently out of reach. Therefore, the
only sensible definition of the sum is a naive sum, possibly dressed with a weight function,
that is compatible with certain natural rules listed in subsection \ref{ssec:road}. These rules are established so that the sum has a chance to define a rigging map. For this purpose it was necessary to overcome the restriction to embedded graphs by allowing abstract ones and to introduce a regulator artificially so that the mathematical expressions converge at least order by order. 

Even though the considerations in \cite{Alesci:2011ia} have shown that the above set-up can at least partially get to work by restricting on the Euclidean theory and a very specific class of foams, the discussion in the last two sections has revealed that this is not feasible for the full model. This conclusion is totally independent of the particulars of face and vertex amplitudes and the details of the Hamiltonian constraint. The essential properties, which are needed to reach at this point, are A. the realization of boundary states of spin foams as kinematical states of the canonical theory and B. the gluing property of the spin foam amplitude. Let us also remark that no proof in the whole section depends on the invariance under edge or face subdivisions since they can be removed prior of the computation of $\hat{Z}$.
 
 On the other hand, it transpires that the single time step operator $\hat{Z}$,
 which defines a sort of elementary transfer matrix from which any
 foam can be generated, is of importance. Its matrix elements between 
 spin network states provide the `integral kernel' (better: summation kernel) of the 
 single time step spin foam evolution. The correct correspondence between 
 $\hat{Z}$ and $\hat{M}$, the master constraint associated to all of the $\hat{H}(N)$, can currently only
 be speculated about which is sketched here for completeness
 (see \cite{SFLQC} for similar ideas):
 Using a skeletonization of the interval $T$ into $n$ intervals of length 
 $T/n$ and a Riemann sum approximation one obtains
 \beq \label{shouldbe}
 \begin{split}
 2\pi\delta(\hat{M})
   &=\lim_{T\to\infty} \int_{-T}^T\; e^{it\hat{M}} 
   =\lim_{T\to\infty} \int_0^T\; [e^{it\hat{M}} +e^{-it\hat{M}}] \\[3pt] 
  & =\lim_{T\to\infty} \;\lim_{n\to \infty} \sum_{k=0}^n  \frac{T}{n}
   \{[e^{iT\hat{M}/n}]^k+[e^{-iT\hat{M}/n}]^k\}~.
\end{split}
\eeq 
If $\hat{Z}$ would be unitary rather than symmetric then \eqref{evaluation2} suggests that $\hat{Z}=e^{i\tau \hat{M}}$ for an artificial synchronization of the limits $T,n\to \infty$ by some constant $\tau=T/n$. Just that in this case the geometric series over the negative powers of $\hat{Z}$ is missing. But, as 
$\hat{Z}$ is symmetric, it appears that one 
wrongly took the geometric series of $\cos(\tau \hat{M})$ rather than the Laurent
series of $e^{it \hat{M}}$ minus one. This can be interpreted as an artifact of summing over both Plebanski sectors (II $\pm$) rather than keeping only the sector corresponding to the Einstein-Hilbert action\footnote{The sector (I $\pm$) is excluded by implementing the linearized constraint while the Einstein-Hilbert sector of the 4-simplex amplitude can be only specified when taking the orientation of the 4-simplex into account. See \cite{Engle:2011ps} for details.}. 
Both identifications for unitary and symmetric $\hat{Z}$ suppose, of course, that $\hat{Z}$ has unit norm so that the geometric series is marginally divergent as it is expected for a $\delta$-distribution. In that case it would be natural to define $\tau M:=\frac{1}{i}\ln(Z)$ (modulo $2\pi$) or $\tau(M)={\rm arcos}(Z)$ (modulo $\pi$) where the value of $\tau$ is irrelevant as $M$ is a constraint. With this identification of $\hat{M}$ in terms of $\hat{Z}$, (\ref{evaluation2}) should then, perhaps, be replaced by $\delta(\hat{M})$.

A similar problem can be also observed in a very different set-up. Namely, the semiclassical limit \cite{bdfgh2009} of the 4-simplex vertex amplitude also suffers from the appearance of multiple terms, each consisting in an exponential of the Regge-action, that resembles a series of the cosine rather then the one exponential term one would expect. In \cite{Engle:2011ps}, the author showed that this can be cured for the Euclidean theory by inducing an additional constraint. It would be of course interesting to check whether such a constraint can also resolve the problems of the postulated rigging map. 

Apart form that, the whole argument would break down if $\hat{Z}[\kappa]$ cannot be split into smaller subfoams. This can be understood as a hint that the vertex amplitude is too local in the sense that it only depends on the coloring of the adjacent faces and not on other vertices in its surrounding and therefore allows splitting. Indeed there are indications that a more non-local action is needed to reconstruct diffeomorphism invariance on a lattice \cite{Dittrich:2012qb}. On the other hand, also the action of the Hamiltonian constraint \cite{Thiemann96a} is local. Whether this locality is a problem is still under debate in both approaches, the canonical as well as the covariant one. 

One could, of course, also argue that the assumption on the weight in the above section is too naive and the weight possibly does not feature the same factorization properties as the non-normalized operator $\hat{Z}[\kappa]$. For example, it might not be necessary to introduce a spin cut-off a priori since many colorings are restricted by the compatibility conditions imposed by the labeling of the boundary spin nets. Nevertheless, there exist regions, so called bubbles (see \cite{Smerlak:2012je} and references therein), where the spins are not restricted by the boundary data. These region therefore require a separate regularization if the spins are not restricted by a general cut-off. For attempts into this direction see \cite{Christodoulou:2012af,Riello:2013bzw,Riello:2013gja}. By splitting the foam it might happen that one cuts through such a bubble so that it is no longer present in the single time step foams. Only after gluing the blocks this structures reappear. It therefore deserves further investigations whether this bubbles can be renormalized by weights of the above form .

In addition, the equivalence class $\kappa\sim\kappa'\Leftrightarrow: Z[\kappa]=Z[\kappa']$ is huge as many combinatorial inequivalent foams can lead to the same amplitude. For example, two regions of a foam can be swapped if the boundary spin nets induced on a closed surface around this regions are equal. To avoid over-counting it might be necessary to introduce a factor in the model determining the multiplicities of interchangeable regions. A similar factor was already advertised in order to prove a relation between summing over all foams and refining and was argued to be related to the volume of the orbit of diffeomorphism acting on a colored complex (see \cite{Rovelli:2010qx}). This idea is misleading here since  a map, somehow related to diffeomorphism, should be at least continuous while cutting out parts of $\kappa$ and gluing them in somewhere else does not define a continuous function. Despite, we are working with abstract complexes while a diffeomorphism is only affecting the embedding so in this sense we have already taken care of diffeomorphisms. Instead, we propose to include a purely statistical factor related to the (heuristic) expansion of the exponential in \eqref{eqn:BFpart}. Why could such a factor cure the problem? Suppose the constraint $\hat{C}$ can be implemented via group averaging
\bq
\int\!\!\dif{\alpha}\,\scal{m|\exp{i\alpha\hat{C}}|n}~.
\eq
Expanding the exponential and inserting a resolution of unity $\sum_m\ket{m}\bra{m}$ yields
\bq
\begin{split}
&\int\!\!\dif{\alpha}\left(\delta_{m,n}+\sum_{N=1}\frac{(i\,\alpha)^N}{N!}\sum_{m_1}\cdots\sum_{m_{N-1}} C_{m m_1}C_{m_1m_2}\cdots C_{m_{N-1} n}\right)\\
=&\int\!\!\dif{\alpha}\left(\delta_{m,n}+i\alpha\sum_{m_1} C_{mm_1}\left\{\delta_{m_1n}+\sum_{N=1}\frac{(i\,\alpha)^N}{(N+1)!}\sum_{n_1}\cdots\sum_{n_{N-1}} C_{m_1 n_1}\cdots C_{n_{N-1} n}\right\}\right)
\end{split}
\eq
 with $C_{mn}\equiv\scal{m|\hat{C}|n}$. Due to the factor $N!$  this expression cannot be written as a formal geometric series as it was done for the spin foam transfer matrix. Of course, for spin foams the situation is more complicated since matching conditions depending on the bulk structure and coloring of each foam have to be respected. The inclusion of a statistical factor $N(\kappa)$ as advertised above can only solve the problem if this factor is non-local in the sense that $N(\kappa\sharp\kappa')\neq N(\kappa)\,N(\kappa')$ so that $N(\kappa)Z[\kappa]$ does no longer posses a gluing property. 

Yet, the reason for demanding a gluing property \eqref{eqn:glue} is deeply rooted in the `sum over histories' interpretation of spin foams: two `histories' glued together should yield a new `history'. However, it was argued earlier \cite{Kiefer199153, timeqg,Calcagni:2010ad} that (causal) propagators do not entail a physical scalar product or projector. For instance, in \cite{Kiefer199153} the author illustrated that mainly due to the absence of an extrinsic time parameter a propagator $G$ in Wheeler-deWitt Cosmology cannot define a projector since (if $G$ can be normalized) it is not idempotent\footnote{For the renormalized projector to be idempotent the single amplitudes do not necessarily have to satisfy a gluing property. Thus there is no contradiction between demanding idempotency and violation of a gluing property.}. More recently, Calcagini, Gielen and Oriti \cite{Calcagni:2010ad} analyzed different two-point functions in LQC coupled to a scalar field and found that only certain two-point functions\footnote{So-called non-relativistic Newton-Wightman functions} $G(x',t';x,t)$ define a positive-definite physical scalar product satisfying $G^2=G$ while all constructed causal propagators fail to either satisfy an adequate composition property or are not defining a positive definite scalar product. Even more severe, the Feynman propagator does not even solve the constraint equation rather it defines a Green's function. All two-point functions\footnote{Except for the relativistic causal two-point function} considered in \cite{Calcagni:2010ad} can be transformed in a `vertex expansion' which closely resembles the spin foam model and was later constructed in \cite{SFLQC}. But there are two important differences between spin foams in the full theory and two-point functions in (L)QC. For Wheeler-deWitt as well Loop Cosmology  the properties of the propagator highly depend on the contour of integration and a super-selected sector of solutions (see \cite{Kiefer199153,Calcagni:2010ad,Halliwell:1989dy,Henneaux1982127} and references therein) while in LQG  neither the complete set of solutions  nor the correct path integral measure is known. Thus, the impact of a `contour' is rather obscure. Second, in the presence of a scalar field the vertex expansion in LQC is always \emph{non-local}.
 
Obviously, the whole problematic is bypassed if $\hat{Z}$ is itself a projector. This idea is supported by the computation in \cite{Alesci:2011ia} where it was sufficient to consider only complexes with a single internal vertex excluding the trivial evolution. Also in \cite{Bahr:2012q} the authors showed that for BF-theory it is in fact possible to construct a `spin foam transfer matrix'  that annihilates the 4-dimensional curvature. Their transfer matrix is constructed by gluing arbitrary but fixed building blocks embedded in space time. But BF-theory is topological and therefore independent of the triangulation which is certainly not the case for quantum gravity. Therefore, it is questionable that the transfer matrix defined here could already implement the constraint. Also heuristically there is no good argument why trivial foams and larger foams with vertices of several generations should be excluded

{\acknowledgments
 A.Z. acknowledges financial support of the `Elitenetzwerk Bayern' on the grounds of `Bayerische Elitef\"order Gesetz'.
}

\appendix
\section{Harmonic analysis of $\SU(2)$}
\label{app:HA}
Since $\SU(2)$ is compact and semisimple its representations are completely reducible in the sense that a given unitary representation $\rho: \SU(2)\mapsto U(\hilbert)$ where $U(\hilbert)$ is the set of bounded unitary operators on the Hilbert space $\hilbert$ decomposes into a direct sum of irreducible (finite dimensional) representations $\hilbert_j$ labeled by spin $j\in\frac{1}{2}\mathbb{N}$. Thus the Hilbert space $\sqi{\SU(2),\mu_H}$ of square integrable functions $\Phi:\SU(2)\to\C$, where $\mu_H$ is the Haar measure is isomorphic to $\bigoplus_j\hilbert_j$ and an orthogonal basis is spanned by the representation matrix elements (Wigner matrix) $R^j_{nm}(g)$ :
\beq
\label{eqn:schurA}
\int\dif{\mu_H(g)}\; \overline{R^j_{m n}(g)} \;R^k_{r s}(g)=\frac{1}{d_j}\delta_{j,k}\;\delta_{m,r}\;\delta_{r,s}
\eeq
where $d_j=2j+1$. Since $\SU(2)$ is unitary $\overline{R^{j}_{mn}(g)}=R^{j}_{nm}(g^{-1})$. A convolution on this space is defined by using characters $\chi^j(g)=\tr\;R^j(g)$:
\beq
\int\dif{\mu_H(g)} \;\underbrace{\sum_j d_j\; \chi^j(hg^{-1})}_{\delta_h(g)} \;R^k_{mn}(g)= R^j_{mn}(h)
\eeq
For $\SU(2)$ $\tr(g^{-1})= \tr(g)$ since every group element can be expanded in terms of Pauli matrices whose trace is zero.

An intertwiner is a function $\iota:V\to W$ from a representation space $V$ into $W$ such that it commutes with $\rho$. For completely reducible representations $\iota$ is either zero or it defines an isometry between an invariant subspace of  $V$ and $W$. For example, $V=\hilbert_j\otimes\hilbert_k$ decomposes into a sum over irreps which obey the triangle inequality $|j-k|\leq l\leq j+k$ and $j+k+l\in\mathbb{N}$:
\beq
\hilbert_j\otimes\hilbert_k=\bigoplus_{l=|j-k|}^{j+k} \hilbert_l
\eeq 
Note, each irrep $l$ occurs with multiplicity one and therefore the space of intertwiners $\iota:\hilbert_j\otimes\hilbert_k\to\hilbert_l$ is one-dimensional and $\iota$ can be e.g expressed by \emph{Clebsch-Gordan coefficients}
\beq
C^{j,\mu;\;k,\nu}_{l,r}:=\scal{l,\lambda |j,\mu;k,\nu}
\eeq
where $\mu=-j,\dots,j$ and $\nu,\lambda$ are magnetic indices. Alas, they are not normalized 
\beq
\sum_{\mu,\nu,\lambda}\overline{C^{j,\mu;\;k,\nu}_{l,\lambda}}C^{j,\mu;\;k,\nu}_{l,\lambda}=d_l
\eeq
and not invariant under cyclic permutations of the indices. Instead we can use $3j$-symbols 
\begin{gather}
\label{eqn:int3j}
\begin{gathered}
\{3j\}:\hilbert_j\otimes\hilbert_k\otimes\hilbert_l\to\C\\
\iota_{j,\mu;\;k,\nu;\;l,\lambda}=\threej{j}{\mu}{k}{\nu}{l}{\lambda}\;.
\end{gathered}
\end{gather}
which are non zero only if $\mu+\nu=\lambda$ and $j,k,l$ are compatible. They span the (one-dimensional) invariant space $\Inv\left(\hilbert_j\otimes\hilbert_k\otimes\hilbert_l\right)$, are invariant under cyclic permutations of $j,k,l$ and normalized:
\beq
\label{eqn:firstorth}
\sum_{\mu,\nu}\threej{j}{\mu}{k}{\nu}{l}{\lambda}\threej{j}{\mu}{k}{\nu}{\tilde{l}}{\tilde{\lambda}}=\frac{1}{d_l}\; \delta_{l,\tilde{l}}\;\;\delta^{\tilde{\lambda}}_{\;\;\lambda}
\eeq
Equally, Clebsch-Gordan coefficients are elements of $\Inv\left(\hilbert_j\otimes\hilbert_k\otimes\hilbert_l^{\ast}\right)$ where $\,^{\ast}$ denotes the dual. An index of a $3j$-symbol is dualized  by contraction with the (unique) two valent intertwiner $\epsilon\in\Inv(\hilbert^{\ast}_j\otimes\hilbert_j^{\ast})$
\beq
\epsilon_{j}^{\nu,\mu}:=(-1)_{j-\nu}\delta^{-\nu,\mu}
\eeq
and thus 
\beq
\iota_{j,\mu;\;k,\nu}^{\;\qquad l,\lambda}=(-1)^{l-\lambda}\threej{j}{\mu}{k}{\nu}{l}{-\lambda}\;.
\eeq
Since $\iota_{j,-\mu;\;k,-\nu;\;l,-\lambda}=(-1)^{j+k+l} \iota_{j,\mu;\;k,\nu;\;l,\lambda}$ and $\mu+\nu+\lambda=0$ the dual $3j$-symbol $\iota^{j,\mu;\;k,\nu;\;l,\lambda}$ is equivalent to $\iota_{j,\mu;\;k,\nu;\;l,\lambda}$. 

In contrast to three valent intertwiners the space of four valent intertwiners $\Inv[\bigotimes\limits_{i=1}^4\hilbert_{j_i}]$ is not one dimensional since the trivial representation occurs with multiplicity $N_0$ in
\beq
\bigotimes\limits_{i=1}^4\hilbert_{j_i}=
\left(\bigoplus_{a=|j_1-j_2|}^{j_1+j_2}\hilbert_a\right)
\otimes\left(\bigoplus_{a'=|j_3-j_4|}^{j_3+j_4}\hilbert_{a'}\right)~.
\eeq
The number $N_0$ is determined by the total number of irreps $a\in\{\max(|j_1-j_2|,|j_3-j_4|),\dots,\min(j_1+j_2,j_3+j_4)\}$. A normalized intertwiner of that kind is defined by
\beq
\label{eqn:fourv}
(\iota_a)_{j_1,\mu_1;j_2,\mu_2;j_3,\mu_3;j_4,\mu_4}\;=
d_a\; \sum_{\alpha}\;\threej{j_1}{\mu_1}{j_2}{\mu_2}{a}{\alpha}\threej{a}{\alpha}{j_3}{\mu_3}{j_4}{\mu_4}
\eeq
We could have also started by coupling for instance $j_1,j_3$ and $j_2,j_4$ by an intermediate irrep $b$ and would have arrived by the same result. The intertwiners $(\iota_a)_{j_1,j_2,j_3,j_4}$ and $(\iota_b)_{j_1,j_3,j_2,j_4}$ are related by a change of basis through $6j$ symbols
\begin{gather}
\begin{gathered}
\sum_{\alpha}\;\threej{j_1}{\mu_1}{j_2}{\mu_2}{a}{\alpha}\threej{a}{\alpha}{j_3}{\mu_3}{j_4}{\mu_4}\\
= \sum_b\sixj{j_1}{j_4}{j_2}{j_3}{b}{a}
\sum_{\beta}\;\threej{j_1}{\mu_1}{j_3}{\mu_3}{b}{\beta}\threej{b}{\beta}{j_2}{\mu_2}{j_4}{\mu_4}~.
\end{gathered}
\end{gather}
All higher valent intertwiners can be obtained in the same manner. For more details and an explicit graphical calculus see \cite{Alesci:2011ia}.

Intertwiners commute with the group action and thus
\beq
\label{eqn:orth2}
\begin{split}
&R^{j_1}_{\alpha\tilde{\alpha}}(g)\;R^{j_2}_{\beta\tilde{\beta}}(g)\\
	&=\sum_{j_3=|j_1-j_2|}^{j_1+j_2}d_{j_3}\sum_{\gamma,\tilde{\gamma}=-j_3}^{j_3}
	\threej{j_1}{\tilde{\alpha}}{j_2}{\tilde{\beta}}{j_3}{\tilde{\gamma}}
	\threej{j_1}{\alpha}{j_2}{\beta}{j_3}{\gamma}\overline{R^{j_3}_{\gamma\tilde{\gamma}}(g)}
\end{split}
\eeq
or in the index notation \eqref{eqn:int3j}
\beq
\label{eqn:recoupl}
[R^{j_1}(g)]^{\,\alpha}_{\;\;\tilde{\alpha}}\;[R^{j_2}(g)]^{\,\beta}_{\;\;\tilde{\beta}}=
(\iota^{\dagger})^{j_1,\alpha; j_2,\beta;j_3,\gamma}\; [R^{j_3}(g^{-1})]^{\tilde{\gamma}}_{\;\;\gamma} \;\;\iota_{j_1,\tilde{\alpha}; j_2,\tilde{\beta};j_3,\tilde{\gamma}}
\eeq
By first coupling $j_1$ and $j_2$ and then using \eqref{eqn:schurA}
\beq
\label{eqn:othogjint}
\int\dif{\mu_H(g)}\chi^{j_1}(g)\;\chi^{j_2}(g)\; \chi^{j_3}(g)=\tr (\iota^{\dagger}\iota)
\eeq
and
\beq
\label{eqn:othogjint4}
\int\dif{\mu_H(g)}\chi^{j_1}(g)\;\chi^{j_2}(g)\; \chi^{j_3}(g)\;\chi^{j_4}(g)=\sum_a\tr (\iota_a^{\dagger}\iota_a)~.
\eeq

\section{Some facts on piecewise-linear topology and triangulations}
\label{app:triangulation}
In this section some results on triangulation and piecewise-linear topology of 3- and 4-manifolds is presented. The exposition is mainly based on \cite{pl1}.  Since any cell-complex can be subdivided into a simplicial complex without introducing new vertices a cell-complex is assumed to be simplicial if not stated otherwise. 
\begin{definition}
A locally finite simplicial complex $K\subset\R^n$ is a collection of simplices such that 
\begin{enumerate}
\item $\sigma,\tau\in K\implies \sigma\cap\tau=\emptyset$ or it is a common face
\item $\sigma\in K$, $\tau$ a face of $\sigma$ then $\tau\in K$
\item $\forall\, x\in\overline{K}\;\; \exists \,U\in\R^n$ s.t. $U$ is an open neighborhood of $x$ meeting only finitely many simplices of $K$.
\end{enumerate}  
\end{definition}
As before $\overline{K}$ denotes the underlying polyhedron, i.e. the union of cells of $K$. A map $f:\overline{K} \to\overline{L}$ between polyhedra $\overline{K}$ and $\overline{L}$ is piecewise linear (p.l.) iff the graph $\gamma(f):=\{(x,f(x))|x\in K\}$ is a polyhedron. A p.l.-map is simplicial if the restriction of $f$ to any simplex $\sigma\in K$ is linear. Note, that a simplicial map is determined completely by its values on its vertices.

A p.l. m-ball is p.l. homeomorphic to an m-simplex in $\R^m$. If every point $x\in\overline{K}$ lies in the interior of a p.l. m-ball or $(m-1)$-ball then $\overline{K}$ is a \textit{p.l. manifold} of dimension $m$ with boundary $\partial\overline{K}$, which is the submanifold consisting of all points $x\in\overline{K}$ whose neighborhood in $\partial\overline{K}$ is homeomorphic to an $(m-1)$-ball.  
\begin{definition}
Let $K$ be a locally finite cell complex and $M$ a smooth manifold then $f:\overline{K}\to M$ is \underline{piecewise differentiable} (PD) if for every point $x\in\overline{K}$ one can find a closed neighborhood $U\subset\overline{K}$ and a subdivision $K'_x$ of $K$ such that $U\cap K'_x$ is a finite simplicial complex and the restriction of $f$ to each simplex of $K'_x\cap U$ is smooth. The map $f$ is a PD homeomorphism if $f$ is PD, a homeomorphism and the restriction of $f$ to each simplex has an injective differential at each point.
\end{definition}
A smooth triangulation of a smooth n-manifold is a triple $(M,K,f)$ where $M$ is a smooth manifold, $\overline{K}$ a p.l. n-manifold and $f:\overline{K}\to M$ a PD homeomorphism. 
\begin{theorem}[Whitehead]
Every smooth n-manifold $M$ has a triangulation $(M,K,f)$ which is unique up to PD homeomorphism. 
\end{theorem} 
Originally Whitehead worked in the $C^1$-category \cite{whitehead:1940} instead of smooth manifolds and PD maps. Yet in this case, $\overline{K}$ is not necessarily a p.l.-manifold and thus the triangulation is not unique e.g. $S^5$ allows triangulation that are not p.l. manifolds \cite{cannon1979}.

The above theorem can be proven by showing that any map $f:\overline{K}\to M$ of class $C^k$ can be approximated by a p.l. map. Lets assume for simplicity that $K$ is finite then for every $\epsilon,\rho>0$ one can find a simplicial subdivision $K'$ of $K$ and a simplicial map $L_f$ defined by the values $f(x_i)$ on the vertices $x_i$ of $K'$ such that 
\beq
\|L_f-f\|\leq\epsilon\quad\text{and}\quad \|\dif L_f-\dif f\|\leq\rho
\eeq
on every simplex of $K$. Furthermore, the subdivision of $K$ can be chosen fine enough such that $L_f$ is non-degenerate if $f$ is non-degenerate, i.e. the Jacobian matrix has full rank at each point of $f$. 

On the other hand every p.l. manifold of dimension less than seven has a unique differentiable structure, thus to every p.l. n-manifold $\overline{K}$ with $n<7$, corresponds a unique triangulation $(K,f.M)$ of a smooth manifold $M$ up to diffeomorphism (see \cite{Munkres1960,Smale1959, Hirsch1963}). In dimension lower than four even every topological manifold has a unique p.l. and differentiable structure \cite{Moise1977}.

\end{document}